\documentclass[letterpaper,onecolumn,twoside]{IEEEtran}
\usepackage{amsmath, amssymb, graphicx, psfrag, cite}

\newlength{\widthA}
\newcommand{\beq}{\begin{equation}}
\newcommand{\eeq}{\end{equation}}
\newcommand{\beqa}{\begin{eqnarray}}
\newcommand{\eeqa}{\end{eqnarray}}
\newcommand{\beqan}{\begin{eqnarray*}}
\newcommand{\eeqan}{\end{eqnarray*}}

\def\ts{\textstyle}
\def\del{\partial}

\def\half{{\ts\frac{1}{2}}}
\def\Half{{\frac{1}{2}}}
\def\var{\mathop{\mathrm{var}}}
\def\R{{\mathbb{R}}}
\def\N{{\mathbb{N}}}

\newcommand{\Lonethird}{\ensuremath{\mathcal{L}^{1/3}}}

\def\ghat{\widehat{g}}
\def\Xhat{\widehat{X}}
\def\Yhat{\widehat{Y}}

\def\I{\mathcal{I}}
\def\length{\mathop{\mathrm{length}}}
\def\max   {\mathop{\mathrm{max}}}
\def\median{\mathop{\mathrm{median}}}
\def\gPL{g_{\rm PL}}
\def\lord{\lambda_{\rm ord}}
\def\lfr{\lambda_{\rm fr}}
\def\lvr{\lambda_{\rm vr}}
\def\Dord{D_{\rm ord}}
\def\Dfr{D_{\rm fr}}
\def\Dvr{D_{\rm vr}}
\def\Dsw{D_{\rm sw}}

\def\dhr{d^{\rm HR}}
\def\Kfr{K_{\rm fr}}
\def\Kvr{K_{\rm vr}}

\def\Kfrhr{K_{\rm fr}^{\rm HR}}
\def\Kvrhr{K_{\rm vr}^{\rm HR}}

\def\Khr{K^{\rm HR}}

\def\lB{{\boldsymbol{\lambda}}}
\def\wB{{\bf w}}
\def\KB{{\bf K}}
\def\RB{{\bf R}}
\def\Dfrhr{D_{\rm fr}^{\rm HR}}
\def\Dvrhr{D_{\rm vr}^{\rm HR}}
\def\Dswhr{D_{\rm sw}^{\rm HR}}
\newcommand{\Frac}[2]{{{#1}/{#2}}}  

\def\RS{\bar{R}}
\def\wfr{w_{\rm fr}}
\def\argmin{\mathop{\mathrm{arg \, min}}}

\renewcommand{\P}[1]{\textbf{P}\left({#1}\right)} % Probability
\newcommand{\E}[1]{\textbf{E}\left[{#1}\right]}   % Expectation
\newcommand{\iE}[1]{\textbf{E}[{#1}]}  % i for inline

\newcommand{\eqlabel}[1]{ \stackrel{(#1)}{=} }
\newcommand{\geqlabel}[1]{ \stackrel{(#1)}{\geq} }
\newcommand{\leqlabel}[1]{ \stackrel{(#1)}{\leq} }

\newtheorem{theorem}{Theorem}
\newtheorem{lemma}[theorem]{Lemma}

\newtheorem{corollary}[theorem]{Corollary}
\newtheorem{definition}{Definition}
\newtheorem{example}{Example}

\newcommand{\qed}{\nobreak \ifvmode \relax \else
      \ifdim\lastskip<1.5em \hskip-\lastskip
      \hskip1.5em plus0em minus0.5em \fi \nobreak
      \vrule height0.75em width0.5em depth0.25em\fi}
      
\title{Distributed Scalar Quantization for Computing:
  High-Resolution Analysis and Extensions}
\author{Vinith~Misra,
        Vivek~K~Goyal,~\IEEEmembership{Senior~Member,~IEEE}, and
        Lav~R.~Varshney,~\IEEEmembership{Member,~IEEE}%
\thanks{This material is based upon work supported by the National Science Foundation
under Grant No. 0729069.}%
\thanks{The material in this paper was presented in part at
        the Information Theory and its Applications Workshop,
        La Jolla, California, January/February 2008; and
        the IEEE Data Compression Conference, Snowbird, Utah, March 2008.}%
\thanks{V. Misra (email: vinith@stanford.edu) was with the Massachusetts
        Institute of Technology when this work was completed and is now
        with the Department of Electrical Engineering, Stanford University,
        Stanford, CA 94305 USA\@.  V.~K. Goyal (email: vgoyal@mit.edu),
        and L.~R. Varshney (email: lrv@mit.edu) are with
        the Department of Electrical Engineering and Computer Science and
        the Research Laboratory of Electronics,
        Massachusetts Institute of Technology,
        Cambridge, MA 02139 USA\@.
        L.~R. Varshney is also with the Laboratory for Information and Decision
        Systems.}}

\begin{document}
\maketitle

\begin{abstract}

Communication of quantized information is frequently followed by a computation.
We consider situations of \emph{distributed functional scalar quantization}: 
distributed scalar quantization of (possibly correlated) sources followed by
centralized computation of a function.
Under smoothness conditions on the sources and function,
companding scalar quantizer designs are developed to
minimize mean-squared error (MSE) of the computed function as the quantizer resolution
is allowed to grow.  Striking improvements over
quantizers designed without consideration of the function are possible and
are larger in the entropy-constrained setting than in the fixed-rate setting.
As extensions to the basic analysis,
we characterize a large class of functions for which regular quantization
suffices, consider certain functions for which asymptotic optimality is
achieved without arbitrarily fine quantization, and allow limited
collaboration between source encoders.
In the entropy-constrained setting, a single bit per sample communicated between
encoders can have an arbitrarily-large effect on functional distortion.
In contrast, such communication has very little effect in the
fixed-rate setting.

\end{abstract}

\section{Introduction}
\IEEEPARstart{C}{onsider} a collection of $n$ spatially-separated sensors,
each measuring a scalar $X_j$, $j=1,\,2,\,\ldots,\,n$.
As shown in Fig.~\ref{fig:dfsc-block}, the measurements are encoded
and communicated over rate-limited links to a sink node without any
interaction between the sensors.
The sink node computes an estimate of
the function $g(X_1^n) = g(X_1,\,X_2,\,\ldots,\,X_n)$ from the received data.
This may be interpreted as a special case of the \emph{distributed source coding}
problem in which distortion is measured as the mean-squared error of the
 function estimate.  We refer to this special case as
 \emph{distributed functional source coding} to emphasize that it is
 the function $g(X_1^n)$ and not the source vector $X_1^n$ that is being reconstructed.
Similarly, we will refer to approximate representation of $X_1^n$ 
under mean-squared error distortion as \emph{ordinary} source coding.
Restricting to scalar quantization,
this \emph{distributed functional scalar quantization} (DFSQ) problem
is the central subject of this paper.
Compared to ordinary source coding, DFSQ can provide performance 
improvements in addition to any that
are rooted in statistical dependence of the $X_j$s; for clarity,
most examples presented here are for cases with independent $X_j$s.

%Functional source coding is a trivial problem when the encoding is centralized;
%in that case the encoder mapping can be the composition of the
%function $g$ and a good encoder for the random variable $g(X_1^n)$.
%With the constraint of distributed encoding, no single encoder can
%compute the function, and the situation is thus more intricate.
%
\subsection{Summary of Main Contributions}
The primary aim of this paper is to develop a high resolution approach
to the analysis of DFSQ.
To this end, we consider for each source variable $X_j$ a sequence of companding quantizers
$\{Q^{j}_K\}$ of increasing resolution $K$.
Under fairly loose smoothness requirements on the function $g(x_1^n)$
and the source probability density function (pdf) $f(x_1^n)$, high-resolution analysis yields a choice for $\{(Q^{1}_K,\ldots,Q^{n}_K)\}_{K=1}^\infty$ that
outperforms any other choice of companding quantizer sequences at sufficiently high resolution.
This analysis also
gives an approximation for the resulting distortion-rate function that has relative error which vanishes as $K \rightarrow \infty$.

There are situations in which designing quantizers
to minimize the MSE of the function estimate is no different than designing them for low MSEs
$\iE{(X_j - \Xhat_j)^2}$, $j=1,\,2,\,\ldots,\,n$.
Our analysis will show, for example, that there is little advantage from
accounting for $g$ when $g$ is linear.
However, there are also cases in which the improvement is very large
for large values of $n$;
examples in Section~\ref{sec:Scaling} feature distortion improvement
over ordinary source coding by a factor that is
polynomial in $n$ in the fixed-rate case and exponential in $n$
in the variable-rate case.

In addition to developing a basic theory in which there are no interactions
between quantizers and certain limitations on $g$ simplify our analysis,
we consider several extensions.
First, we permit nonregular quantizers
and demonstrate that if the function $g(x_1^n)$ satisfies a loose
\emph{equivalence-free} condition
then optimal quantizers are regular at sufficiently high rate.
Next, we explore a situation in which the high-resolution analysis
breaks down because there is an interval where the marginal density $f_{X_j}$
is positive but
the optimal companding quantizer sequence for $X_j$ is not arbitrarily fine.
This prompts the concept
of a \emph{don't care interval}, a mixture of low- and high-resolution,
and connections with~\cite{DoshiSM2007}.
Finally, we allow rate-constrained information
communicated
from encoder 2 to encoder 1 to affect the encoding of $X_1$.
We call this \emph{chatting} and bound its effect on the distortion $D$.
In the fixed-rate setting, the reduction in distortion
can be no more than if $R_1$ were increased by the same rate; %1 bit;
in the variable-rate setting, the reduction in distortion can be
arbitrarily large.
%companding scalar quantizers for DFSQ that outperform
%all other companding scalar quantizers at a sufficiently high rate.
%High-resolution analysis is employed to this purpose
% to obtain an approximation for their distortion-rate function,
% and to select quantizers.
% Through high-resolution analysis,
%that are asymptotically optimal in the sense that the distortion-rate function is
%lower than any other quantizer
%to analyze their performance via high-resolution techniques, and to optimize their
%design.  
%Here as in ordinary source coding,
%the high-resolution approach yields optimality among regular quantizers.
%In ordinary source coding, 
%this is an insignificant limitation because, quite generally,
%optimal quantizers are regular.
%For DFSQ, some restrictions on $g$ are needed to ensure that
%the optimal quantizers are regular.
%This provides another key contrast to previous work.
%Using only the graph coloring approach to FSC of Doshi \emph{et al.}~\cite{DoshiSM2007}
%provides no improvement under these restrictions on $g$,
%so the present work is a complement to~\cite{DoshiSM2007}.
%Combining the two approaches is discussed in Section~\ref{sec:DontCare}.

\begin{figure}
 \centering
  \psfrag{X1}[r][][1][0]{\small $X_1$}
  \psfrag{X2}[r][][1][0]{\small $X_2$}
  \psfrag{X3}[r][][1][0]{\small $X_n$}
  \psfrag{Bullets}{\vdots}
  \psfrag{Y1}[][][1][0]{\raisebox{7mm}
                  {\small \begin{tabular}{c} $\Xhat_1$ \\ rate $R_1$ \end{tabular}}}
  \psfrag{Y2}[][][1][0]{\raisebox{7mm}
                  {\small \begin{tabular}{c} $\Xhat_2$ \\ rate $R_2$ \end{tabular}}}
  \psfrag{Y3}[][][1][0]{\raisebox{7mm}
                  {\small \begin{tabular}{c} $\Xhat_n$ \\ rate $R_n$ \end{tabular}}}  
  \psfrag{G}[l][][1][0]{\small $\ghat(\Xhat_1^n)$}
  \includegraphics[width=2.25in]{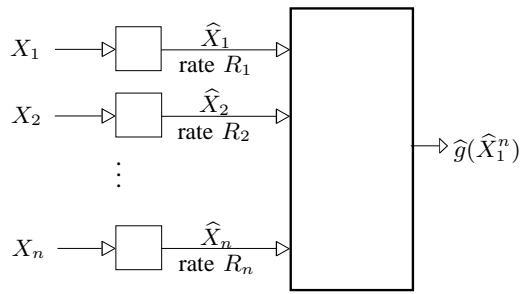}
  \caption{Distributed functional source coding.}
 \label{fig:dfsc-block}
\end{figure}

%\subsection{Basic Problem Statement}
%\label{sec:ProblemStatement}
%An information sink wishes to obtain an estimate of $g(X_1^n)$ 
%where $g: \R^n \rightarrow \R$
%satisfies some smoothness conditions
%and the random variables $\{X_j\}_{j=1}^n$
%(denoted more compactly $X_1^n$)
%have some known joint distribution.
%The estimate $\ghat(\Xhat_1^n)$ is computed from scalar-quantized values
%$$
%  \Xhat_j = Q_j(X_j), \qquad j=1,\,2,\,\ldots,n,
%$$
%where $Q_j$ applied to $X_j$ has rate $R_j$.
%In the fixed-rate (codebook-constrained) setting,
%this means $Q_j$ has $K_j = 2^{R_j}$ levels;
%in the variable-rate (entropy-constrained) setting,
%this means $H(Q_j(X_j)) = R_j$
%where $H(\cdot)$ denotes the entropy.
%
%
%The accuracy of the approximation is measured by the mean-squared error (MSE)
%$$
%D = \dE{ \left( g(X_1^n) - \ghat(\Xhat_1^n) \right)^2 } \mbox{.}
%$$
%For a given set of rates $\{R_j\}_{j=1}^n$
%or a maximum sum rate,
%we seek designs of the quantizers $\{Q_j\}_{j=1}^n$
%such that distortion $D$ is minimized.
%The problem is approached under the standard assumptions for
%high-resolution analysis~\cite{GrayN1998},
%and restrictions on $g$ are applied as needed.

For ordinary quantization problems,
high-resolution analysis is not interesting for a discrete source because
the distortion reaches zero at some finite resolution.
Indeed, as in most works using high-resolution analysis,
we assume that the source random variables are jointly continuous,
i.e., that a joint probability density function for $X_1^n$ exists.
Similarly, high-resolution analysis of DFSQ may be uninteresting when
$g(X_1^n)$ is discrete because zero functional distortion may be achieved at
some finite resolution.
We do not explicitly require $g(X_1^n)$ to be a continuous random variable,
but the continuity of $g$ that we do require eliminates many situations
in which zero functional distortion may be achieved at some finite resolution.

\subsection{Related Work}
DFSQ has strong connections to several problems that have been studied in 
prior work on quantization and distributed source coding.  We provide a brief summary
of some of these connections here.  This paper is restricted to high-resolution analysis of companding 
scalar quantizers for real-valued sources.  Contrarily, some related works deal with 
lossless source coding or lossy vector quantization, often in the (Shannon-theoretic) limit of large 
block length, at any rate.

Consider the situation depicted in Fig.~\ref{fig:dfsc-block} with $n=2$.
In general, $X_1$ and $X_2$ are memoryless, stationary random processes
and $g$ is a function of the two.  Several topics
arise by considering special cases of this formulation.

When $g$ is the identity function, the goal is to reconstruct the 
	source variables themselves; often the correlation between $X_1$ and $X_2$ is of primary
      interest.  Slepian and Wolf solve this problem in the infinite blocklength regime for lossless 
	representation of sources drawn from a discrete alphabet~\cite{SlepianW1973}.  The lossy problem
      for sources from a discrete alphabet, restricted to scalar quantization followed by block entropy coding,
	is considered in~\cite{Servetto2005}.  

In the setting with lossy representation of continuous sources, one might consider applying Slepian--Wolf 
	coding to the output of local quantizers for each of the sources.  This approach, with vector quantization 
	performed on blocks of each of the sources, is optimal at all rates for jointly Gaussian sources and MSE 
	distortion~\cite{WagnerTV2008}.  This approach is also optimal in the asymptotic regime of both 
	large block length and high resolution \cite{ZamirB1999}.  The general lossy multiterminal source coding
	problem for large block length but finite rates, whether for discrete or continuous alphabet sources, is open.
    
While this paper restricts to scalar quantization of the sources, the
      use of Slepian--Wolf coding on the output of these quantizers is considered (Sec.~\ref{sec:slepian-wolf}). 
      Note that since the identity function has a vector output, our DFSQ formulation
      technically does not permit this choice of $g$, but that only minor modification of the proofs are required 
	to permit vector-valued functions.

If $g(X_1,X_2) = X_1$ and $R_2$ is unconstrained, then $X_2$ can
      be viewed as receiver side information available at the decoder.
      The trade-off between $R_1$ and distortion (of $X_1$ alone) in the large block length regime is given
      by the Wyner-Ziv rate-distortion function~\cite{WynerZ1976,Zamir1996}.
      Rebollo-Monedero \emph{et al.} examined this scenario at
      high resolution but any block length, and showed that providing receiver
      side information to the encoder yields no improvement in
      performance~\cite{Rebollo-MonederoRAG2006}, cf.~\cite{LiuCLX2006}.
	Under suitable constraints on the distortion metric,
      one may also view $X_2$ as receiver side information that
      determines the distortion measure on $X_1$,
      drawing a connection to~\cite{MartinianWZ2004} and to work on non-MSE distortion functions \cite{LinderZZ2000}.
   
For general $g$ and unconstrained $R_2$, the lossy problem has been
      studied by Yamamoto~\cite{Yamamoto1982} and later by Feng \emph{et al.}~\cite{FengES2004},
      who provide an assortment of rate-loss bounds on performance in the large block length setting.
The lossless setting has been explored
       by Orlitsky and Roche~\cite{OrlitskyR2001}.  
%\LRV{Reviewer says: ``p.3, bullet-6: The work of Orlitsky and Roche does not fit the
%    situation depicted in Fig.~1 as it allows messages to go from
%    terminal 1 to 2 and also back.'' In response letter, say that Orlitsky-Roche consider two problems.
%	The first fits in to our schema.  Second is interactive, and as reviewer says, it does not.  We were referring to first.}

In the large block length regime for lossless coding, Han and Kobayashi~\cite{HanK1987} studied
      the classification of functions according to whether the rate region
      is the same as that for the identity function
      (i.e., the same as the Slepian--Wolf rate region).
      Their results are conclusive when $n=2$ and the source alphabets
      are finite.  This distributed version of the problem for general $g$,
      minimizing the sum-rate $R_1+R_2$,
      was later investigated by Doshi \emph{et al.}~\cite{DoshiSMJ2007}.

Let $Y = g(X_1,X_2)$. Then $Y$ may be interpreted as a
      \emph{remote source} that is observed only through
      $X_1$ and $X_2$, leading to a remote source multiterminal source 
      coding problem~\cite{YamamotoI1980}.  Alternatively, $\{Y=X_0,X_1,X_2\}$,
	can be thought of as a source triple and the problem in Fig.~\ref{fig:dfsc-block} 
	as a two-help-one problem with $R_0 = 0$ \cite{KornerM1979}.

Most of the above examples involve block coding of $X_1$ and $X_2$,
and results are obtained by allowing the block length to grow arbitrarily large.
While the variable-length DFSQ analysis does utilize block entropy coding
and Slepian--Wolf coding, $X_1$ and $X_2$ must first pass through
scalar quantizers.  Even though the samples of $X_1$ and $X_2$ are i.i.d.,
there would still be geometric benefits to using vector quantization over blocks of 
samples; this is left to future work.

Quantization with a functional motive bears strong resemblance to the
idea of ``task-oriented quantization.''  There has been considerable work in
this direction for detection, classification, and estimation, including high-rate
analysis \cite{Poor1988,BenitzB1988,GuptaH2003}.
%~\cite{Kassam1977,PicinbonoD1988},
%classification~\cite{XieO2002}, and estimation~\cite{VasudevanOM2003,Gubner1993};
%see also the review article~\cite{HanA1998}.
%\LRV{I omitted references to work that does not use high-rate theory.  The reasoning for this
%was to avoid a Pandora effect.  For example, the reviewer requested Gubner, but if one puts that in,
%one seemingly needs to put in huge numbers of other things.  I am open to suggestions on how to handle this.}
The use of a function at the decoder can be seen as
inducing a non-MSE distortion measure on the source data.
In this sense, a thread may be drawn to perceptual
source coding~\cite{LiCG1999}, where a non-MSE distortion
reflects human sensitivity to audio or video.

Under appropriate constraints on the function $g$, one may
consider it as having introduced a \emph{locally quadratic} distortion
measure on the source $X_1^n$.  In~\cite{LinderZZ1999}, Linder \emph{et al.}\
consider quantization via companding functions for
locally quadratic distortion measures. 
We say more about connections to this work in
Section~\ref{sec:locally-quadratic}.

Interesting related problems have also arisen without
a requirement of distributed coding.
Rather than having a single function $g$, one may consider a set
of functions $\{g_a\}_{a \in \mathcal{A}}$ and define
\[
D_g = \E{d(g_\alpha(X_1^n),g_\alpha(\Xhat_1^n))}\mbox{,}
\]
where $\alpha$ is a random variable taking values in index set $\mathcal{A}$.
One may consider this a special case of the Wyner-Ziv problem with
$\alpha$ as decoder side information and a functional distortion measure.
In such a setting,
fixed- and variable-rate quantization to minimize MSE was
studied by Bucklew in the high-rate regime~\cite{Bucklew1984}.  Note that if the function
were known deterministically to the encoder, one could do no better
than to simply compute the function and encode the result.

\subsection{Structure of Paper}
We start in Section~\ref{sec:Background} by reviewing the high-resolution approximation
techniques used in our analysis.
In Section~\ref{sec:Single} we obtain optimal fixed- and variable-rate functional quantizers for the $n = 1$ case;
while not important in practice, this case illustrates the role of monotonicity
and smoothness of $g$.
Generalizations to arbitrary $n$, under similar restrictions on $g(\cdot)$,
are given in Section~\ref{sec:Multi}.
Some notable examples in Section~\ref{sec:Scaling}
are those that show dramatic scaling of distortion with respect to $n$.
Some arguments in Sections~\ref{sec:Background} and~\ref{sec:Single}
are meant only to build intuition; the technical results of those
sections are rigorously justified as special cases of statements
in Section~\ref{sec:Multi}.

The second half of the paper extends the basic theory of Section~\ref{sec:Multi}.
Section~\ref{sec:NonMono} addresses the use of non-regular companding quantizers and shows
that a weak \emph{equivalence-free} condition guarantees regularity
of the optimal companding quantizer sequence.
In the process we develop the notion of high-resolution non-regular quantization.
In Section~\ref{sec:DontCare}, we consider certain conditions that cause
the high-resolution approach to lead to an optimal quantizer for $X_j$
that does not have high resolution over the entire support of $f_{X_j}$.
A modified analysis and design procedure yields a ``rate amplification''
in the variable-rate case.
Limited communication between encoders, or chatting, is studied in Section~\ref{sec:Chatting},
and concluding comments appear in Section~\ref{sec:Conclusion}.

\section{Univariate Ordinary Quantization}
\label{sec:Background}
To introduce both notation and techniques,
the high-resolution analysis of scalar quantizers under
MSE distortion is reviewed in this section.

\subsection{Definitions}
\label{sec:Definitions}
A $K$-level quantizer on $[0,1]$ is a function $Q_K: [0,1] \rightarrow [0,1]$
with a range consisting of $K$ points.
The expected distortion of $Q_K$ applied to random variable $X$ taking values
in $[0,1]$ is given by $D(Q_K) = \E{d(X,Q(X))}$,
where $d:[0,1]\times[0,1]\rightarrow [0,\infty)$ is an appropriately chosen
distortion function.  Squared-error distortion $d(x,y) = (x-y)^2$ is both a frequent
and analytically-tractable choice.
In \emph{fixed-rate} (or codebook-constrained) quantization,
the rate is defined as the logarithm of the number of levels, $R = \log K$,
where all logarithms have base 2\@.
In \emph{variable-rate} (or entropy-constrained) quantization,
the rate is defined as the entropy of the quantizer output, $R = H(Q_K(X))$.
An optimal fixed-rate or variable-rate quantizer minimizes distortion subject
to a constraint on the applicable rate.

A value in the range of $Q_K$ is called a \emph{quantizer point}
or \emph{reconstruction point},
and the inverse image under $Q_K$ of a quantizer point is called a \emph{cell}
or \emph{partition region}.
If each cell is an interval and the associated reconstruction point lies within
the interval, the quantizer is called \emph{regular}.
For a distortion function that increases with the difference of its arguments
(e.g. squared-error distortion), the
optimal fixed-rate quantizer is regular.  If the distortion function is also
convex in the difference of its arguments and the source distribution is non-atomic, the optimal
variable-rate quantizer is regular as well~\cite[Sect.~6.2]{GershoG1992}~\cite{GyorgyL2002}.

A \emph{compander} function $w: [0,1]\rightarrow[0,1]$ is
continuous, increasing, differentiable almost everywhere, and invertible on $[0,1]$.  
Furthermore, $w(0) = 0$ and $w(1) = 1$.
The $K$-level uniform quantizer on $[0,1]$ is defined as
\[
Q_{K}^U(x) = \left\{ \begin{array}{rl} \frac{2i-1}{2K}, & \mbox{for $x \in \left( \frac{i-1}{K},\frac{i}{K} \right]$, $i=1,\,2,\,\ldots,\,K$}; \\[0.2ex]
\frac{1}{2K}, & \mbox{for $x = 0$}.
\end{array} \right.
\]
For squared-error distortion and more generally,
optimal quantizers satisfy a stronger condition than regularity:
\[
  x < y
\qquad
\mbox{implies}
\qquad
  Q_K(x) \leq Q_K(y).
\]
They can thus be realized in companding form:
\[
Q_{K}(x) = w^{-1}(Q_{K}^U(w(x))) \mbox{.}
\] 

A quantizer that has a companding form 
may equivalently be defined by its \emph{point density function}
$\lambda(x)$:
\[
\lambda(x) = w'(x) \mbox{,}
\]
which always satisfies $\int_0^1 \lambda(x) \, dx = w(1) - w(0) = 1 - 0 = 1$
by the fundamental theorem of calculus.
For small $\delta$ and large resolution $K$, one may observe that $\delta \, \lambda(x)$ approximates
the fraction of quantizer points in an interval of length $\delta$ around $x$.
Because of this intuitive relationship to quantizer structure, we will use
the point density description instead of the compander description
whenever possible, with $Q^{\lambda}_K(x)$ denoting a quantizer of
resolution $K$ and point density function $\lambda$.
A \emph{companding quantizer sequence} $\{Q_K^{\lambda}\}_{K=1}^{\infty}$ refers to a sequence of quantizers
generated with the same point density $\lambda$ and indexed by resolution $K$.
Our interest will be in optimizing these quantizer sequences.

The \emph{distortion-resolution function} $d(K;\lambda)$ for a companding quantizer sequence
$\{Q_K^{\lambda}\}$ indexes the distortion of the sequence by the resolution $K$:
\[
d(K;\lambda) = \E{\left| X - Q_{\lambda}^K(X) \right|^2 } \mbox{.}
\]
The fixed-rate \emph{resolution-rate} function $\Kfr(R;\lambda) = \lfloor 2^R \rfloor$ is
the largest resolution that satisfies a fixed-rate constraint.  Similarly, the variable-rate resolution-rate
function $\Kvr(R;\lambda)$ is the largest resolution that satisfies a variable-rate constraint.
Specifically, $\Kvr(R;\lambda)$ is the largest resolution such that the entropy of the quantized
output $H(Q_{\Kvr}^{\lambda}(X))$ is less than the rate constraint $R$:
\[
K(\lambda;R) =\max_{H(Q_K^{\lambda}(X)) \leq R} K \mbox{.}
\]

The quality of a quantizer sequence $\{Q_K^{\lambda}\}$ is measured by its
distortion-rate function.  The fixed-rate distortion-rate function measures
the distortion of the highest-resolution element of the sequence that satisfies the
fixed-rate constraint: $\Dfr(R;\lambda) = d(\Kfr(R;\lambda),\lambda)$.  Similarly,
the variable-rate distortion-rate function measures the distortion of the highest-resolution
element of the sequence that satisfies the variable-rate constraint: 
$\Dvr(R;\lambda)= d(\Kvr(R;\lambda),\lambda)$.

Under a fixed-rate constraint, we say that a companding quantizer sequence $\{Q_K^{\lambda^*}\}$ 
is \emph{asymptotically better} than another $\{Q_K^{\lambda}\}$ if 
\[ 
\limsup_{R\rightarrow \infty} \frac{\Dfr(R;\lambda^*)}{\Dfr(R;\lambda)} \leq 1 \mbox{.}
\]
Essentially, we compare the best rate-$R$ quantizers from each sequence.
If $\{Q_K^{\lambda^*}\}$ is asymptotically better than all other
quantizer sequences, we say $\{Q_K^{\lambda^*}\}$ and $\lambda^*$ are asymptotically \emph{fixed-rate optimal}.

Analogously, an asymptotically \emph{variable-rate optimal} quantizer sequence $\{Q_K^{\lambda^*}\}$
is asymptotically better than any other $\{Q_K^{\lambda}\}$:
\[
\limsup_{R\rightarrow \infty} \frac{\Dvr(R;\lambda^*)}{\Dvr(R;\lambda)} \leq 1 \mbox{.}
\]

Note that while we only consider optimality among the set of regular companding quantizer sequences, 
Linder \cite{Linder1991} provided conditions for a source 
probability distribution function under which a companding quantizer sequence
can be optimal in a more general sense.

\subsection{Problem Statement}
\label{sec:BackgroundProblemStatement}
A sequence of quantizers is to be applied
to a source $X$ with pdf $f_X$ supported on the interval
$[0,1]$.  The distortion of the quantizers is measured by squared error.
For any fixed- or variable-rate constraint, the optimal quantizer can be
realized in companding form, so
we seek an asymptotically optimal companding function.

%For any variable-rate or fixed-rate constraint, the optimal quantizer
%is regular.  For a sequence of companding quantizers to be regular,
%the companding function must be monotonic.  Furthermore, any sequence of 
%quantizers generated by a nonincreasing companding function can be equivalently
%generated by a nondecreasing companding function.  Therefore, attention may be
%constrained to nondecreasing companding functions.  

%Furthermore,
For high-resolution techniques to be valid, both the companding function
and the source pdf must satisfy certain smoothness requirements.
We assume the source satisfies conditions UO1 and UO2, and we optimize only among
 companding functions that satisfy UO3 and UO4:
%All assumptions are summarized below.

\begin{description}
\item[UO1.] The source pdf $f$ is bounded and supported on the interval $[0,1]$.
\item[UO2.] The first derivative of the source pdf $f'$ is defined and bounded on all
but a finite number of points in $[0,1]$.
\item[UO3.] We optimize among companding functions that are
differentiable.
\item[UO4.] The integral $\int_0^1 f(x)w'(x)^{-2} \, dx$ is finite.
\end{description}

%g' must be defined and bounded on all but a finite number of points.
%f must be piecewise continuous, bounded, and its derivative must be defined and bounded on all but a finite number of points.

\subsection{Solution via High-Resolution Analysis}
\label{sec:HighResAnalysis}
%point density function = derivative of companding function.
%precise assumptions
%distortion in terms of point density
%rate in terms of point density.
%derivation of ``optimal'' point densitiies
%make note of approximations and their asymptotic validity.
%asymptotically-optimal bit allocations
The quantities of fundamental interest in the analysis of companding quantizer
sequences are the fixed- and variable-rate distortion-rate functions $\Dfr(R;\lambda)$ and
$\Dvr(R;\lambda)$, which describe the distortion of fixed- and variable-rate companding 
quantizers with rate $R$ and point density $\lambda$.
%These quantities may be expressed through the
%distortion-resolution function $d(K;\lambda)$ and the 
%resolution-rate function $K(\lambda;R)$.
High resolution analysis consists of several approximations
that allow one to derive asymptotically accurate versions of both
$\Dfrhr(R;\lambda)$ and $\Dvrhr(R;\lambda)$.  Specifically, under appropriate
restrictions on the source pdf we will show that
\beq
\lim_{R\rightarrow \infty} \frac{\Dfrhr(R;\lambda)}{\Dfr(R;\lambda)} =
\lim_{R\rightarrow \infty} \frac{\Dvrhr(R;\lambda)}{\Dvr(R;\lambda)} = 
1 
\mbox{.} \label{eq:asymptoticaccuracy}
\eeq

In Sec. \ref{sec:BackgroundHighResolutionDistortion}, the
approximate distortion-resolution function $\dhr(K; \lambda)$ is derived.
Then, in Sec. \ref{sec:BackgroundHighResolutionRate}, the
approximate resolution-rate function $\Khr(R; \lambda)$ is obtained
for both fixed- and variable-rate constraints.  Finally, in 
Sec. \ref{sec:BackgroundHighDistortionRate} these two quantities yield
the approximate distortion-rate functions $\Dfrhr(R;\lambda)$ and $\Dvrhr(R;\lambda)$.
The derivation we provide is left informal and
is not intended to prove that assumptions UO1--UO4 yield \eqref{eq:asymptoticaccuracy};
this follows either from Linder \cite{Linder1991} or as a special 
case of Theorem \ref{thm:single-distortion} in
Sec. \ref{sec:Single}.   For further technical details and references
to original sources, see~\cite{GrayN1998}.  Finally, in Sec. \ref{sec:review-optimal}, the approximate distortion-rate
functions are optimized through choice of point density (companding function).
The sequences of companding quantizers yielded by this optimization are shown to be
asymptotically fixed- or variable-rate optimal.

\subsubsection{The Distortion-Resolution Function}
\label{sec:BackgroundHighResolutionDistortion}
As previously defined, $d(K;\lambda)$ is the distortion of the companding quantizer with
resolution $K$.  We now define
an approximation $\dhr(K;\lambda)$, known as the approximate distortion-resolution function.
For rigorous proof that 
\beq
\lim_{K\rightarrow\infty} \dhr(K;\lambda)/d(K;\lambda) = 1 \mbox{,} \label{eq:distortionResolutionAsymptotic}
\eeq
we refer
to the main result of Linder \cite{Linder1991}, or to Theorem 
\ref{thm:multi-distortion} with $g(x) = x$.

Let $X$ be a random variable with pdf $f_X(x)$, 
let $Q^{\lambda}_K$ be a $K$-point companding quantizer, and suppose $\lambda$ and $f$
satisfy assumptions UO1--UO4\@.  Let $\{\beta_i\}_{i\in\I} = Q^{\lambda}_K([0,1])$ be the reconstruction
points, and let $S_i = \left(Q_K^{\lambda}\right)^{-1}(\beta_i)$, $i\in\I$, be the corresponding 
partition regions.
%For optimality, it is necessary for each set in the partition to be an interval,
%i.e., the quantizer is \emph{regular}~\cite[Sect.~6.2]{GershoG1992}.

The distortion of the quantizer is
\begin{eqnarray}
  \label{eq:general-quant-dist}
  d(K;\lambda) & = & \E{(X-\Xhat)^2} \nonumber \\
      & = & \sum_{i \in \I}
          \E{(X - \beta_i)^2 \mid X \in S_i}
          \P{X \in S_i}
\end{eqnarray}
by the law of total expectation.
The initial aim of high-resolution theory is to express this distortion
as an integral involving $f_X$.
To that end, we make the following approximations about the source
and quantizer:
\begin{enumerate}
\item[HR1.] $f_X$ may be approximated as constant on each $S_i$.  
\item[HR2.] The size of the cell containing $x$ is approximated with
the help of the point density function:
\beq
  \label{eq:cell-length}
  x \in S_i \quad \Rightarrow \quad \length(S_i) \sim (K\lambda(x))^{-1} \mbox{,}
\eeq
where $\sim$ means that the ratio of the two quantities goes to 1 with increasing
resolution $K$.  This is the meaning of ``$\sim$'' for the remainder of the paper.
\end{enumerate}
The first approximation follows from the smoothness of $f_X$ (assumptions UO1 and UO2),
while the second follows from the smoothness of $w(x)$ (assumption UO3).

Now we can approximate each non-boundary term in \eqref{eq:general-quant-dist}.
By HR1, $\beta_i$ should be approximately at the center of $S_i$,
and the length of $S_i$ then makes the conditional expectation
approximately $\frac{1}{12}(K\lambda(\beta_i))^{-2}$.
Invoking Assumption HR1 again,
the $i$th term in the sum is $\int_{x \in S_i} \frac{1}{12}(K\lambda(\beta_i))^{-2} f_X(x) \, dx$.
Finally,
\begin{eqnarray}
\label{eq:unoptimizedHrDist}
d(K;\lambda) \sim \int_0^1 \frac{(K\lambda(x))^{-2}}{12}  f_X(x) \, dx
  & = & \frac{1}{12K^2} \E{\lambda^{-2}(X)} \\
  & = & \dhr(K;\lambda) \mbox{.} \nonumber
\end{eqnarray}
%This approximation holds in the sense that the ratio of the two
%quantities approaches 1 as the rate increases.

\subsubsection{The Resolution-Rate Function}
\label{sec:BackgroundHighResolutionRate}
For a fixed-rate quantizer, the resolution-rate relationship
is given simply by $\Kfr(R;\lambda) = \lfloor 2^R \rfloor$, and
it is approximated with vanishing relative error by $\Kfrhr(R;\lambda) = 2^R$.  The variable-rate resolution-rate
function is more difficult to approximate.

As long as the quantization is fine ($\lambda(x) > 0$) wherever the
density is positive,
we can approximate the output entropy of a quantizer using the
point density.  Defining $p(x)$  as $\P{X \in S_i}$ for $x \in S_i$,
and letting $h(X)$ denote the differential entropy of $X$,
\begin{eqnarray}
H(Q^{\lambda}_K(X)) & = & - \sum_{i\in\I} \P{X \in S_i} \log \P{X \in S_i} \nonumber \\
   & \eqlabel{a} & - \int_0^1 f_X(x) \log p(x) \, dx \nonumber \\
   & \stackrel{(b)}{\sim} & - \int_0^1 f_X(x) \log( f_X(x)/(K\lambda(x)) ) \, dx \nonumber \\
   & = & - \int_0^1 f_X(x) \log f_X(x) \, dx \nonumber \\
   &   & \quad + \: \int_0^1 f_X(x) \log(K\lambda(x)) \, dx \nonumber \\
   & = & h(X) + \log K + \E{\log \lambda(X)} \mbox{,}
 \label{eq:1drate}
\end{eqnarray}
where (a) follows from the definition of $p(x)$; and
(b) involves approximating the source pdf as constant
in each cell and \eqref{eq:cell-length}.

A generalized version of this approximation is proven rigorously
in \cite{LinderZZ1999}.  We state it here as a lemma.
\begin{lemma}
\label{lem:ResolutionRate}
Suppose the source $X$ has a density over $[0,1]$ and a finite differential
entropy $h(X)$. 
% Additionally suppose that at least one of the quantizers in the
%companding quantizer sequence $\{Q_{K}^{\lambda}\}$ possesses finite (discrete) entropy
%$H(Q^{\lambda}_K(X))$.  
Then if $\E{\log \lambda(X)}$ is finite,
\[
\lim_{R\rightarrow\infty} \left[ H(Q^{\lambda}_{K(R;\lambda)}(X)) - \log K(R;\lambda) \right] = h(X) + \E{\log\lambda(X)} \mbox{.}
\]
\end{lemma}
\begin{IEEEproof}
Follows as a special case of Proposition 2 in \cite{LinderZZ1999}.
\end{IEEEproof}

With the insight of this approximation, we define:
\begin{definition}
The variable-rate
approximate resolution-rate function $\Kvrhr(R;\lambda)$ is given by
\[
\log \Kvrhr(R;\lambda) = R-h(X)-\E{\log\lambda(X)} \mbox{.}
\]
\end{definition}
\begin{lemma}
\label{lem:ResRateError}
The error between the log of the variable-rate approximate resolution-rate function $\log \Kvrhr (R;\lambda)$
and the log of the actual resolution-rate function $\Kvr(R;\lambda)$ goes to zero, i.e.
\[ \lim_{R\rightarrow\infty} \log \Kvrhr(R;\lambda) - \log \Kvr(R;\lambda) = 0 \mbox{.}\]
\end{lemma}

\begin{IEEEproof}
The error of the approximation $\Kvrhr$ may be written as
\[
\log \Kvr(R;\lambda) - \log \Kvrhr(R;\lambda) = \epsilon_R + H(Q^{\lambda}_{\Kvr(R;\lambda)}(X)) - R \mbox{,}
\]
where $\epsilon_R$ goes to zero by Lemma~\ref{lem:ResolutionRate}.
Furthermore, by definition $\Kvr(R;\lambda)$ has been chosen to be the largest resolution
such that $H(Q^{\lambda}_{\Kvr(R;\lambda)}(X)) \leq R$. We then have that 
\[ R-H(Q^{\lambda}_{\Kvr(R;\lambda)}(X) ) < H(Q^{\lambda}_{\Kvr(R;\lambda)+1}(X)) - H(Q^{\lambda}_{\Kvr(R;\lambda)}(X)) \mbox{,}\]
i.e. the second term in the rate approximation error is bounded
by the increment in entropy from an increment in resolution.  
By Lemma~\ref{lem:ResolutionRate} once again, the increment in entropy
may be bounded as 
\beqan
\lefteqn{H(Q^{\lambda}_{\Kvr(R;\lambda)+1}(X)) - H(Q^{\lambda}_{\Kvr(R;\lambda)}(X))} \\ & = & 
h(X) +\log (\Kvr(R;\lambda)+1) + \E{\log \lambda(X)} - h(X) - \log \Kvr(R;\lambda) -\E{\log \lambda(X)} + \delta(R) \\
& = & \log (\Kvr(R;\lambda)+1) - \log \Kvr(R;\lambda) + \delta(R) \\
& = & \log \frac{\Kvr(R;\lambda)+1}{\Kvr(R;\lambda)} + \delta(R) \mbox{,}
\eeqan
where $\delta(R)$ goes to zero.  Since $\Kvr(R;\lambda)$ diverges to infinity with $R$, this error goes to zero.
\end{IEEEproof}

\subsubsection{The Distortion-Rate Functions}
\label{sec:BackgroundHighDistortionRate}
The high-resolution distortion-rate function can be obtained by combining
the distortion-resolution and resolution-rate functions.  For fixed-rate,
\begin{subequations}
\beq
\Dfrhr(R) = \frac{1}{12} \E{\lambda^{-2}(X)} 2^{-2R} \mbox{,} \label{eq:BackgroundFixedRateDistortion}
\eeq
whereas for variable-rate
\beq
\Dvrhr(R) = \frac{1}{12} \E{\lambda^{-2}(X)} 2^{-2(R-h(X)-\E{\log\lambda(X)})} \mbox{.}
\label{eq:BackgroundVariableRateDistortion}
\eeq
\end{subequations}
Asymptotic validity in the sense of \eqref{eq:asymptoticaccuracy}
follows in the fixed-rate case from \eqref{eq:distortionResolutionAsymptotic} and from
the fact that $\left( \Kfr(R;\lambda)/\Kfrhr(R;\lambda) \right)^2$ goes to 1\@.
In the variable-rate case, we may bound the error from use of $\Khr(R;\lambda)$
in place of $K(R;\lambda)$ as a multiplying factor of $2^{2|\Khr(R;\lambda) - K(R;\lambda)|}$,
which by Lemma \ref{lem:ResRateError} goes to 1\@.

\subsubsection{Asymptotically-Optimal Companding Quantizer Sequences}
\label{sec:review-optimal}
We seek asymptotically-optimal companding quantizer sequences for
both fixed-rate and variable-rate constraints.  By the following lemma,
this reduces to minimizing the high-resolution distortion-rate functions
of \eqref{eq:BackgroundFixedRateDistortion} and \eqref{eq:BackgroundVariableRateDistortion}.

\begin{lemma}
\label{lem:optimizationIsLegit}
Suppose $\lambda_{\rm fr}^*$ and $\lambda_{\rm vr}^*$ minimize $\Dfrhr(R;\lambda)$ 
and $\Dvrhr(R;\lambda)$ respectively.
Then the quantizer sequences $\{Q_{K}^{\lambda_{\rm fr}^*}\}$ and $\{Q_{K}^{\lambda_{\rm vr}^*}\}$
are asymptotically fixed- and variable-rate optimal.
\end{lemma}

\begin{IEEEproof}
As the proof is virtually identical for fixed- and variable-rate cases, we only provide it for 
the variable-rate case.

Let $\{Q_{K}^{\lambda}\}$ be any companding quantizer sequence.  We are interested in proving
that
\[
\limsup_{R\rightarrow \infty} \frac{\Dvr(R;\lambda_{\rm vr}^*)}{\Dvr(R;\lambda)} \leq 1 \mbox{.}
\]
The supremum limit on the left may be factored:
\beqan
\limsup_{R\rightarrow \infty} \frac{\Dvr(R;\lambda_{\rm vr}^*)}{\Dvr(R;\lambda)} & = &
\limsup_{R\rightarrow \infty} \frac{\Dvr(R;\lambda_{\rm vr}^*)}{\Dvrhr(R;\lambda_{\rm vr}^*)}
\,
							    \frac{\Dvrhr(R;\lambda_{\rm vr}^*)}{\Dvrhr(R;\lambda)}
\,
							    \frac{\Dvrhr(R;\lambda)}{\Dvr(R;\lambda)} \\
& \leqlabel{a} &
\limsup_{R\rightarrow \infty} \frac{\Dvr(R;\lambda_{\rm vr}^*)}{\Dvrhr(R;\lambda_{\rm vr}^*)}
\,
\limsup_{R\rightarrow \infty} \frac{\Dvrhr(R;\lambda_{\rm vr}^*)}{\Dvrhr(R;\lambda)}
\,
\limsup_{R\rightarrow \infty} \frac{\Dvrhr(R;\lambda)}{\Dvr(R;\lambda)} 
\eeqan
because the supremum limit of a product of positive sequences is upper-bounded by the product of their individual supremum limits.  We can now bound each of these factors.

We have, by optimality
of $\lambda_{\rm vr}^*$, that $\Dvrhr(R;\lambda) \geq \Dvrhr(R;\lambda_{\rm vr}^*)$
for any $R$ and therefore that
\[
\limsup_{R\rightarrow\infty} \frac{\Dvrhr(R;\lambda_{\rm vr}^*)}{\Dvrhr(R;\lambda)} \leq 1 \mbox{.}
\]
Furthermore, by \eqref{eq:asymptoticaccuracy}, we have that
\[
\lim_{R\rightarrow\infty} \frac{\Dvr(R;\lambda_{\rm vr}^*)}{\Dvrhr(R;\lambda_{\rm vr}^*)} = 
\lim_{R\rightarrow\infty} \frac{\Dvrhr(R;\lambda)}{\Dvr(R;\lambda} = 
1 \mbox{.}
\]
This proves the lemma.
\end{IEEEproof}

Now we optimize the distortion-rate expressions.
Because analogous optimizations appear in Sections~\ref{sec:Single}
and~\ref{sec:Multi}, we explicitly derive
both the optimizing point densities and the resulting distortion-rate functions.
Our approach follows~\cite{GrayG1977}.

In the fixed-rate case, the problem is to minimize \eqref{eq:BackgroundFixedRateDistortion} for
a given value of $R$.
This minimization may be performed with the help of H\"{o}lder's inequality:
\beqan
\Dfrhr(R;\lambda) & = & 	\frac{1}{12}2^{-2R} \int_0^1 f_X(x) \lambda^{-2}(x) dx \\
		& = & 	\frac{1}{12}2^{-2R} \int_0^1 f_X(x) \lambda^{-2}(x) dx  \left( \int_0^1 \lambda(x)dx\right)^2\\
		& \geq & 	\frac{1}{12}2^{-2R} \int_0^1 \left( f_X(x) \lambda^{-2}(x)\right)^{1/3}
									   \left(\lambda(x) \right)^{2/3}
									   dx\\
		& = & \frac{1}{12}2^{-2R} \left( \int_0^1 f_X^{1/3}(x) \right)^3	\mbox{,}		
\eeqan
with equality only if $\lambda(x)\propto f_X^{1/3}(x)$.
Thus,
$\Dfrhr$ is minimized by
\beq
  \label{eq:fixed-opt-lambda}
    \lambda(x) = f_X^{1/3}(x) / \left({\textstyle \int_0^1 f_X^{1/3}(t) \,dt}\right) \mbox{.}
\eeq
The resulting minimal distortion is  
\beq
  \label{eq:fixedHrDist}
\Dfrhr(R) = \frac{1}{12}2^{-2R} \left( \int_0^1 f_X^{1/3}(x) \, dx \right)^3
    = \frac{1}{12} \| f_X \|_{1/3} 2^{-2R} \mbox{,}
\eeq
where we have introduced a notation for the $\mathcal{L}^{1/3}$ quasinorm.

For the variable-rate optimization, we use Jensen's inequality rather than
 H\"{o}lder's inequality:
\beqan
\Dvrhr(R;\lambda)  & = & \frac{1}{12}2^{-2(R-h(X))} \E{\lambda^{-2}(X)}2^{2\E{\log \lambda(X)}} \\
& \stackrel{(a)}{\geq} & \frac{1}{12}2^{-2(R-h(X))} \E{\lambda^{-2}(X)}2^{2\log\E{\lambda(X)}} \\
& = & \Dvrhr(R) \mbox{,}
\eeqan
where (a) follows from the convexity of $-\log(\cdot)$.
This lower bound is achieved when $\lambda(X)$ is a constant.
Thus
$\lambda(x) = 1$ is asymptotically optimal, i.e., the quantizer should be uniform.%
%\footnote{Recall that for the variable-rate case we are assuming
% $f_X$ is supported on $[0,1]$. For other bounded supports,
% the optimal point density would still be a constant, but perhaps
% different from 1\@.  Unbounded supports require the use of an
% unnormalized point density.}
%The corresponding minimal distortion is 
%\beq
%  \label{eq:varHrDist}
%D_{HR} \approx \frac{1}{12} 2^{2h(X)} 2^{-2R} \mbox{.}
%\eeq

%Note that both optimal point densities are positive on the entire
%support of $f_X$.  Thus, at high enough resolution, the quantization is
%fine \emph{pointwise over $X$}.  In the functional settings,
%this will be used to justify piecewise linear approximation of the function $g$.
Note that both variable- and fixed-rate quantization
have $\Theta(2^{-2R})$, or $-6$ dB/bit, dependence of distortion on rate.
This is a common feature of
ordinary quantizers with MSE distortion, but we demonstrate in Section~\ref{sec:DontCare}
that certain functional scenarios can cause distortion to fall even faster
with the rate.

%One way to concretely specify a quantizer from a point density is to require
%$$
%  \Lambda(\beta_i) = i-\half, \qquad i=1,\,2,\,\ldots,\,K\mbox{,}
%$$
%where $\Lambda(x) = \int_0^x \lambda(t) \, dt$
%is the ``cumulative'' point density.
%However, analysis of quantizers through point densities does not rely on
%the precise placement of codewords and cell boundaries.
%Under the assumptions of high-resolution analysis,
%$o(1/K)$ deviations in the $\beta_i$s do not affect the distortion.
%We return to this point in Section~\ref{sec:single-discontinuous}
%to partially generalize the basic analysis to discontinuous functions.

\subsubsection{Optimal Bit Allocation}
As a final preparatory digression, we state the solution to a typical
resource allocation problem that arises several times in Section~\ref{sec:Multi}.
\begin{lemma}
  \label{lem:bit-alloc}
Suppose $D = \sum_{j=1}^n c_j 2^{-2R_j}$
for some positive constants $\{c_j\}_{j=1}^n$.
Then the minimum of $D$ over the choice of $\{R_j\}_{j=1}^n$
subject to the constraint $\sum_{j=1}^n R_j \leq nR$ is attained with
$$
  R_j = R + \Half \log \frac{c_j}{ \left( \prod_{j=1}^n c_j \right)^{1/n} },
\qquad j=1,\,2,\,\ldots,\,n,
$$
resulting in
$$
  D = n \left( {\textstyle \prod_{j=1}^n c_j} \right)^{1/n} 2^{-2R}.
$$
\end{lemma}
\begin{IEEEproof}
The result can be shown using the inequality for arithmetic and geometric means.
It appeared first in the context of bit allocation in~\cite{HuangS1963};
a full proof appears in \cite[Sect.~8.3]{GershoG1992}.
\end{IEEEproof}

The lemma does not restrict the $R_j$s to be nonnegative or to be integers.
Such restrictions are discussed in~\cite{FarberZ2006}.

\section{Univariate Functional Quantization}
\label{sec:Single}
%proposed (new) structure:

%state problem, with limitations on g and f, state why the ``easy'' solution is not being used.
%sufficiency of ghat = g
%sufficiency of regular quantizers
%Companding quantizers
%high resolution analysis of companding quantizers
	%Distortion/rate derivation
	%optimal point densities
%delete discontinuous functions section?

Let $X$ be a random variable with pdf $f_X(x)$ defined over $[0,1]$,
and let $g: [0,1]\rightarrow \R$ be the function of interest.
A sequence of companding quantizers $\{Q_{K}^\lambda\}$ is applied to
the source $X$, and an estimate $\ghat(Q^{\lambda}_K(X))$
is formed at the decoder, where $\ghat$
is the estimator function.
Functional distortion is measured by squared error
$D = \iE{(g(X) - \ghat(Q^{\lambda}_K(X)))^2}$.
We seek an asymptotically-optimal estimator $\ghat$ and companding function $w$
that satisfy certain constraints.

Since we seek to answer this design question with high-resolution techniques,
the function $g$ and the source $X$ must be restricted in a manner
similar to conditions UO1--4 in Section~\ref{sec:BackgroundProblemStatement}. For the moment we err on the 
side of being too strict. 
Sections~\ref{sec:NonMono} and~\ref{sec:DontCare} will
significantly loosen these requirements.

\begin{enumerate}
\item[UF1.] $g$ is monotonic.
\item[UF2.] $g$ is Lipschitz continuous on $[0,1]$, and the first- and second- derivatives of $g$ are
defined except possibly on a set of zero Jordan measure.
\item[UF3.] The source pdf $f$ is continuous, bounded, and supported on the interval $[0,1]$.
\item[UF4.] We optimize among companding functions $w$ that are
piecewise differentiable (and therefore a point density description $\lambda$ is appropriate).
\item[UF5.] The integral $\int_0^1 f(x)g'(x)^2\lambda(x)^{-2} \, dx$ is defined and finite.
\end{enumerate}

Throughout this paper, we assume that $\ghat(t) = \E{g(X) \mid X \in S_i}$ for all $t \in S_i$.
This achieves the minimum possible functional distortion $\E{\var \left( g(X) \mid Q^{\lambda}_K(X) \right)}$.

\subsection{Sufficiency of Regular Quantizers} 
The following lemma relates monotonicity to regularity of optimal quantizers,
thus justifying the optimization among companding quantizers:
\begin{lemma}
\label{lem:MonoRegular}
If $g$ is monotonic,
there exists an optimal functional quantizer of $X$ that is regular.
\end{lemma}
\begin{IEEEproof}
The optimal functional quantizer in one dimension is
induced by the optimal ordinary quantizer for the variable $Y = g(X)$. 
That is, one may compute the function $g(X)$ and quantize it directly. 
Since the optimal ordinary quantizer for a real-valued source is regular,
the optimal quantizer for $Y$, denoted by $Q_Y(y)$ and having points
$\{\widehat{y}_i\}_{i \in \I}$, is regular.

$Q_Y(y)$ may be implemented by a quantizer for $X$ with
cells
given by $g^{-1}(Q_Y^{-1}(\widehat{y}_i))$.
We know that $Q_Y^{-1}(\widehat{y}_i)$ is an interval since $Q_Y$ is regular. 
Also, since
$g$ is monotonic, the inverse map $g^{-1}$ applied to any interval
in the range of $g$ gives an interval.
Thus $g^{-1}(Q_Y^{-1}(\widehat{y}_i))$ is an interval,
which demonstrates that there exists a regular quantizer in $X$ that
is optimal.
\end{IEEEproof}

\subsection{The Distortion-Resolution Function} 
Assumption UF2 is introduced so that a piecewise linear approximation
of $g$ suffices in estimating the functional distortion of the quantizer.
More specifically, recalling the notation $\{\beta_i\}_{i\in\I}$ for the quantizer points
and $\{S_i\}_{i\in\I}$ for the partition,
$$
  \gPL(x) = g(\beta_i) + g'(\beta_i)(x - \beta_i), \quad \mbox{for $x \in S_i$}, \quad i \in \I
$$
may be interpreted as  an approximation of $g$ that leads to the high-resolution
approximate distortion-resolution function.
%Because $g$ is differentiable almost everywhere,
%and because $g'$ is bounded, $g$ is Lipschitz continuous.
%Therefore there exists a positive constant K such that
%Excluding partition cells in which $g''(x)$ does not exist,
%for any $x \in S_i$,
%\begin{equation}
%  \label{eq:taylor}
%  |g(x) - \gPL(x)| \leq \half \left(\max_{\xi \in S_i} |g''(\xi)| \right)
%                              (\length(S_i))^2
%\end{equation}
%by Taylor's theorem.
%Then, invoking Assumption UF2 and the fact that $\length(S_i)$ vanishes for all
%$S_i$s that intersect the support of $f_X$, we see that $g$ converges to $\gPL$
%within every quantizer cell in a precise sense.

The use of $\gPL$ prompts us to give a name to the magnitude of the
derivative of $g$.  The distortion is then expressed using this function.
\begin{definition}
  The \emph{univariate functional sensitivity profile} of $g$ is defined
  as $\gamma(x) = |g'(x)|$.
\end{definition}

\begin{theorem}
 \label{thm:single-distortion}
Suppose a source $X \in [0,1]$
is quantized by a sequence of companding quantizers $\{Q^{\lambda}_K\}$ with point density $\lambda(x)$
and increasing resolution $K$.
Further suppose that the source, quantizer, and function
$g: [0,1] \rightarrow \R$ satisfy Assumptions UF1--5\@.
Then the high-resolution distortion-resolution function is an
\emph{asymptotically accurate} approximation of the true distortion-resolution function:
\beq
  \label{eq:1dUnoptimizedDist}
d(K; \lambda) = \E{\var \left( g(X) \mid Q_{\lambda}^K(X) \right)}
  \sim \frac{1}{12K^2} \E{\left(\Frac{\gamma(X)}{\lambda(X)}\right)^2} = \dhr(K;\lambda) \mbox{.}
\eeq
\end{theorem}
\begin{IEEEproof}
Follows as a special case of Theorem \ref{thm:multi-distortion}.
\end{IEEEproof}

\subsection{The Resolution-Rate Functions}
The relationship between resolution and rate in the functional
context is unchanged from the ordinary context.  
For a fixed-rate constraint, the resolution-rate function
is given by $\Kfr(R;\lambda) = \lfloor 2^R \rfloor$ and
is approximated at high-resolution by $\Kfrhr(R;\lambda) = 2^R$.
For a variable-rate constraint, the resolution-rate function
is given by the highest resolution such that the entropy of the
quantized output is less than the rate constraint.  This is approximated
as before by $\log \Kvrhr(R;\lambda) = R-h(X) - \E{\log \lambda(X)}$.
Both of these approximations continue to be asymptotically accurate,
regardless of the distortion measure in use.

\subsection{The Distortion-Rate Functions}
By combining the distortion-rate function with the resolution-rate function,
the high-resolution distortion-rate function can be obtained.  For fixed-rate,
\begin{subequations}
\beq
\Dfrhr(R;\lambda) = \frac{1}{12} \E{(\gamma(X)/\lambda(X))^2}2^{-2R} \mbox{,}
\label{eq:frDistortionRate1D}
\eeq
whereas for variable-rate,
\beq
\Dvrhr(R;\lambda) = \frac{1}{12}\E{(\gamma(X)/\lambda(X))^2}2^{-2(R-h(X)-\E{\log \lambda(X)})} \mbox{.}
\label{eq:vrDistortionRate1D}
\eeq
\end{subequations}

The asymptotic validity of these two expressions, as in \eqref{eq:asymptoticaccuracy}, 
holds as it did in the ordinary case.  For the fixed-rate expression, this follows from
Theorem \ref{thm:single-distortion} and the fact that $\lfloor 2^R \rfloor 2^{-R}$ approaches
1.  For the variable-rate expression, the error from use of $\Khr(R;\lambda)$ in the distortion-rate
expression instead of $K(R;\lambda)$ can be bounded as a multiplying factor of $2^{2|\Khr(R;\lambda) - K(R;\lambda)|}$,
which by Lemma \ref{lem:ResRateError} goes to 1.

\subsection{Asymptotically-Optimal Companding Quantizer Sequences}
We seek asymptotically-optimal companding quantizer sequences for fixed- and
variable-rate constraints under a functional distortion measure.  The lemma below
demonstrates that
it suffices to optimize the high-rate distortion-rate functions $\Dfrhr$ and $\Dvrhr$.

\begin{lemma}
\label{lem:singleOptimizationIsLegit}
Suppose $\lambda^*_{\rm fr}$ and $\lambda^*_{\rm vr}$ minimize
$\Dfrhr(R;\lambda)$ and $\Dvrhr(R;\lambda)$ respectively.  Then the quantizer
sequences $\{Q_{K}^{\lambda^*_{\rm fr}}\}$ and $\{Q_{K}^{\lambda^*_{\rm vr}}\}$ are asymptotically
fixed- and variable-rate optimal.
\end{lemma}
\begin{IEEEproof}
The proof is virtually identical to that of Lemma \ref{lem:optimizationIsLegit}.
\end{IEEEproof}

The distortion expression \eqref{eq:1dUnoptimizedDist} bears strong resemblance to \eqref{eq:unoptimizedHrDist}, but with the
probability density $f_X(x)$ replaced with a \emph{weighted density}
$\gamma^2(x)f_X(x)$.  Unlike the density $f_X(x)$, the weighted density $\gamma^2(x)f_X(x)$ need not
integrate to one.
Optimal point densities and the resulting distortions now follow easily.

For fixed-rate coding, we are attempting to minimize the distortion
\eqref{eq:1dUnoptimizedDist} for a given value of $K$.
Following the arguments in Section~\ref{sec:review-optimal},
the optimal point density is proportional to
the cube root of the weighted density:
\beq
 \label{eq:1dFixedLambda}
  \lambda^*_{\rm fr}(x) = \frac{\left(\gamma^2(x)f_X(x)\right)^{1/3}}
                    {\int_0^1 \left(\gamma^2(t)f_X(t)\right)^{1/3} \, dt} \mbox{.}
\eeq
The admissibility of this point density (assumption UF5) requires
positivity of $\lambda(x)$ everywhere $f_X$ is positive.
This excludes the possibility that $\gamma(x) = 0$ for
an interval $x \in (a,b)$ such that $\P{X \in (a,b)} > 0$
because in this case the quantization is not fine for $X \in (a,b)$.
We revisit this restriction in Section~\ref{sec:DontCare}.
By evaluating \eqref{eq:frDistortionRate1D} with point density
\eqref{eq:1dFixedLambda}, the resulting distortion is
\beq
 \label{eq:1dFixedDist}
   \Dfrhr(R) = \Dfrhr(R;\lambda^*_{\rm fr}) = \frac{1}{12} \left\| \gamma^2 f_X \right\|_{1/3} 2^{-2R} \mbox{.}
\eeq

For variable-rate coding, a derivation very similar to that of ordinary variable-rate quantization
may be performed.  This yields an optimal point density that is proportional to
the functional sensitivity profile:
\beq
 \label{eq:1dOptimalVariableLambda}
  \lambda^*_{\rm vr}(x) = \frac{\gamma(x)}{\int_0^1 \gamma(t) \, dt} \mbox{.}
\eeq
The restriction for $\lambda$ to be positive wherever $f_X$ is positive
takes the same form as above (assumption UF5). The resulting distortion is
\beq
 \label{eq:1dVariableDist}
   \Dvrhr(R) = \Dvrhr(R;\lambda^*_{\rm vr}) = \frac{1}{12}  2^{2h(X)+2\E{\log \gamma(X)}} \, 2^{-2R} \mbox{.}
\eeq

The example below shows that even for univariate functions,
there are benefits from functional quantization.
It also illustrates the difference between the fixed- and variable-rate cases.
While quantizing $X$ instead of $g(X)$ seems na{\"i}ve,
as we move to the distributed multivariate case it will not be possible
to compute the function before quantization.

\smallskip

\begin{example}
\label{ex:single}
Suppose $X$ is uniformly distributed over $[0,1]$ and $g(x) = x^2$.
For both fixed- and variable-rate, the optimal ordinary quantizer is uniform,
i.e., $\lord = 1$.
With $\gamma(x) = 2x$, evaluating \eqref{eq:frDistortionRate1D} gives
$\Dfrhr(R;\lord) = \Dvrhr(R;\lord) = \frac{1}{9} 2^{-2R} \approx 0.111 \cdot 2^{-2R}$.

The optimal point density for fixed-rate functional quantization is
$\lfr^*(x) = \frac{5}{3}x^{2/3}$ and yields distortion
\[
\Dfrhr(R) = \frac{1}{12}\| (2x)^2 \|_{1/3} \cdot 2^{-2R}
    = \frac{9}{125} 2^{-2R}
    \approx 0.072 \cdot 2^{-2R} \mbox{.}
\]
The optimal point density for variable-rate functional quantization is
$\lvr^*(x) = 2x$.
With $h(X) = 0$ and $\E{\log \gamma(X)} = 1 - 1/(\ln 2)$,
the resulting distortion is
\[ 
\Dvrhr(R) = \frac{1}{12} \cdot 4e^{-2} \cdot 2^{-2R}
   \approx 0.045 \cdot 2^{-2R} \mbox{.}
\]
Quantizers designed with the three derived optimal point densities
are illustrated in Fig.~\ref{fig:single-example} for rate $R = 4$.
The functionally-optimized quantizers put more points at higher
values of $x$, where the function varies more quickly.
In addition, the variable-rate quantizer is allowed more points
($K = 21$) while meeting the rate constraint.

\begin{figure}
  \centering
    \psfrag{a}[l][l]{\small $\lord$}
    \psfrag{b}[l][l]{\small $\lfr$}
    \psfrag{c}[l][l]{\small $\lvr$}
    \psfrag{0}[][t]{\small 0}
    \psfrag{1}[][t]{\small 1}
    \includegraphics[width=3in]{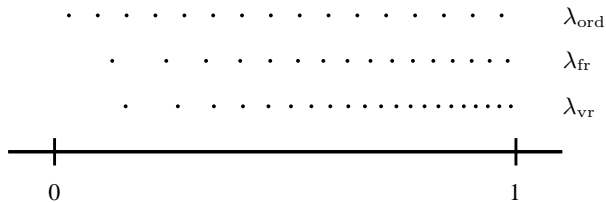}
  \caption{Quantizer points illustrating the point densities derived in
    Example~\ref{ex:single} at rate $R=4$.}
  \label{fig:single-example}
\end{figure}

The interested reader can verify that $\Dfrhr(R)$ and $\Dvrhr(R)$
exactly match the performance obtained by designing optimal
quantizers for $Y = X^2$. \hfill $\Box$
\end{example}

\smallskip

In the second example, we use a nonuniform source pdf with the
same nonlinear function $g$ to illustrate various quantities.

\smallskip

\begin{example}
\label{ex:single-2}
Suppose $X$ has the pdf $f_X(x) = 3x^2$ over $[0,1]$ and $g(x) = x^2$.
We illustrate a codebook-constrained quantizer with rate $R=2$
designed with the high-resolution analysis.

By evaluating \eqref{eq:1dFixedLambda},
the asymptotically-optimal point density for
fixed-rate functional quantization is
$\lfr^*(x) = \frac{7}{3}x^{4/3}$.
Integrating the point density gives the corresponding compander function
$\wfr^*(x) = x^{7/3}$.
As shown in the top panel of Fig.~\ref{fig:single-example-2},
the points are given by
$$
  \beta_i = {\wfr^*}^{-1}((2i-1)/8),
\qquad
i = 1,\ 2,\ 3,\ 4,
$$
and the cell boundaries are given by 
${\wfr^*}^{-1}(\{0,\, 1/4,\, 1/2,\, 3/4,\, 1\})$.
The middle panel shows $f_X$ and an approximation $\widehat{f}_X$ that is constant
on each cell of the quantizer.
The bottom panel shows $g$ and the approximation $\gPL$,
which is linear on each cell of the quantizer and tangent to $g$
at each point. \hfill $\Box$
\end{example}

\begin{figure}
  \centering
    \psfrag{x}[][]{\small $x$}
    \psfrag{w}[b][]{\small $\wfr^*$}
    \psfrag{f}[b][t]{\small $f_X$}
    \psfrag{g}[b][]{\small $g$}
    \includegraphics[width=3in]{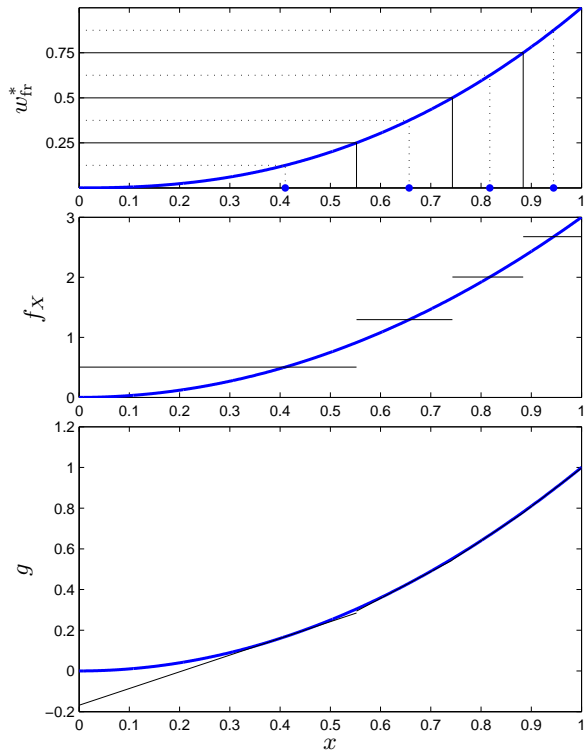}
  \caption{Illustrations for Example~\ref{ex:single-2}.
    Top panel:  points and cell boundaries of the quantizer are
          determined by the companding function $\wfr^*$.
    Middle panel:  source pdf $f_X$ and its piecewise constant approximation.
    Bottom panel:  function $g$ and its piecewise linear approximation.}
  \label{fig:single-example-2}
\end{figure}

\smallskip

Referring to Fig.~\ref{fig:single-example-2} for examples,
the high-resolution distortion-resolution function
$\dhr(K;\lambda)$ can be thought of as a computation
of the MSE of $\gPL$ when the source with piecewise constant pdf $\widehat{f}_X$
is quantized with companding quantizer employing compander $\wfr^*$.
In this case $\ghat$, the optimal function estimate,
is given by evaluating $\gPL$ at the center of the cell containing
the source variable.
Informally, as resolution $K$ increases,
$\widehat{f}_X \rightarrow f$,
$\gPL \rightarrow g$,
and the centers of the cells approach the corresponding quantizer points.
These intuitions extend to multivariate functions as well,
but our formal justifications in Section~\ref{sec:Multi}
use techniques that do not explicitly form approximations
$\widehat{f}_X$ or $\gPL$.

%The approach of linearly approximating $g$ will generalize to
%allow optimization of quantizers.

\subsection{Discontinuous Functions}
\label{sec:single-discontinuous}
Our main result on univariate functional quantization,
Theorem~\ref{thm:single-distortion},
assumes the continuity of $g$.
One can effectively sidestep this assumption,
but doing so requires the quantizer to be described more precisely
than by a point density function alone.

For simplicity, assume $f_X$ is strictly positive on $[0,1]$.
Suppose we were to allow $g$ to have a point of discontinuity $x_0 \in (0,1)$
with
$$
  c_0 = \lim_{\delta \rightarrow 0} |g(x_0+\delta) - g(x_0-\delta)| > 0.
$$
The difficulty that arises is that if $x_0$ is an interior point of
a partition cell $S_i$, this cell produces a component of the
functional distortion proportional to $c_0^2 \P{X \in S_i}$.
Since $c_0^2 \P{X \in S_i} = \Theta(K^{-1})$, 
it is not negligible in comparison to the (best case) $\Theta(K^{-2})$
functional distortion.
Thus having a point of discontinuity of $g$ in the interior of a
partition cell disrupts the asymptotic distortion calculation
\eqref{eq:1dUnoptimizedDist}.

The representation of quantizers by number of levels $K$ and
point density function $\lambda$ cannot prevent a point
of discontinuity from falling in the interior of a partition cell.
However, if we augment the description of the quantizer with
specified partition boundaries, we can still obtain the distortion estimate
\eqref{eq:1dUnoptimizedDist}.

\begin{corollary}
 \label{cor:single-discontinuous}
Suppose a companding quantizer sequence for a source $X \in [0,1]$
is described by point density function $\lambda(x)$.
Further suppose that the source, quantizer, and function
$g: [0,1] \rightarrow \R$ satisfy Assumptions UF1--5
with the exception of discontinuities at $M$ points $\{x_m\}_{m=1}^M$.
Then a quantizer sequence obtained by adding partition cell boundaries
at $\{x_m\}_{m=1}^M$ will have distortion
$$
d_g = \E{(g(X)-g(\Xhat))^2}
  \sim \frac{1}{12K^2} \E{\left(\Frac{\gamma(X)}{\lambda(X)}\right)^2} \mbox{.}
$$
\end{corollary}
\begin{IEEEproof}
This follows from Theorem~\ref{thm:single-distortion} applied separately to each of 
the subintervals where $g$ is continuous.
\end{IEEEproof}

In the sequel, we will not consider discontinuous functions.
The multivariate extension of
Corollary~\ref{cor:single-discontinuous}
requires points of discontinuity to be in the Cartesian product
of finite sets of discontinuity for each variable.
Such separable sets of points of discontinuity are not general
and can be handled rather intuitively.

\section{Multivariate Functional Quantization}
\label{sec:Multi}
With Section~\ref{sec:Single} as a warm-up, we may now establish the
central results of distributed functional quantization.
\subsection{Definitions}
\label{sec:multi-definitions}
An $n$-dimensional \emph{distributed companding quantizer} $Q^{\wB}_\KB$
is specified by $n$ companding functions $\wB = (w_1,w_2,\ldots,w_n)$
and an $n$-vector of resolutions $\KB = (K_1,K_2,\ldots,K_n)$.  When applied to an $n$-tuple $x_1^n \in[0,1]^n$,
$Q^{\wB}_\KB$ quantizes each component $x_j$ of $x_1^n$ separately with compander
$w_j$ and resolution $K_j$:
\[
Q^{\wB}_\KB(x_1^n) = 
\left(
			Q^{w_1}_{K_1}(x_1),  
			Q^{w_2}_{K_2}(x_2),
			\ldots,
			Q^{w_n}_{K_n}(x_n)
\right)
\mbox{.}
\]
%A sequence of distributed companding quantizers $\{Q^{\wB,\aB}_K\}$
%consists of distributed quantizers $Q^{\wB,\aB}_K$ indexed by an 
%increasing resolution parameter $K$. 
%If the companding functions $\wB$ are all monotonic, they may equivalently
A distributed companding quantizer may equivalently be specified  by 
$n$ point density functions $\lB = (\lambda_1,\lambda_2,\ldots,\lambda_n)$, 
in which case it is denoted by $Q^{\lB}_\KB$.
%and the quantizer
%sequence written as $\{Q_{K}^{\lB,\aB}\}$.

An estimation function $\ghat:[0,1]^n \rightarrow \R$ estimates the value of 
$g(X_1^n)$ from the quantized representation $Q^{\lB}_\KB(X_1^n)$.
The distortion of a distributed quantizer paired with an estimator $\ghat$ is given by
the distortion-resolution function
\[d_{\ghat}(\KB;\lB) = \iE{| g(X_1^n) - \ghat(Q^{\lB}_\KB(X_1^n)) |^2} \mbox{.} \]
In this paper, use of the optimal estimator
\[ \ghat(x_1^n) = \iE{g(X_1^n) \mid Q^{\lB}_\KB(X_1^n) = Q^{\lB}_\KB(x_1^n)} \]
will be indicated by omitting the subscript: $d(\KB;\lB)$.

%The allocation-optimal distributed quantizer $Q^{\lB}_K$ is defined for a source
%$X_1^n$ as the lowest-distortion distributed companding quantizer with point densities
%$\lB$ and a total cell count upper bounded by $K$: 
%\[ Q^{\lB}_K = \argmin_{Q_\KB^\lB: \prod_{j=1}^n K_j \leq K} d(\KB;\lB) \mbox{.} \]
%The corresponding distortion is denoted by 
%\[ d(K;\lB) = \min_{Q_\KB^\lB: \prod_{j=1}^n K_j \leq K} d(\KB;\lB) \mbox{.} \]
%We write $\{Q^{\lB}_K\}$ to indicate a sequence of allocation-optimal distributed quantizers
%indexed by the total resolution $K$.

The \emph{rate} $R$ of a distributed quantizer takes on three different meanings.  
A \emph{fixed-rate} constraint limits the total resolution $K = \prod_{j=1}^n K_j \leq 2^R$,
and we assume that the $j$th quantizer communicates to the decoder with rate $R_j = \log K_j$.  
A \emph{variable-rate}
(marginal entropy) constraint limits the sum of the marginal entropies 
$\sum_{j=1}^n H(Q^{\lambda_j}_{K_j}(X_j)) \leq R$, and we assume that the $j$th quantizer
utilizes entropy-coding to the decoder to attain rate $R_j = H(Q^{\lambda_j}_{K_j}(X_j))$.  
A \emph{Slepian--Wolf}
(joint entropy) constraint limits the joint entropy $H(Q^{\lB}_K(X_1^n)) \leq R$, and
we assume that the $j$th quantizer utilizes Slepian-Wolf coding to the decoder to attain
rate $R_j = H(Q^{\lambda_j}_{K_j} (X_j) \mid Q^{\lambda^{j-1}}_{K^{j-1}}(X^{j-1}))$, where
$Q^{\lambda^{j-1}}_{K^{j-1}}(X^{j-1})$ is used to represent 
$(Q^{\lambda_1}_{K_1}(X_1), \ldots, Q^{\lambda_{j-1}}_{K_{j-1}}(X_{j-1}))$.
Note that the choice of this particular point on the Slepian-Wolf rate boundary is arbitrary.
%The fixed-rate, variable-rate, and Slepian--Wolf resolution-rate functions 
%$\Kfr(R;\lB)$, $\Kvr(R;\lB)$,
%and $\Ksw(R;\lB)$ are the largest total resolutions such that
%$Q^{\lB}_{\Kfr}$, $Q^{\lB}_{\Kvr}$, and $Q^{\lB}_{\Ksw}$ satisfy their respective 
%rate constraints.
%The component resolution-rate functions
%$K_{{\rm (fr,vr,sw)},i}(R;\lB;\aB) = \lfloor K_{\rm (fr,vr,sw)}(R;\lB;\aB)^{\alpha_i} \rfloor$ are defined
%as the resolutions for each of the components of a rate-constrained distributed quantizer.
The resulting performance is measured by the distortion-rate functions 
\[ \Dfr(R;\lB) = \min_{\KB: \prod_{j=1}^n K_j \leq 2^{R}} d(\KB;\lB) \mbox{,} \]
\[\Dvr(R;\lB) = \min_{\KB: \sum_{j=1}^n H(Q^{\lambda_j}_{K_j}(X_j)) \leq R} d(\KB;\lB) \mbox{,}\]
and 
\[\Dsw(R;\lB) = \min_{\KB: \sum_{j=1}^n H(Q^{\lambda_j}_{K_j}(X_j) \mid Q^{\lambda^{j-1}}_{K^{j-1}}(X^{j-1})) \leq R}  d(\KB;\lB) \mbox{.} \]

A quantizer point density $\lB^*$ is \emph{asymptotically better} than another 
$\lB$ under a fixed-rate, variable-rate, or Slepian--Wolf constraint if
the ratio of the distortion-rate functions is at most one:
\begin{subequations}
\label{eq:asymptoticAccuracyMulti}
\beqa
\lim_{R\rightarrow \infty} \frac{D_{\rm fr}(R;\lB^*)}{D_{\rm fr}(R; \lB)} 
& \leq & 1 \mbox{,} \\
\lim_{R\rightarrow \infty} \frac{D_{\rm vr}(R;\lB^*)}{D_{\rm vr}(R; \lB)} 
& \leq & 1 \mbox{, or} \\
\lim_{R\rightarrow \infty} \frac{D_{\rm sw}(R;\lB^*)}{D_{\rm sw}(R; \lB)} 
& \leq & 1 \mbox{.}
\eeqa
\end{subequations}
If $\lB$ is asymptotically better than any other distributed quantizer sequence,
%under a (fixed-rate,variable-rate,Slepian--Wolf) constraint,
it is \emph{asymptotically optimal}.

\subsection{Problem Statement}
\label{sec:multi-assumptions}
Let $X_1^n$ be a random vector with joint pdf $f_{X_1^n}(x_1^n)$ defined over $[0,1]^n$,
and let $g: [0,1]^n\rightarrow \R$ be the function of interest.
A distributed companding quantizer $\{Q^{\lB}_\KB\}$
is applied to $X_1^n$.  Equivalently, a companding quantizer
$Q^{\lambda_j}_{K_j}$ is applied to each component of the source $X_j$.
The decoder then forms an estimate $\ghat(Q^{\lB}_{\KB}(X_1^n))$,
where $\ghat(Q^{\lB}_{\KB}(X_1^n)) = \iE{g(X_1^n) \mid Q^{\lB}_K(X_1^n)}$ is the optimal estimation function.
Distortion is measured by squared error in the function 
$D = \iE{(g(X_1^n)-\ghat(Q^{\lB}_{K}(X_1^n)))^2}$, which
for the optimal estimator reduces to $D = \E{\var \left( g(X_1^n) \mid Q^{\lB}_K(X_1^n) \right)}$.
Fig.~\ref{fig:dfsc-block} depicts this scenario, with $R_j = \log K_j$
in the fixed-rate case, $R_j = H(Q^{\lambda_j}_{K_j})$ in the variable-rate case,
and $R_j = H(Q^{\lambda_j}_{K_j}(X_j) \mid Q^{\lambda^{j-1}}_{K^{j-1}}(X^{j-1}))$ in the
Slepian-Wolf case.
We wish to choose $\lB$ to be asymptotically optimal.

As in Section~\ref{sec:Single}, we will impose restrictions on the function $g$
and the joint probability distribution function of $X_1^n$
so that a local affine approximation is effective.  

\begin{enumerate}
\item[MF1.] $g$ is Lipschitz continuous,
%continuous in each variable, its gradient is bounded when defined,
and the first and second derivatives of $g$ are defined except possibly on a set
of zero Jordan measure.
\item[MF2.] The source pdf $f$ is continuous and supported on $[0,1]^n$, and is therefore bounded.
\item[MF3.] We optimize among companding functions $w_j$ that are
piecewise differentiable (and therefore a point density description $\lambda_j$ is appropriate).
\item[MF4.] Letting $g_j(x_1^n)$ denote $\Frac{\del g(x_1^n)}{\del x_j}$, the integrals 
\[
\int_0^1 f(x_j)\E{\left| g_j(X^n)\right|^2 \mid X_j = x_j}\lambda_j(x_j)^{-2} \, dx_j
\]
are defined and positive for all $j \in \{1,\,2, \,\ldots,n\}$.
\end{enumerate}
Constraints MF1--MF4 are more restrictive than they need to be,
but this helps in simplifying proofs.  For instance, condition MF4 
guarantees that every source variable must be finely quantized
for distortion to approach zero.  If this is violated for the $j$th source
variable, it merely implies that a finite-resolution quantization of $X_j$
suffices.

Note that there is no analogue to the monotonicity assumption UF1 in
the multivariate case.  It can be shown that if $g$ is monotonic in each
of its variables the optimal fixed-rate distributed quantizer is regular.
With the added restriction that the source variables be independent,
it can be shown that the optimal variable-rate distributed quantizer is also regular,
via techniques similar to those of \cite{GyorgyL2002}.  Rather than constraining
the function $g$ and the source pdf $f$ in this manner, however,
assumption MF3 explicitly restricts optimization to the space of regular companding quantizer
sequences, regardless of whether regularity is optimal.
In Sec. \ref{sec:NonMono} it is shown that nonregular companding quantizer sequences
are asymptotically suboptimal for a wide variety of functions $g$, giving this constraint some validity.

\subsection{High-Resolution Analysis}

\subsubsection{The Distortion-Resolution Function}
Our main technical task in finding the optimal quantizers is to justify an
approximation of the distortion in terms of point density functions.
Since the quantization is distributed, our concept of functional sensitivity
is now extended to each variable separately, with averaging performed over
the remaining variables.
\begin{definition}
The $j$th \emph{functional sensitivity profile} of $g$ is defined as
\beq
  \gamma_j(x) = \left(\E{\left|{g_j(X_1^n)}\right|^2 \mid X_j = x}\right)^{1/2}\mbox{.}
\eeq
\end{definition}
\begin{theorem}
 \label{thm:multi-distortion}
Suppose $n$ sources $X_1^n \in [0,1]^n$
are quantized by a distributed companding quantizer
$Q^{\lB}_\KB$, and suppose that the source, quantizers, and function
$g: [0,1]^n \rightarrow \R$ satisfy assumptions MF1--4\@.
Let $d(\KB;\lB) = \E{\var \left( g(X_1^n)\mid Q^{\lB}_\KB(X_1^n) \right) }$ denote
the true distortion-resolution function, and let
$d^{\rm HR}$ denote the high-resolution approximate distortion-resolution function:
\beq
d^{\rm HR}(\KB;\lB) = \sum_{j=1}^n \frac{1}{12K_j^2} \E{\left(\frac{\gamma_j(X_j)}{\lambda_j(X_j)}\right)^2} \mbox{.}
  \label{eq:ndUnoptimizedDist}
\eeq
Then $d(\KB;\lB) \sim d^{\rm HR}(\KB;\lB)$, where $\sim$ indicates that the
ratio of the two quantities approaches one as the smallest element of the vector
$\KB$ grows without bound.
\end{theorem}
\begin{IEEEproof}
See Appendix~\ref{app:distortion}.
\end{IEEEproof}

\subsubsection{Connecting Resolution to Rate}
To convert the distortion-resolution function to a distortion-rate function, we first introduce
a slight generalization of the high-resolution resolution-rate relationship.
\begin{lemma}
\label{lem:ResRateErrorSW}
If the source $X_1^n$ has a density over $[0,1]^n$ with finite differential entropy
$h(X_1^n)$ and
% and additionally if any of the distributed quantizers in the distributed companding
%quantizer sequence $\{Q_{K}^{\lB,\aB}\}$ possess finite discrete entropy $H(Q_K^{\lB,\aB}(X_1^n))$,
if $\E{\log \lambda_j(X_j)}$ is finite for all $j \in \{1,\ldots,n\}$, 
then as each component of the resolution vector $\KB$ diverges,
\[
 H\left(Q_{\KB}^{\lB}(X_1^n)\right) - \sum_{j=1}^n \log K_i \rightarrow h(X_1^n) + \sum_{j=1}^n\E{\log\lambda_j(X_j)} \mbox{.}
\]
\end{lemma}
\begin{IEEEproof}
Suppose $W = (w_1(x_1),w_2(x_2),\ldots,w_n(x_n))$
is an $n$-dimensional companding function that is applied to the source $X_1^n$ prior to
quantization by a rectangular lattice quantizer $Q^U$ with side length $K_j^{-1}$ on the $j$th side,
 and furthermore suppose $W^{-1}$ is then applied to estimate
the source.  The output of this quantization process $W^{-1}(Q^U(W(X_1^n)))$ is 
identical to the scenario we consider, and since $W^{-1}$ is one-to-one,
the joint discrete entropy of the outputs are identical as well:
$H(Q^U(W(X_1^n))) = H(Q_{K}^{\lB}(X_1^n))$.

Since the volume of each cell of the rectangular lattice $Q^U$ is
equal to $K^{-1}$, and since the diameter of each cell falls to zero, a special case of a result
by Csisz\'{a}r \cite{Csiszar1973,LinderZ1994} tells us that
\[
\lim_{K\rightarrow \infty} H(Q^U(W(X_1^n))) - \log K = h(W(X_1^n)) \mbox{.}
\]
Since the differential entropy of a continuously differentiable function of $X$
is given by $h(f(X)) = h(X) +\E{\log \det J_f(X)}$, where $J_f(X)$
denotes the Jacobian matrix for the function $f$, we may reduce the expression to
\[
\lim_{K\rightarrow \infty} H(Q^U(W(X_1^n))) - \log K = 
h(X_1^n) + \sum_{j=1}^n \E{\log \lambda_j(X_j)} \mbox{.}
\]
Recalling that $H(Q^U(W(X_1^n))) = H(Q_{K}^{\lB}(X_1^n))$,
the proof is complete.
\end{IEEEproof}

Armed with this, the distortion-resolution function may be modified to include
considerations of rate.
\begin{lemma}
Define the fixed-rate, variable-rate, and Slepian-Wolf distortion-resolution functions as
\beqan
d^{\rm HR}_{\rm fr}(\KB;\lB) & = & \sum_{j=1}^n \frac{1}{12K_j^2} \E{\left(\frac{\gamma_j(X_j)}{\lambda_j(X_j)}\right)^2} \mbox{,}   \\
d^{\rm HR}_{\rm vr}(\KB;\lB) & = & \sum_{j=1}^n \frac{1}{12} 2^{-2H(Q^{\lambda_j}_{K_j}(X_j)) + 2h(X_j) +2\E{\log \lambda_j(X_j)}} \E{\left(\frac{\gamma_j(X_j)}{\lambda_j(X_j)}\right)^2}   \mbox{,} \\
d^{\rm HR}_{\rm sw}(\KB;\lB) & = & \sum_{j=1}^n \frac{1}{12} 2^{-2H(Q^{\lambda_j}_{K_j}(X_j)\mid Q^{\lambda^{j-1}}_{K^{j-1}}(X^{j-1})) + 2h(X_j \mid X^{j-1}) +2\E{\log \lambda_j(X_j)}} \E{\left(\frac{\gamma_j(X_j)}{\lambda_j(X_j)}\right)^2} \mbox{.}
\eeqan
Then $d(\KB; \lB) \sim d^{\rm HR}_{\rm fr,vr,sw}(\KB;\lB)$.
\end{lemma}
\begin{IEEEproof}
By Theorem \ref{thm:multi-distortion}, $d(\KB;\lB) \sim d_{\rm fr}^{\rm HR}(\KB;\lB)$.   This establishes the first of the asymptotic equalities.  

For the second (variable-rate) asymptotic equality, we observe that by Lemma \ref{lem:ResolutionRate}, 
\[ K_j \sim 2^{-2H(Q^{\lambda_j}_{K_j}(X_j)) + 2h(X_j) +2\E{\log \lambda_j(X_j)}} \]
and therefore that $d_{\rm vr}^{\rm HR}(\KB; \lB) \sim d^{\rm HR}(\KB;\lB)$.  Again, by Theorem \ref{thm:multi-distortion}, $d(\KB;\lB) \sim d^{\rm HR}(\KB; \lB)$.

For the third (Slepian-Wolf) asymptotic equality, we start by noting that by Lemma \ref{lem:ResRateErrorSW},
\[ \prod_{i=1}^j K_i \sim 2^{-2H(Q^{\lambda^{j}}_{K^j}(X^j)) + 2h(X^j) +2\sum_{i=1}^j \E{\log \lambda^j(X^j)}} \mbox{,} \]
and similarly
\[ \prod_{i=1}^{j-1} K_i \sim 2^{-2H(Q^{\lambda^{j-1}}_{K^{j-1}}(X^{j-1})) + 2h(X^{j-1}) +2\sum_{i=1}^{j-1} \E{\log \lambda^{j-1}(X^{j-1})}} \mbox{.} \]
Dividing the first by the second yields that
\[ K_j \sim 2^{-2H(Q^{\lambda_j}_{K_j}(X_j) \mid Q^{\lambda^{j-1}}_{K^{j-1}}(X^{j-1})) + 2h(X_j \mid X^{j-1}) + 2\E{\log \lambda_j(X_j)}} \]
and therefore that $d_{\rm sw}^{\rm HR}(\KB; \lB) \sim d^{\rm HR}(\KB;\lB) \sim d(\KB;\lB)$.
\end{IEEEproof}

\subsubsection{The Distortion-Rate Functions}
We may now establish high-resolution approximations to the distortion-rate
function under each of the three rate constraints.

\begin{lemma}
\label{lem:DistortionRateUnoptimized}
Define the fixed-rate, variable-rate, and Slepian-Wolf  high-resolution distortion-rate functions as
\begin{subequations}
\label{eq:allDistortionRate}
\beq
\Dfrhr(R;\lB) = \frac{n}{12} 2^{-2R/n}\left( \prod_{j=1}^n\E{\left(\frac{\gamma_j(X_j)}{\lambda_j(X_j)}\right)^2} \right)^{1/n}\mbox{,}
\label{eq:frDistortionRate}
\eeq
\beq
\Dvrhr(R;\lB) = \frac{n}{12} 2^{-2R/n} \left(\prod_{j=1}^n 2^{2h(X_j) + 2\E{\log \lambda_j(X_j)}}\E{\left(\frac{\gamma_j(X_j)}{\lambda_j(X_j)}\right)^2} \right)^{1/n}\mbox{,}
\label{eq:vrDistortionRate}
\eeq
\beq
\Dswhr(R;\lB) =  \frac{n}{12}2^{-2R/n} \left( 2^{2h(X_1^n)}\prod_{j=1}^n 2^{2\E{\log \lambda_j(X_j)}}\E{\left(\frac{\gamma_j(X_j)}{\lambda_j(X_j)}\right)^2} \right)^{1/n}\mbox{.}
\label{eq:jeDistortionRate}
\eeq
\end{subequations}
Then $D_{\rm fr,vr,sw}(R;\lB) \sim D^{\rm HR}_{\rm fr,vr,sw}(R;\lB)$.
\end{lemma}
\begin{IEEEproof}
See Appendix \ref{app:DistortionRateUnoptimized}.
\end{IEEEproof}

%We now prove that the approximate resolution-rate functions may be inserted
%into the approximate distortion-resolution functions and yield asymptotically
%accurate distortion-rate functions.
%
%\begin{lemma}
%If a resolution-rate approximation is valid---i.e.\ $\Kall(R;\lB) \sim \Kallhr(R;\lB)$ where * stands in for fr, vr, or sw---then the corresponding distortion-rate approximation is valid:
%\[
% D_{\rm *}(R;\lB) = d(\Kall(R;\lB); \lB) \sim D_{\rm *}^{\rm HR}(R; \lB).
%\]
%\end{lemma}
%\begin{IEEEproof}
%If $\Kall(R;\lB) \sim \Kallhr(R;\lB)$ then 
%$\Kall(R;\lB)^{-2} \sim \Kallhr(R;\lB)^{-2}$.
%Since the expression for $\dhr(K;\lB)$ has this dependency on $K$,
%we have
%$\dhr(\Kallhr(R;\lB);\lB) \sim \dhr(\Kall(R;\lB);\lB)$.
%Finally, by Theorem \ref{thm:multi-distortion}, $\dhr(\Kall(R;\lB);\lB) \sim d(\Kall(R;\lB);\lB)$
% and therefore $\dhr(\Kallhr(R;\lB);\lB) \sim d(\Kall(R;\lB);\lB)$.
%Since the first of these terms is the high-resolution approximate distortion-rate function $D_{\rm *}^{\rm HR}$
%and the second is the true distortion-rate function $D_{\rm *}$, this proves the lemma.
%\end{IEEEproof}
%
%For the fixed-rate case, inserting the approximate resolution-rate function
%$\Kfrhr(R;\lB) = 2^R$ into the expression \eqref{eq:ndUnoptimizedDistAlloc} yields an unoptimized fixed-rate distortion-rate
%function:
%\begin{subequations}
%
%For the variable-rate case, the resulting distortion-rate function is
%
%For the Slepian--Wolf case, the distortion-rate function is
%
%\end{subequations}

\subsubsection{Asymptotically Optimal Distributed Quantizers}
\label{sec:multi-densities}
The expressions \eqref{eq:allDistortionRate}
%\eqref{eq:frDistortionRate}--\eqref{eq:jeDistortionRate}
decouple the problem of designing $n$ point densities $\lB$ into $n$ separate problems of designing
a single point density $\lambda_j$.
Furthermore, each design problem
(the minimization of an expression in \eqref{eq:allDistortionRate})
%\eqref{eq:frDistortionRate}, \eqref{eq:vrDistortionRate},
%or \eqref{eq:jeDistortionRate})
is of a familiar form.
Thus we obtain the following theorem.
\begin{theorem}
 \label{thm:multi-summary}
The asymptotic fixed-rate (codebook-constrained) distortion-rate expression \eqref{eq:frDistortionRate}
is minimized by the choice
\beq
 \label{eq:multi-fixed-lambda}
  \lambda_j^*(x) = \frac{\left(\gamma_j^2(x)f_{X_j}(x)\right)^{1/3}}
                    {\int_0^1 \left(\gamma_j^2(t)f_{X_j}(t)\right)^{1/3} \, dt},
                 \quad j=1,\,2\,\,\ldots,\,n \mbox{,}
\eeq
yielding distortion
%\beq
% \label{eq:multi-fixed-dist}
% \Dfrhr(R;\aB) = \frac{n}{12} \sum_{j=1}^n \| \gamma_j^2 f_{X_j} \|_{1/3} 2^{-2\alpha_j R} \mbox{.}
%\eeq
%If $\aB$ is also optimized, the resulting distortion is given by
\beq
 \label{eq:multi-fixed-alloc-dist}
 \Dfrhr(R) = \frac{n}{12}\left( \prod_{j=1}^n
                                           \| \gamma_j^2 f_{X_j} \|_{1/3}
                        \right)^{1/n} 2^{-2R/n} \mbox{.}
\eeq
%is achieved with $\lambda_j$s given by \eqref{eq:multi-fixed-lambda} and

%\beq
% \label{eq:multi-fixed-alloc}
%  R_j = R + \Half \log \frac{ \| \gamma_j^2 f_{X_j} \|_{1/3} }
%             {\left(\prod_{k=1}^n \| \gamma_k^2 f_{X_k} \|_{1/3}\right)^{1/n} },
%                 \quad j=1,\,2\,\,\ldots,\,n\mbox{.}
%\eeq
The asymptotic variable-rate (marginal entropy-constrained) distortion-rate expression \eqref{eq:vrDistortionRate}
is minimized by the choice
\beq
 \label{eq:multi-var-lambda}
  \lambda_j^*(x) = \frac{\gamma_j(x)}
                    {\int_0^1 \gamma_j(t) \, dt},
                 \quad j=1,\,2\,\,\ldots,\,n\mbox{,}
\eeq
yielding distortion
%\beq
% \label{eq:multi-var-dist}
% \Dvrhr(R;\aB) = \frac{n}{12}\sum_{j=1}^n
%                     2^{-2\alpha_j R + 2\alpha_jh(X_j)+2\alpha_j\E{\log \gamma_j(X_j)}}
%                        \mbox{.}
%\eeq
%%With optimization of the fractional allocation $\aB$ as well,
%%the resulting distortion-rate function 
%%can be written as
%If $\aB$ is also optimized, the resulting distortion is given by
%\beq
% \label{eq:multi-var-dist}
%   D \approx \sum_{j=1}^n \frac{1}{12} \| \gamma_j \|_1^2 \, 2^{2h(X_j)+2\E{\log \gamma_j(X_j)}} \, 2^{-2R_j} \mbox{,}
%\eeq
%where $R_j = H(\Xhat_j)$ is the output entropy of $Q_j$.
%If $\sum_{j=1}^n R_j$ is fixed to $nR$ and $R$ is large enough,
%the minimum distortion
\beq
 \label{eq:multi-var-alloc-dist}
  \Dvrhr(R) = \frac{n}{12}\left( \prod_{j=1}^n
                     2^{2h(X_j)+2\E{\log \gamma_j(X_j)}}
                        \right)^{1/n} 2^{-2R/n} \mbox{.}
\eeq
%is achieved with $\lambda_j$s given by \eqref{eq:multi-var-lambda} and
%\beq
% \label{eq:multi-var-alloc}
%  R_j = R + \Half \log
%           \frac{ \| \gamma_j \|_1^2 \, 2^{2h(X_j)+2\E{\log \gamma_j(X_j)}} }
%             {\left( \prod_{k=1}^n
%                     \| \gamma_k \|_1^2 \, 2^{2h(X_k)+2\E{\log \gamma_k(X_k)}}
%                        \right)^{1/n} }\mbox{,}
%\eeq
%for $j=1,\,2\,\,\ldots,\,n$.

The asymptotic Slepian--Wolf (joint entropy-constrained) distortion-rate expression
\eqref{eq:jeDistortionRate} is optimized by a choice of point densities
identical to the variable-rate case \eqref{eq:multi-var-lambda}.  The resulting distortion is
\beq
 \label{eq:multi-je-alloc-dist}
 \Dswhr(R) = \frac{n}{12}\left( 2^{2h(X_1^n)}\prod_{j=1}^n
                     2^{2\E{\log \gamma_j(X_j)}}
                        \right)^{1/n} 2^{-2R/n} \mbox{.}
\eeq

The distributed quantizer point densities yielded by the above optimizations
%are (fixed-rate, variable-rate, Slepian--Wolf)-optimal.
are asymptotically optimal.
\end{theorem}
\begin{IEEEproof}
To prove \eqref{eq:multi-fixed-lambda} gives the optimal point density
for fixed-rate coding and
\eqref{eq:multi-var-lambda} gives the optimal point density for both
variable-rate and Slepian--Wolf coding,
it suffices to note that minimizing the $n$ terms in
\eqref{eq:frDistortionRate}, \eqref{eq:vrDistortionRate}, and \eqref{eq:jeDistortionRate} separately
gives problems identical to those in Section~\ref{sec:Single}.

%Finding the optimal choice of $\alpha_1^n$ is precisely addressed by Lemma~\ref{lem:bit-alloc};
%this yields \eqref{eq:multi-fixed-alloc-dist}, \eqref{eq:multi-var-alloc-dist}, 
%and \eqref{eq:multi-je-alloc-dist}.

The proof that the choice of $\lB$ that minimizes the high-resolution expression
is asymptotically optimal
is virtually identical to that of Lemma \ref{lem:optimizationIsLegit}, so
it is omitted.
\end{IEEEproof}

\subsection{Variation: Joint Entropy Constraint}
\label{sec:slepian-wolf}
Distortion expressions \eqref{eq:multi-fixed-alloc-dist}
and \eqref{eq:multi-var-alloc-dist} are minimum distortions subject
to a sum-rate constraint.
The individual rates given by $R_j = \log K_j$ (fixed-rate) or by \eqref{eq:1drate}
(variable-rate)
implicitly specify no entropy coding or separate entropy coding
of the $\Xhat_j$s, respectively.

If the $\Xhat_j$s are not independent---which is anticipated whenever the $X_j$s
are not independent---one may employ Slepian--Wolf coding
of the $\Xhat_j$s without violating the distributed coding requirement
implicit in Fig.~\ref{fig:dfsc-block}.
This lowers the total rate from $\sum_{j=1}^n H(\Xhat_j)$ to
$H(\Xhat_1,\,\Xhat_2,\,\ldots,\Xhat_n)$ and changes the marginal entropy
constraint into a joint entropy constraint.  While the optimal compander
choice \eqref{eq:multi-var-lambda} is unchanged by this modification, the resulting
distortion-rate function reduces from \eqref{eq:multi-var-alloc-dist} 
to \eqref{eq:multi-je-alloc-dist}.

Some remarks:
\begin{enumerate}
\item
By comparing \eqref{eq:multi-je-alloc-dist} to \eqref{eq:multi-var-alloc-dist},
we see that the inclusion of Slepian--Wolf coding has reduced the
sum rate to achieve any given distortion by
$$
  \left(\sum_{j=1}^n h(X_j)\right) - h(X_1^n) \mbox{.}
$$
This is, of course, not unexpected as it represents the excess information
in the product of marginal probability distributions as compared to the joint probability distribution.
This has been termed the \emph{multiinformation}~\cite{StudenyV1998}
and equals the mutual information when $n=2$.
\item
While the resolution allocation $\KB$ amongst the $n$ sources has a unique
minimizing choice, there is some flexibility in rate allocations 
for the Slepian--Wolf encoder.  Any point on the Slepian--Wolf 
joint-entropy boundary may be achieved with arbitrarily low probability of error.
%VM: I don't think the below is true, and I'm wondering why I didn't notice before.  Slepian--Wolf is a lossless 
%encoding/decoding scheme that has flexibility built in.  The quantizer is unaffected.

%The optimal $\Rtilde_j$s and resulting distortion $D$ would not be
%affected, but the $K_j$s would change.
%One can interpret this as a flexibility in resolution allocation
%(slightly distinct from bit allocation) that can be used to control
%inaccuracies due to the high-resolution approximations.

\item
The theorem seems to analytically separate correlations among sources from
functional considerations, exploiting correlation even though the
quantizers are regular.
In reality, the binning introduced by Slepian--Wolf coding transforms the
scalar quantizers of each source component into nonregular vector quantizers
so as to remove redundancy between sources.
\end{enumerate}

\subsection{Relationship to Locally-Quadratic Distortion Measures}
\label{sec:locally-quadratic}
Linder \emph{et al.} consider the class of ``locally-quadratic''
distortion measures for variable-rate high-resolution quantization
in~\cite{LinderZZ1999}.
They define locally-quadratic measures as those having the following two properties:
\begin{enumerate}
\item Let $x$ be in $\R^n$. For $y$ sufficiently close to $x$ in the Euclidean
metric,
the distortion between $x$ and $y$ is well approximated by
$\sum_{i=1}^n M_i(x)|x_i-y_i|^2$, where $M_i(x)$ is a positive scaling factor.
In other words, the distortion is a space-varying non-isotropically scaled MSE\@.
\item The distortion between two points is zero if and only if the points are identical.
\end{enumerate}

\noindent For these distortion measures, the authors consider high-resolution 
variable-rate regular quantization, generalize Bucklew's
results \cite{Bucklew1984} to non-functional distortion measures, and demonstrate
the use of multidimensional companding functions to implement
these quantizers.  Of particular interest is the comparison
they perform between joint vector quantization and separable scalar quantization.  When Slepian--Wolf
coding is employed for the latter, the scenario is similar to the
developments of this section.

The source of this similarity is the implicit distortion measure we work with:
$d_g(x,y) = |g(x)-g(y)|^2$.  When $x$ and $y$ are very close to each other,
Taylor approximation reduces this error to a quadratic form:
\[ 
|g(x)-g(y)|^2 \approx \sum_{i=1}^n \left|\frac{\del g(x_1^n)}{\del x_i}\right|^2 |x_i-y_i|^2.
\]
From this, one may obtain the same variable-rate Slepian--Wolf
performance as \eqref{eq:multi-je-alloc-dist} through the analysis in
\cite{LinderZZ1999}.

However, there are important differences
between
locally-quadratic distortion measures and the functional distortion measures
we consider. First and foremost: a continuous scalar function of $n$
variables, $n > 1$, is \emph{guaranteed} to have an uncountable number of pairs
$x \neq y$ for which $g(x) = g(y)$ and therefore that $d_g(x,y) = 0$.
This violates the second condition of a locally-quadratic distortion measure,
and the repercussions are felt most strikingly for non-monotonic
functions---those for which regular quantizers are not necessarily optimal
(see Section~\ref{sec:NonMono}).

The second condition is also violated by functions that are
not \emph{strictly} monotonic in each variable; one finds that without
strictness, variable-rate analysis of the centralized encoding problem
is invalidated.  Specifically, if the derivative vector 
\[
\left(\frac{\del g(x_1^n)}{\del x_1},\,
                  \frac{\del g(x_1^n)}{\del x_2},\,
                  \ldots, \,
                  \frac{\del g(x_1^n)}{\del x_n}\right)
\]
has nonzero probability of possessing a zero component, the expected
variable-rate distortion as derived by both Bucklew and Linder \emph{et al.}\
is $D = 0$, regardless of rate. 
This answer arrives from the null derivative having violated
the high-resolution approximation, and it implies that the distortion
falls faster than $2^{-2 R/n}$.
In future work, generalizations of our results in
Section~\ref{sec:DontCare} may be able to address such deficiencies.

\section{Examples}
\label{sec:Scaling}
Before moving on to extensions of the basic theory,
we present a few examples to show how optimal ordinary scalar quantization and
optimal DFSQ differ.
We especially want to highlight a few simple examples in which performance
scaling with respect to $n$ differ greatly between ordinary and
functionally-optimized quantization.
To draw attention to this scaling, we define the \emph{rate-per-source}
$\RS$ as the sum-rate divided by the number of sources $R/n$, 
and hold this quantity constant as the number of sources increases.

\smallskip

\begin{example}[Linear function]
Consider the function $g(x_1^n) = \sum_{j=1}^n a_j x_j$
where the $a_j$s are scalars.
Then for any $j$, $\gamma_j(x) = |a_j|$.
Since $\gamma_j(x)$ does not depend on $x$,
it has no influence on the optimal point density
for either the fixed- or variable-rate case;
see \eqref{eq:multi-fixed-lambda} and \eqref{eq:multi-var-lambda}.

Although $\gamma_j(x)$ gives no information on which values of $X_j$
are more important than others
(or rather shows that they are all equally important)
the set of $\gamma_j$s shows the relative importance of the components.
This is reflected in the allocation of rate.
\hfill $\Box$
\end{example}

\smallskip

\begin{example}[Maximum]
\label{ex:maximum}
Let the set of sources $X_1^n$ be uniformly distributed on $[0,1]^n$
and hence mutually independent.
Consider the function
$$g(x_1^n) = \max(x_1,\,x_2,\,\ldots,\,x_n).$$
Note that this function is differentiable outside the sets 
$A_{i,j} = \{x_1^n: x_i = x_j\}$, where $i,j \in \{1,\ldots, n\}$.
Each $A_{i,j}$ is an $(n-1)$-dimensional plane and therefore
has Jordan measure zero, and since a finite union of Jordan-measure-zero
sets has Jordan measure zero, condition MF1 is satisfied.
Though very simple, this function is more interesting than a linear function
because the derivative with respect to one variable depends sharply on all
the others.  
The function is symmetric in its arguments, so for
notational convenience consider only the design of the quantizer for $X_1$.

The partial derivative $g_1(x_1^n)$ is 1 where the maximum is $x_1$
and is 0 otherwise.
Thus,
\begin{eqnarray*}
  \gamma_1^2(x) & = & \E{ |g_1(X_1^n)|^2 \mid X_1 = x } \\
     & = & \P{ \max(X_1^n) = X_1 \mid X_1 = x } \\
     & = & x^{n-1} \mbox{,} 
\end{eqnarray*}
where the final step uses the probability of all $n-1$ variables
$X_2^n$ being less than $x$.

The optimal point density for fixed-rate quantization is found by
evaluating \eqref{eq:multi-fixed-lambda} to be
$$
  \lambda_1(x) = {\textstyle\frac{1}{3}} (n+2) x^{(n-1)/3} \mbox{.}
$$
The resulting distortion when each quantizer has rate $\RS$ (equal rate
allocations) is found
by evaluating \eqref{eq:multi-fixed-alloc-dist} to be
\begin{eqnarray*}
  \Dfrhr(n\RS) & = & \frac{n}{12} \, \| \gamma_1^2 \|_{1/3} \, 2^{-2\RS} 
    \ = \ \frac{n}{12} \, \left(\frac{3}{n+2}\right)^3 \, 2^{-2\RS} \\
    & = & \frac{9n}{4(n+2)^3} \, 2^{-2\RS} \mbox{.}
\end{eqnarray*}

The optimal point density for variable-rate quantization is found by
evaluating \eqref{eq:multi-var-lambda} to be
$$
  \lambda_1(x) = \half(n+1)x^{(n-1)/2} \mbox{.}
$$
Substituting
$h(X_1) = 0$ and
$2^{2\E{\log \gamma_1(X_1)}} = e^{-n+1}$
into \eqref{eq:multi-var-alloc-dist} gives
$$
  \Dvrhr(n\RS) = 
     \frac{n}{12} e^{-n+1} \, 2^{-2\RS}\mbox{.}
$$

The two computed distortions decrease sharply with $n$.
This is in stark contrast to the results of ordinary quantization.
When functional considerations are ignored, one optimally uses a
uniform quantizer, resulting in
$\iE{(X_j-\Xhat_j)^2} \approx {\textstyle\frac{1}{12}}2^{-2R_j}$
for any component.
Since the maximum is equal to one of the components,
the functional distortion is
$\Dord^{\rm HR}(n\RS) = {\textstyle\frac{1}{12}}2^{-2\RS}$, unchanging with $n$.

The optimal point densities
computed above are shown
in Fig.~\ref{fig:max-lambdas}.
The distortions are presented along with the results of the
following example in Fig.~\ref{fig:example-dists}.
\hfill $\Box$
\end{example}

\begin{figure*}
 \centering
  \begin{tabular}{cc}
   {
    \psfrag{x}[][]{$x$}
    \psfrag{lambda(x)}[][t]{$\lambda(x)$}
    \includegraphics[width=\widthA]{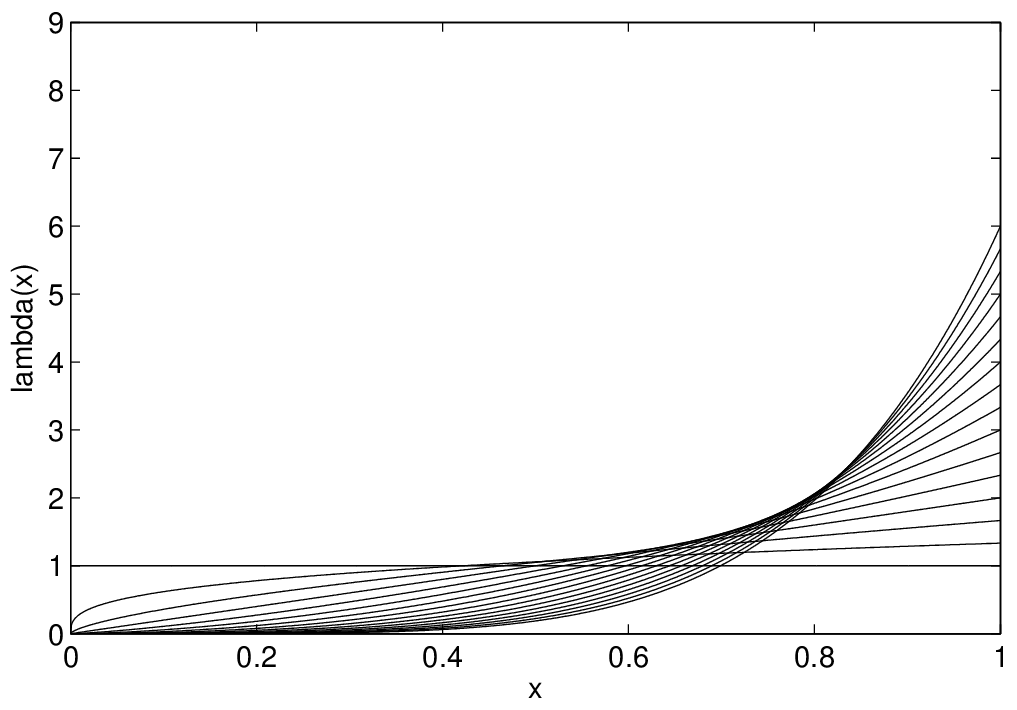}
   } &
   {
    \psfrag{x}[][]{$x$}
    \psfrag{lambda(x)}[][t]{$\lambda(x)$}
    \includegraphics[width=\widthA]{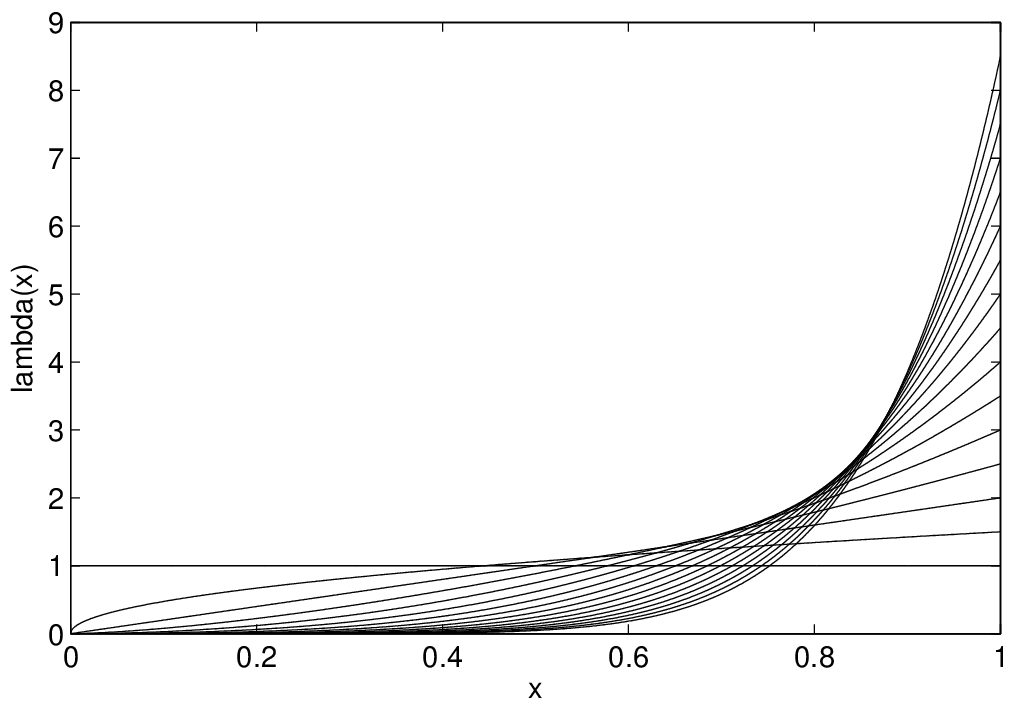}
   } \\
   {\small (a) fixed-rate} &
   {\small (b) variable-rate}
  \end{tabular}
  \caption{Optimal point densities for Example~\ref{ex:maximum} (maximum),
    $n = 1,\,2,\,\ldots,\,16$.  As $n$ increases, the sensitivities
    $\gamma_j(x)$ become more unbalanced toward large $x$;
    this is reflected in the point densities,
    more so in the variable-rate case than in the fixed-rate case.}
  \label{fig:max-lambdas}
\end{figure*}

\begin{figure}
 \centering
  {
   \psfrag{n}[][]{$n$}
   \psfrag{D}[][t]{$D \cdot 12 \cdot 2^{2\RS}$}
   \psfrag{a}[][]{\small $\Dvrhr$, max}
   \psfrag{b}[][]{\small $\Dvrhr$, med}
   \psfrag{c}[][]{\small $\Dfrhr$, max}
   \psfrag{d}[][]{\small $\Dfrhr$, med $\;\;\;$}
   \includegraphics[width=\widthA]{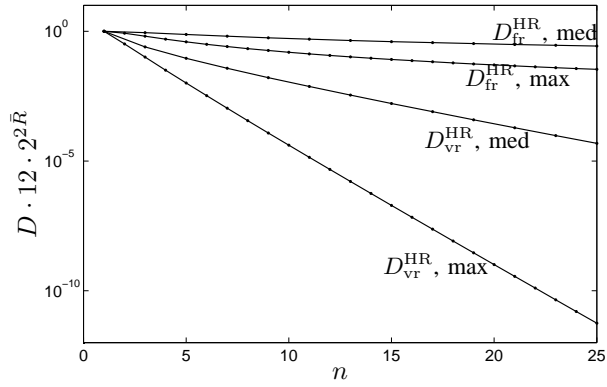}
  }
  \caption{Distortions of optimal fixed- and variable-rate functional
    quantizers for maximum and median functions from Examples~\ref{ex:maximum}
    and~\ref{ex:median}.  Shown is the dependence
    on the number of variables $n$; by plotting $D \cdot 12 \cdot 2^{2\RS}$
    we see the performance relative to ordinary quantization.}
  \label{fig:example-dists}
\end{figure}

\smallskip

\begin{example}[Median]
\label{ex:median}
Let $n = 2m+1$, $m \in \N$,
and again let the set of sources $X_1^n$ be uniformly distributed on $[0,1]^n$.
The function
$$g(x_1^n) = \median(x_1,\,x_2,\,\ldots,\,x_n)$$
provides a similar but more complicated example.  Note that, as
in Example \ref{ex:maximum}, this function is differentiable outside
the zero-Jordan-measure sets $A_{i,j}$, and it therefore satisfies
condition MF1.

The partial derivative $g_1(x_1^n)$ is 1 where the median is $x_1$
and is 0 otherwise.
Thus,
\begin{eqnarray*}
  \gamma_1^2(x) & = & \E{ |g_1(X_1^n)|^2 \mid X_1 = x } \\
     & = & \P{ \median(X_1^n) = X_1 \mid X_1 = x } \\
     & = & {2m \choose m} x^m (1-x)^m\mbox{,} 
\end{eqnarray*}
where the final step uses the binomial probability for the event of
exactly $m$ of the $2m$ variables $X_2^n$ exceeding $x$.

The optimal point density for fixed-rate quantization is found by
evaluating \eqref{eq:multi-fixed-lambda} to be
$$
  \lambda_1(x) = \frac{ x^{m/3}(1-x)^{m/3} }
                     { B(m/3+1,m/3+1) }
$$
where $B$ is the beta function.
The resulting distortion when each quantizer has rate $\RS$ is found
by evaluating \eqref{eq:multi-fixed-alloc-dist} to be
\begin{eqnarray*}
  \Dfrhr(n\RS) & = & \frac{2m+1}{12} \, \| \gamma_1^2 \|_{1/3} \, 2^{-2\RS} \\
    & = & \frac{2m+1}{12} \, {2m \choose m} \left( B\left(\frac{m}{3}+1,\frac{m}{3}+1\right) \right)^3 \, 2^{-2\RS} \mbox{.}
\end{eqnarray*}
To understand the trend for large $m$, we can substitute in the Stirling
approximations ${2m \choose m} \sim (m\pi)^{-1/2}2^{2m}$
and
$$B(m/3+1,m/3+1) \sim \sqrt{6\pi/m} \, 2^{-(2m/3+3/2)}$$
to obtain
$$
  \Dfrhr(n\RS) \sim \frac{m}{6} \, \frac{2^{2m}}{\sqrt{m\pi}} \, \left(\frac{6\pi}{m}\right)^{3/2} 2^{-(2m+9/2)} \, 2^{-2\RS}
    = \frac{\pi \sqrt{3}}{16m} \, 2^{-2\RS} \mbox{.}
$$

The optimal point density for variable-rate quantization is found by
evaluating \eqref{eq:multi-var-lambda} to be
$$
  \lambda_1(x) = \frac{ x^{m}(1-x)^{m} }
                      { B(m+1,m+1) } \mbox{.}
$$
To evaluate the resulting distortion, note that
%$$\|\gamma_1\|_1^2 = {2m \choose m}\left(B\left(\frac{m}{2}+1,\frac{m}{2}+1\right)\right)^2\mbox{,}$$
$h(X_1) = 0$ and
$2^{2\E{\log \gamma_1(X_1)}} = {2m \choose m}e^{-2m}$.
Substituting into \eqref{eq:multi-var-alloc-dist} gives
$$
  \Dvrhr(n\RS) = 
     \frac{2m+1}{12} {2m \choose m} e^{-2m} \, 2^{-2\RS}\mbox{.}
$$
Using the approximation above for the binomial factor
we obtain
\beqan
  \Dvrhr(n\RS) & \sim & \frac{m^{1/2}}{6\pi^{1/2}} \left(\frac{e}{2}\right)^{-2m} \, 2^{-2\RS}  \mbox{.}
\eeqan

The optimal point densities computed above are shown in
Fig.~\ref{fig:median-lambdas}.
The distortions are presented along with the results of
Example~\ref{ex:maximum} in Fig.~\ref{fig:example-dists}.

\begin{figure*}
 \centering
  \begin{tabular}{cc}
   {
    \psfrag{x}[][]{$x$}
    \psfrag{lambda(x)}[][t]{$\lambda(x)$}
    \includegraphics[width=\widthA]{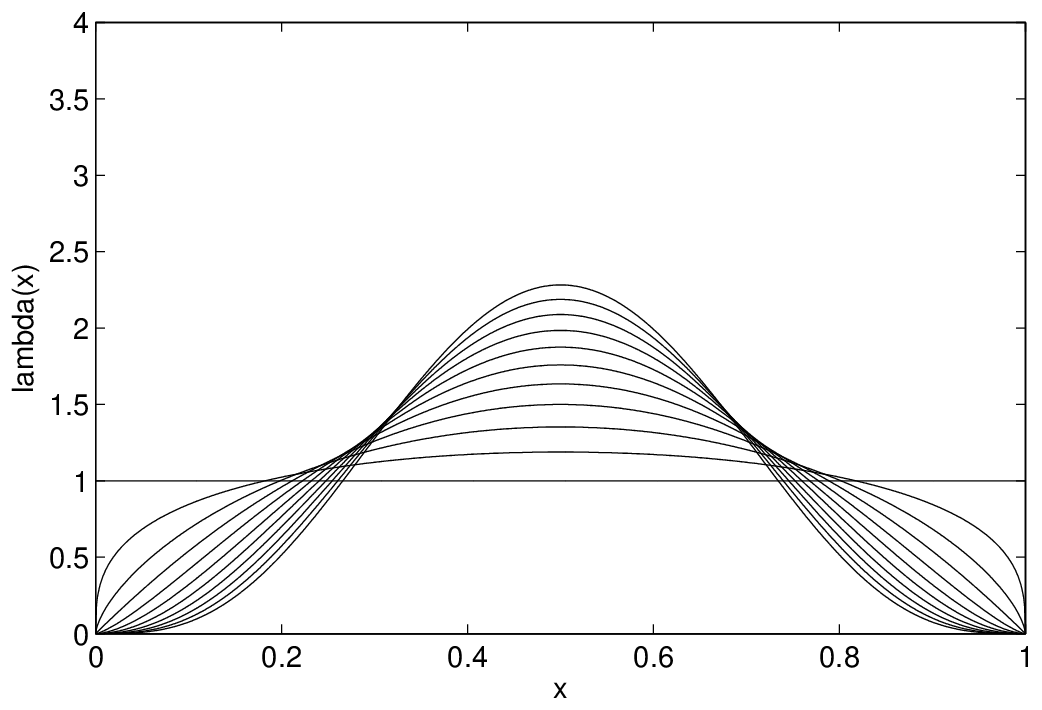}
   } &
   {
    \psfrag{x}[][]{$x$}
    \psfrag{lambda(x)}[][t]{$\lambda(x)$}
    \includegraphics[width=\widthA]{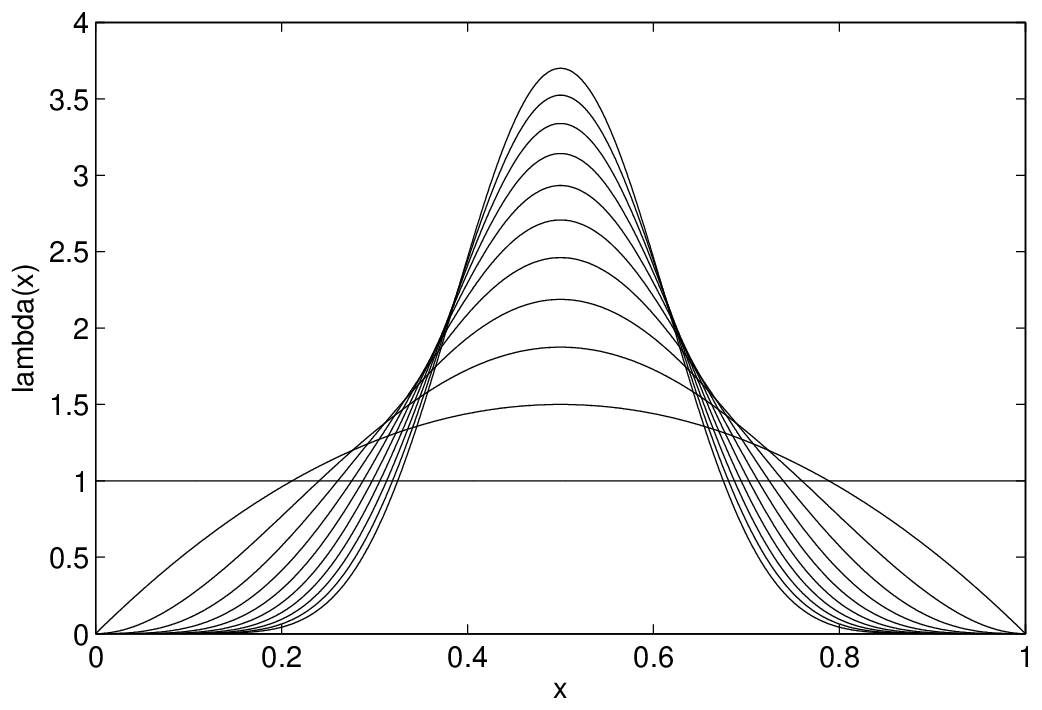}
   } \\
   {\small (a) fixed-rate} &
   {\small (b) variable-rate}
  \end{tabular}
  \caption{Optimal point densities for Example~\ref{ex:median} (median),
    $n = 1,\,3,\,\ldots,\,21$.  As $n$ increases, the sensitivities
    $\gamma_j(x)$ become more unbalanced toward $x = 1/2$;
    this is reflected in the point densities,
    more so in the variable-rate case than in the fixed-rate case.}
  \label{fig:median-lambdas}
\end{figure*}

Note the following similarities to Example~\ref{ex:maximum}:
$\Dord^{\rm HR}$ is constant with respect to $n$,
$\Dfrhr$ decays polynomially with $n$,
and $\Dvrhr$ decays exponentially with $n$.
\hfill $\Box$
\end{example}

\smallskip

The large performance improvement over ordinary quantization
in these examples illustrates the potential
benefits of functional quantization.  Additional examples and details appear in~\cite{Misra2008}.

\section{Non-Monotonic Functions and Non-Regular Quantization}
\label{sec:NonMono}
%Proposed (new) structure
%"'Optimality of regular quantizers and regular companders."'
%description of section
%Regularity, monotonic function in each variable.
	%Conditions for regularity to be optimal: fixed-rate.
	%Conditions for regularity to be optimal: variable-rate.
%Non-regular companding quantization.
%The optimality of regular companding quantization.

The high-resolution approach to quantizer optimization is inherently
limited to the design of regular quantizers.  In particular,
%assumptions UO3, UF4, and MF3 all restrict attention to companding functions, 
we have specified compander functions to be monotonic in Section~\ref{sec:Definitions}.
The analysis of Section~\ref{sec:Multi} therefore
gave us quantizer sequences within the class of regular quantizers. 

In this section we explore less restrictive alternatives to the 
monotonicity requirement.  Specifically, we introduce the concept
of \emph{equivalence-free} and show that if a function has this
property, then non-regular companding quantizer sequences are
asymptotically suboptimal.

Fig.~\ref{fig:NonEquivalent} illustrates the concept.
The function on the left is aligned with the axes in the sense that
$g(x_1,x_2)$ depends only on $x_1$.  Since the dependence on $x_1$ is
not monotonic, there are pairs of distinct points $(x_1^{\dagger},x_1^{\ddagger})$
where $g(x_1^{\dagger},x_2)=g(x_1^{\ddagger},x_2)$ and thus the optimal quantizer at
high enough resolution has $Q_1(x_1^{\dagger}) = Q_1(x_1^{\ddagger})$, giving a non-regular
quantizer.
When the argument vector $(x_1,x_2)$ of the function is rotated as shown on the right,
the resulting function is still non-monotonic.  However,
there is no longer a clearly optimal non-regular quantization scheme.   % with respect to both $x_1$ and $x_2$,
Specifically, for some fixed $x_2$ there may be pairs $(x_1^{\dagger},x_1^{\ddagger})$
such that $g(x_1^{\dagger},x_2)=g(x_1^{\ddagger},x_2)$, but the equality does not
hold for all $x_2$.  As we shall see, this results in the suboptimality
of any compander that maps $x_1^{\dagger}$ in the same way as $x_1^{\ddagger}$.

\begin{figure}
 \centering
  \psfrag{x1}[b][][1][0]{\small $x_1$}
  \psfrag{x2}[b][][1][0]{\small $x_2$}
  \psfrag{g}[b][][1][0]{\small $g(x_1,x_2)$}
  \includegraphics[width=0.45\textwidth]{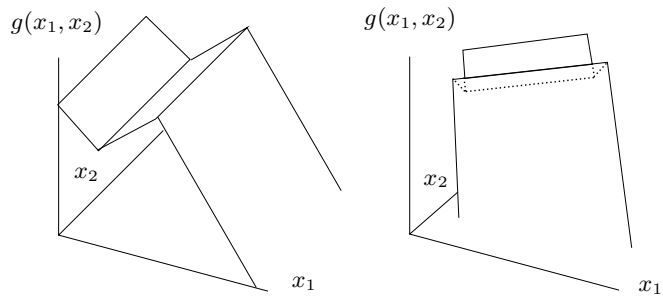}
  \caption{Two functions of two variables are shown. 
    The left function is separable and $X_1$ is best quantized by a
    non-regular quantizer;
    for the right function (a rotated version of the left),
    a regular quantizer is asymptotically optimal. 
    This is due to the right function being ``equivalence-free.''}
  \label{fig:NonEquivalent}
\end{figure}

Our approach is to first create a model for high-resolution
non-regular quantization, then to use this model to expand the class
of functions for which regular quantization is optimal, and
finally to construct asymptotically optimal non-regular quantizers
when regularity is suboptimal.

\subsection{High-Resolution Non-Regular Quantization}
\label{sec:NonRegular}
To accommodate non-regular quantization, we extend the compander-based model
of quantization. 
%Companding is to implement a non-uniform quantizer as $w^{-1}(q(w(x)))$ where
%$w$ is a \emph{compander}, $q$ is a uniform quantizer, and $w^{-1}$ is an
%\emph{expander}.
%The reader is referred to~\cite{GrayN1998} for additional details and
%references to original sources.
In Bennett's development of optimal companding, reviewed in 
Sec. \ref{sec:Background}, it is natural to require
$w$ to be both monotonic and have a bounded derivative everywhere;
the derivative $w'(x)$ is proportional to the quantizer point density
$\lambda(x)$ that has been central in our development thus far.
Whether we look at $\lambda$ or $w$, the role is to set the relative
sizes of the quantization cells.

Since optimal functional quantizers are not necessarily regular,
we adapt the conventional development to implement non-regular quantizers.
\begin{definition} A function $w: [0,1]\rightarrow[0,1]$ is a 
\emph{generalized compander} if it is
continuous, piecewise monotonic with a finite number of pieces, and has
bounded derivative over each piece.
\end{definition}

As in regular companding, $w$ and $w^{-1}$ are used along with a
uniform quantizer $Q_K^U$ as $w^{-1}(Q_K^U(w(x)))$.
The restriction to a finite number of pieces is a limitation on the types of
non-regular quantizers that can be captured with this model: those for which
every quantizer cell is a finite union of intervals. 
Barring certain pathological situations, this restriction is reasonable.

Along with setting relative sizes of cells, $w$ provides for non-regularity
by allowing intervals to be binned together.
To illustrate this, consider a simple example. 
Suppose that the pair $(X_1,X_2)$ is uniformly distributed over $[0,1]^2$,
variable rate quantization is to be performed on both variables, and the
function of interest is defined by
\[
g(x_1,x_2) = x_1(\tfrac{3}{4}-x_1)(1-x_2)\mbox{.}
\]

An optimal functional quantizer---a quantizer for $X_1$ to minimize
$\iE{(g(X_1,X_2)-g(\Xhat_1,\Xhat_2))^2}$---should bin together
$X_1$ values that always yield the same $g(X_1,X_2)$.  Furthermore,
the magnitude of the slope of this quantizer should follow \eqref{eq:multi-var-lambda}.
The choice of 
\[
w_1(x_1) = \tfrac{64}{25}x_1\left(\tfrac{3}{4}-x_1\right)+\tfrac{16}{25}
\]
can be shown to be optimal.  Both $w_1$ and the
resulting quantizer at resolution $K=5$ are illustrated
in Fig.~\ref{fig:NonMonoExample}b.
%
%It can be seen from
%a plot of $g(x_1,x_2)$ (Fig. \ref{fig:NonMonoExample}a) that the segment
%$X_1 \in [0,3/8]$ is identical in this respect to the segment $X_1 \in [3/8,3/4]$.
%This yields the constraint $w_1(x_1) = w_1(3/4-x)$ for $x \in [0,3/4]$.
%Furthermore, \eqref{eq:multi-var-lambda} sets the magnitude of the slope of $w_1$ 
%in relation to the expected magnitude of the slope of $g$:
%$$
%  |w_1'(x)| \propto \frac{3}{4}-2x_1 \mbox{.}
%$$
%This still leaves many choices for $w_1$, the most obvious being $w_1 = g(x_1,0)$.
%A shifted and normalized version of this choice, along with the resulting quantizer,
% is illustrated in Fig.~\ref{fig:NonMonoExample}b.
\begin{figure*}
 \centering
  \begin{tabular}{cc}
   {
    \psfrag{x1}[][]{{\footnotesize $x_1$}}
    \psfrag{x2}[][t]{{\footnotesize $x_2$}}
    \psfrag{g}[][t]{{\footnotesize $g(x_1,x_2)$}}
    \includegraphics[width=\widthA]{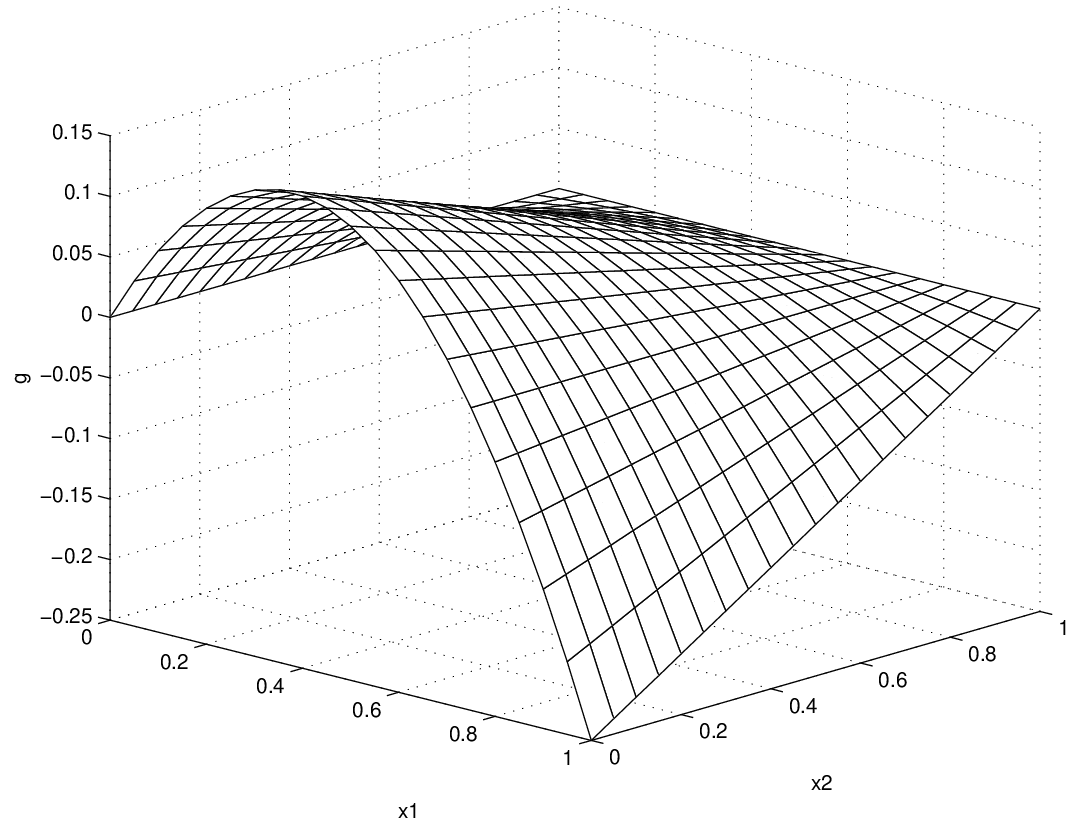}
   } &
   {
    \psfrag{0}[][]{{\footnotesize $0$}}
    \psfrag{1}[][]{{\footnotesize $1$}}
    \psfrag{w}[][]{{\footnotesize $w_1(x_1)$}}
    \psfrag{x}[][]{{\footnotesize $x_1$}}
    \includegraphics[width=\widthA]{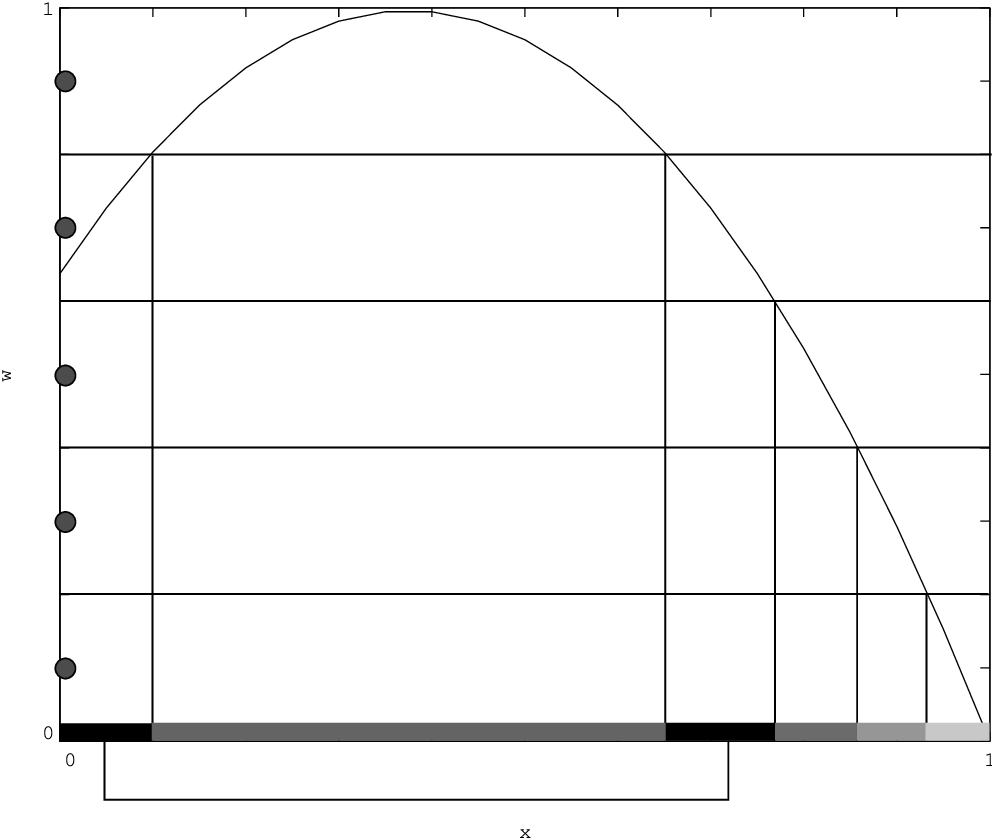}
   } \\
   {\small (a) Function of interest} &
   {\small (b) Generalized compander}
  \end{tabular}
  \caption{Example of a generalized compander $w_1(x_1)$ for a function
  $g(x_1,x_2)$ and the partition resulting from uniform quantization of $w_1(X_1)$.
    Notice that the compander dictates both the relative sizes of cells
    and the binning of intervals of $X$ values.}
  \label{fig:NonMonoExample}
\end{figure*}

\subsection{Equivalence-Free Functions}
\label{sec:equiv-free}
We now define a broad class of functions for which
regular quantization is optimal at sufficiently high resolutions. 
Consider the design of the $j$th quantizer in an $n$-dimensional
distributed functional quantization setting.

We require a set of definitions:

\begin{definition}
For any $s \neq t$ in the support of $X_j$, let
$$
v_j(s,t)
  = \E{ \var\left( g(X_1^n) \mid X_j \in \{s,t\}, \, \{X_i\}_{i \neq j} \right)
        } \mbox{.}
$$
If $v_j(s,t) = 0$ then $(s,t)$ is a
\emph{functional equivalence in the $j$th variable}.
If $g$ has no functional equivalences in any of its variables,
we say it is \emph{equivalence-free}.
\end{definition}

The theorem below demonstrates that
for DFSQ with an equivalence-free function,
quantizer regularity is necessary for asymptotic optimality.
Specifically,
strictly non-regular quantization is shown to introduce a nonzero lower bound on
the distortion, independent of rate.  This is formalized with the
aid of generalized companding.  To simplify the proof somewhat,
we assume that the marginal probability density $f_j(X_j)$ is
nonzero over $[0,1]$.  This assumption is without loss of generality,
since one may consider the subset of $[0,1]$ where $f_j(X_j)$ is
nonzero.

\begin{theorem} \label{thm:EquivalenceFree}
Let $g$ be equivalence-free with respect to the pdf of
$X_1^n$ on $[0,1]^n$.
Suppose quantization of each $X_j$ is performed as
$\Yhat_j = q(w_j(X_j))$ where $w_j$ is a generalized compander and
$q$ is a uniform quantizer.
If there is an index $j$, set $S \subset [0,1]$,
and function $t : \R \rightarrow \R$ such that
$\P{X_j \in S} > 0$, and,
for every $s \in S$,
$s \neq t(s)$ and
$w_j(s) = w_j(t(s))$,
then the distortion has a positive, resolution-independent lower bound.
\end{theorem}
\begin{IEEEproof}
See Appendix~\ref{app:equivalencefree}.
\end{IEEEproof}

The positive, rate-independent lower bound shows that the quantizer
is suboptimal if the rate is sufficiently high;
even naive uniform quantization will yield distortion with $O(2^{-2\RS})$
dependence on rate and
thus will eventually outperform the strictly non-regular quantizer.

When a function has equivalences,
the best asymptotic quantization tactic is to design companders
that bin all the equivalent values in each variable but are otherwise monotonic.
In effect, this procedure losslessly converts the function into 
one that is equivalence-free.  One might consider this a real-valued-source analogue
of the functional compression procedure suggested by Doshi \emph{et al.}~\cite{DoshiSMJ2007}.

\section{Don't-Care Intervals and Rate Amplification}
\label{sec:DontCare}
Ordinary high-resolution analysis produces point-density functions
that reflect the source pdf in the sense that
optimal quantizers never have zero point density where there is
nonzero probability density.  In fact, having zero point density
where there is nonzero probability density invalidates high-resolution analysis.
The situation is more complicated in the functional setting since the optimal
point densities depend on both the functional sensitivity profiles
and the source probability distribution.
%As foreshadowed by the qualifications in Theorem~\ref{thm:multi-summary},
Having zero functional sensitivity where the probability density is
nonzero changes the optimal quantizers in the variable-rate case.

The following example illustrates the potential for failure of the
analysis of Section~\ref{sec:multi-densities}.
Note that the intricacies arise even with a univariate function.

\begin{example}
\label{ex:dontcare}
Let $X$ have the uniform probability distribution over $[0,1]$, and 
suppose the function of interest is $g(X) = \min(X,1/2)$.
It is clear that the optimal quantizer (for both fixed- and variable-rate)
has uniform point density on $[0,1/2]$.
With the functional sensitivity profile given by
\[ \gamma(x) = \left\{ \begin{array}{ll}
                   1, & \mbox{if $x < 1/2$}; \\
                   0, & \mbox{otherwise}\mbox{,}
                       \end{array} \right.
\]
evaluating \eqref{eq:1dFixedLambda} and \eqref{eq:1dOptimalVariableLambda}
is consistent with the intuitive result.

The distortion for the fixed-rate case obtained from \eqref{eq:1dFixedDist}
is $(1/12)(1/2)^3 2^{-2R}$.  This is sensible since for half of the source values
($X > 1/2$) there is zero distortion by having a single codeword at $1/2$, whereas
for the other half of the source values ($X < 1/2$), $2^R - 1$ codewords
quantize a random variable uniformly distributed over $[0,1/2]$. 
However, assumption MF4 is not satisfied by this quantizer point density,
so it is unclear whether this expression is an asymptotically valid approximation for
the distortion-rate function.
% Furthermore, because 
%$(1/12)(1/2)^3 2^{-2R}$ is finite, assumption MF6 is satisfied and this
%expression is an asymptotically accurate approximation to the distortion-resolution function.

The variable-rate case is also problematic.
Since $\E{\log \gamma(X)} = -\infty$, 
evaluating \eqref{eq:1dVariableDist} yields $\Dvrhr = 0$.
Both the distortion-resolution and resolution-rate
analyses fail because the
quantization is not fine over the full support of $f_X$.
However, if an alternative quantization structure is used,
the distortion-rate performance can be accurately determined.
In this alternative structure, the first representation
bit specifies the event $A = \{X < 1/2\}$ or its complement.
Since additional bits are useful only when $A$ occurs,
one can spend $2(R-1)$ bits in those cases to have an average expenditure
of $R$ bits.
The resulting distortion is
\begin{eqnarray*}
  \Dvrhr & = & \P{A} d_{g|A}^{\rm HR} + \P{A^c} d_{g|A^c}^{\rm HR} \\
    & = & \half \cdot {\ts\frac{1}{12}}(\half)^2 2^{-2(2R-2)} + \half \cdot 0
    \ = \ {\ts\frac{1}{6}} 2^{-4R}\mbox{.}
\end{eqnarray*}
Note that the exponent in the distortion--rate relationship
is larger than it was in the fixed-rate case.
\hfill $\Box$
\end{example}

In the example, there is an interval $X \in [1/2, 1]$ of source values
that need not be distinguished for function evaluation.
Let us define a term for such intervals before discussing the example
further.

\begin{definition}
An interval $Z \subset [0,1]$ is called a \emph{don't-care interval}
for the $j$th variable when the $j$th functional sensitivity $\gamma_j$
is identically zero on $Z$, but the probability $\P{X_j \in Z}$ is positive.
\end{definition}

In univariate FSQ, at sufficiently high rates,
each don't-care interval corresponding to a distinct value of
the function should be allotted one codeword.
This follows from reasoning similar to that given in
Section~\ref{sec:equiv-free} and is illustrated by Example~\ref{ex:dontcare}.
In the fixed-rate case, the don't-care intervals simply occupy a few of
the $2^R$ codewords and have a limited effect.
In the variable-rate case, however, the don't-care intervals
produce a subset of source values that can be allotted very little rate.
This gives more rate to be allotted outside the don't-care intervals
and behavior we refer to as \emph{rate amplification}.

We derive the high-resolution distortion-resolution function for this quantizer structure
in Section~\ref{sec:dontcare-fixed}, and in section~\ref{sec:dontcare-variable}
the distortion-rate function is obtained.

\subsection{The Distortion-Resolution Function}
\label{sec:dontcare-fixed}
In the following analysis we will assume that the $j$th variable
has a finite number $M_j$ of don't-care intervals
$\{Z_{j,1},\,Z_{j,2},\,\ldots,\,Z_{j,{M_j}}\}$.
We also assume
\beq
  \label{eq:dontcare-probability}
  \P{ X_j \in Z_j } < 1
\quad
  \mbox{for }j = 1,\,2\,\ldots,\,n\mbox{,}
\eeq
where $Z_j = \cup_{i=1}^{M_j} Z_{j,i}$ denotes the union of don't-care
intervals for the $j$th variable.
Without this, there is no improvement beyond $M_j$ levels in representing $X_j$,
so the high-resolution approach is wholly inappropriate. 
We will denote the event $X_j \notin Z_j$ by $A_j$.

%The optimal operational distortion--resolution expression
%\eqref{eq:ndUnoptimizedDist}
%remains valid when variable $X_j$ has don't-care intervals,
%even though the optimal point density $\lambda_j$ obtained from
%\eqref{eq:multi-fixed-lambda} is zero where $f_{X_j}$ is nonzero
%(invalidating assumption MF4).
%Here we give an argument relying on an explicit characterization of
%the distortion similar to \eqref{eq:ndUnoptimizedDist}.
At sufficiently high rates,
it is intuitive to allot a codeword of $Q_j$
to each don't-care interval $Z_{j,i}$.
The remaining $K_j - M_j$ codewords are assigned optimally to
$[0,1] \setminus Z_j$ according to the basic theory developed
in Section~\ref{sec:Multi}.  We refer to
this quantizer structure as a \emph{don't-care quantizer}.

\begin{theorem}
\label{thm:fixedratesame}
Suppose $n$ sources $X_1^n \in [0,1]^n$ are quantized by a sequence
of distributed don't-care quantizers $Q_K^{\wB}$.
Further suppose that the sources, quantizers, and
function $g: [0,1]^n \rightarrow \R$ satisfy assumptions
MF1--MF3, and assumption MF4 is replaced by the following:
The integrals
\[
\int_{[0,1]\setminus Z_j}f(x_j)\E{|g_j(X^n)|^2 \mid X_j = x_j}\lambda_j(x_j)^{-2} \, dx_j
\]
are finite for every $j\in\{1,\ldots,n\}$.
Finally, assume each source $X_j$ has $M_j$ don't-care intervals
satisfying \eqref{eq:dontcare-probability}.
%Then 
%the optimal point densities satisfy
%\beq
%  \lambda_j(x) = 0
%\quad
%  \mbox{for all $x \in Z_j$}
%  \label{eq:dontCareLambdaZero}
%\eeq
%and yield
Then the high-resolution distortion-resolution function is asymptotically
accurate to the true distortion-resolution function:
\beqa
\dhr(K;\lB) & = & \frac{n}{12}  \left( \prod_{j=1}^n \frac{\P{A_j}}{(K_j-M_j)^2}
          \E{\left(\frac{\gamma_j(X_j)}{\lambda_j(X_j)}\right)^2
             \mid A_j}\right)^{1/n}
          \mbox{,}
  \label{eq:dontCareDist} \\
 & \sim & d(K;\lB) \mbox{.} \nonumber
\eeqa
%The optimal point densities for fixed-rate quantization
%are given by \eqref{eq:dontCareLambdaZero}
%inside the don't-care intervals and by 
%\eqref{eq:multi-fixed-lambda} outside of the don't-care intervals.
%These point densities, along with an optimal choice of the fractional
%allocation $\aB$ yield
%\beqa
%  \Dfrhr &= & \sum_{j=1}^n \frac{1}{12(K^{\alpha_j}-M_j)^2} \|\gamma_j^2 f_{X_j}\|_{1/3}
%  \label{eq:dontCareDistKjMj} \\
%    & \sim & \frac{1}{12} \prod_{j=1}^n \left(\|\gamma_j^2 f_{X_j}\|_{1/3}\right)^{1/n} 2^{-2R/n} \mbox{,}
%  \label{eq:dontCareDistRj} 
%\eeqa
%which is identical to the result of Theorem \ref{thm:multi-summary}.
\end{theorem}
\begin{IEEEproof}
Follows from applying Theorem \ref{thm:multi-distortion} to the
region $\left([0,1]\setminus Z_1\right) \times \cdots \times \left([0,1]\setminus Z_n\right)$.
\end{IEEEproof}

\subsection{The Distortion-Rate Functions}
\label{sec:dontcare-variable}
In the fixed-rate case, the high-resolution resolution-rate function
is unchanged: $\Kfrhr(R;\lB) = 2^R$.  Asymptotic validity is
easily observed: $\lim_{R\rightarrow \infty}\Kfr(R;\lB)/\Kfrhr(R;\lB) = 1$.
Applying this to the distortion-resolution expression \eqref{eq:dontCareDist},
we obtain the unoptimized high-resolution fixed-rate distortion-rate function:
\beqan
\Dfr(R; \lB) \sim \Dfrhr(R;\lB) & = & \frac{n}{12} \left(  \prod_{j=1}^n\frac{\P{A_j}}{(K_j-M_j)^2}
          \E{\left(\frac{\gamma_j(X_j)}{\lambda_j(X_j)}\right)^2
             \mid A_j}\right)^{1/n} \\
& \sim &  \frac{n}{12K^{2/n}} \left( \prod_{j=1}^n \P{A_j}
          \E{\left(\frac{\gamma_j(X_j)}{\lambda_j(X_j)}\right)^2
             \mid A_j}\right)^{1/n}
             \mbox{.}
\eeqan   
The optimal point densities for fixed-rate quantization
are given by 
\eqref{eq:multi-fixed-lambda} outside of the don't-care intervals.
These point densities yield an optimized high-resolution fixed-rate
distortion-rate function
\beq
\Dfrhr(R) = \frac{n}{12} \left( \prod_{j=1}^n\|\gamma_j^2 f_{X_j}\|_{1/3}\right)^{1/n} 2^{-2R/n} \mbox{.}
\eeq                              

The variable-rate case is a bit more involved.
To formalize the analysis,
we define discrete random variables to represent the events of
source variables lying in don't-care intervals.

\begin{definition}
The random variable
\[ I_j = \left\{ \begin{array}{ll}
       i, & \mbox{if $X_j \in Z_{j,i}$ for $i \in \{1,\,2,\,\ldots,\,M_j\}$}; \\
       0, & \mbox{otherwise}
                 \end{array} \right.
\]
is called the $j$th \emph{don't-care variable}.
The previously-defined event $A_j$ can be expressed as $\{I_j = 0\}$.
\end{definition}

At sufficiently high rate, the $j$th encoder communicates $I_j$ and
in addition, \emph{only when $I_j = 0$}, a fine quantization of $X_j$.
The resulting performance is summarized by the following theorem.

\begin{theorem} \label{thm:variableratedifferent}
Under the conditions of Theorem~\ref{thm:fixedratesame},
the optimal point densities for variable-rate quantization follow
\eqref{eq:multi-var-lambda} and yield
%\beqa
%\Dvrhr(R;\aB) & = & \frac{1}{12} \sum_{j=1}^n \rho_j^{-1}  \label{eq:dontCareDistVar}  \\
% &          & \qquad \times  2^{-2(\rho_j (\alpha_jR- H(I_j))
%                +2h(X_j \mid A_j)
%                +2\E{\log(\gamma_j(X_j)) \mid A_j}} \mbox{,}\nonumber
%\eeqa
\beqa
\Dvrhr(R) & = & \frac{n}{12}  \left( \prod_{j=1}^n\rho_j^{-1} \right. \label{eq:dontCareDistVar}   \\
 &          & \left. \qquad \times  2^{-2(\rho_j (R- H(I_j))
                +2h(X_j \mid A_j)
                +2\E{\log(\gamma_j(X_j)) \mid A_j}} \right)^{1/n} \mbox{.}\nonumber
\eeqa
where $\rho_j = 1/\P{A_j}$ is the \emph{amplification}
of $R_j$.  

%Optimizing the fractional allocation $\alpha_1^n$ then yields

\end{theorem}
\begin{IEEEproof}
See Appendix~\ref{app:dontCareVar}.
\end{IEEEproof}

Some remarks:
\begin{enumerate}
\item
The quantity $H(I_j)$
may be identified as the cost of communicating
the indicator information to the decoder.
The remaining rate, $R_j-H(I_j)$, is amplified by factor $\rho_j$
because additional description of $X_j$ is useful only when $X_j \notin Z_j$.
The amplification shows that the standard $-6$ dB/bit-per-source distortion decay
may be exceeded in the presence of don't-care regions.

\item
At moderate rates, it may not be optimal to communicate $I_j$ losslessly,
and it may be beneficial to include $X_j$ values with small but positive
$\gamma_j$ in don't-care intervals.
Study of this topic is left for specific applications.

\item
The rate amplification we have seen in the variable-rate case and the
relative lack of importance of don't-care intervals in the fixed-rate case
have a close analogy in ordinary lossy source coding.
Suppose a source $X$ is a mixed random variable with
an $M$-valued discrete component and a continuous component.
High-resolution quantization of $X$ will allocate one level to each
discrete value and the remaining levels to the continuous component.
The discrete component changes the constant factor in $\Theta(2^{-2R})$
fixed-rate operational distortion--rate performance while it changes
the decay rate in the variable-rate case.
See~\cite{GyorgiLZ1999} for related Shannon-theoretic
(rather than high-resolution quantization) results.
\end{enumerate}

\section{Chatting Encoders}
\label{sec:Chatting}
Our final variation on the basic theory of distributed functional
scalar quantization is to allow limited communication between the encoders.
How much can the distortion be reduced via this communication?
Echoing the results of the previous section, we will find
dramatically different answers in the fixed- and variable-rate cases.

For notational convenience, we will fix the communication to be from
encoder 2 to encoder 1 though the number of source variables $n$ remains
general.
In accordance with the block diagram of Fig.~\ref{fig:chatting},
the information $Y = Y_{2\rightarrow1}$ must be conditionally
independent of $X_1$ given $X_2$. 
We consider only the case where $Y$ is a single bit;
this suffices to illustrate the key ideas.

\begin{figure}
 \centering
  \psfrag{x1}[r][][1][0]{$X_1$}
  \psfrag{x2}[r][][1][0]{$X_2$}
  \psfrag{x1h}[r][][1][0]{$\widehat{X}_1$}
  \psfrag{x2h}[r][][1][0]{$\widehat{X}_2$}
  \psfrag{y}[l][l][1][0]{$Y = Y_{2\rightarrow1}$}
  \psfrag{Q1}[b][][1][0]{$Q_1$}
  \psfrag{Q2}[b][][1][0]{$Q_2$}
  \includegraphics[width=0.3\textwidth]{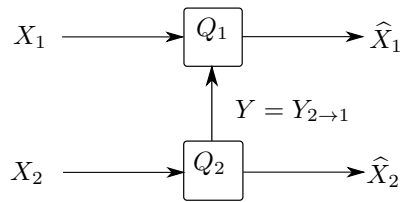}
  \caption{Suppose the encoder for $X_2$ could send a bit to the encoder
     for $X_1$.  Is there any benefit?  How does it compare to sending an
     additional bit to the decoder?}
  \label{fig:chatting}
\end{figure}

In this section, we express the high-resolution distortion as 
\[ 
D^{\rm HR} = \frac{n}{12} 2^{-2R/n} \left( \prod_{j=1}^n D_j\right)^{1/n} \mbox{,}
\]
where various expressions for $D_j$ have been found for different
scenarios, including for fixed-rate \eqref{eq:multi-fixed-alloc-dist} and variable-rate \eqref{eq:multi-var-alloc-dist}
quantization.
%and \eqref{eq:dontCareDistVar}.
At issue is how $D_1$ is affected by $Y$;
the other $D_j$s are not affected.

\subsection{Fixed-Rate Quantization}
In general, the availability of a single bit $Y$ causes one to
choose between two potentially-different quantizers $Q_{1\mid Y=0}$
and
$Q_{1\mid Y=1}$ in the quantization of $X_1$.
We express the optimal quantizers and the resulting distortion contribution
$D_1$ by way of the following concept.

\begin{definition}
The $j$th \emph{conditional functional sensitivity profile} of $g$
given $Y = y$ is defined as
\[ 
  \gamma_{j|Y}(x \mid y)
    = \left(\E{\left|{g_j(X_1^n)}\right|^2 \mid X_j = x, Y=y}\right)^{1/2} \mbox{.}
\]
\end{definition}

Now several results follow by analogy with Theorem~\ref{thm:multi-summary}.
For the case of $Y = y$, the optimal point density is given by
$$
  \lambda_{1|Y}(x \mid y)
       = \frac{\left(\gamma_{1|Y}^2(x \mid y) f_{X_1|Y}(x \mid y)\right)^{1/3}}
              {\int_0^1 \left(\gamma_{1|Y}^2(t \mid y) f_{X_1|Y}(t \mid y)\right)^{1/3} \, dt}
$$
resulting in conditional distortion contribution
$$
  \frac{1}{12K_1^2}\left\| \gamma_{1 | Y=y}^2 f_{X_1 | Y=y} \right\|_{1/3}\mbox{.}
$$
Combining the two possibilities for $Y$ via total expectation gives
\beq
  D_1  =  \sum_{y=0}^1 \P{Y=y} \left\| \gamma_{1 | Y=y}^2 f_{X_1 | Y=y} \right\|_{1/3}.
             \label{eq:chat-fixed}
\eeq
From this expression we reach an important conclusion on the
effect of the chatting bit $Y$.

\begin{theorem} 
\label{thm:FixedRateChatting}
For fixed-rate quantization, communication of one bit of information from 
decoder 2 to decoder 1 will asymptotically reduce $D_1$ by at most a factor of 4\@.
\end{theorem}
\begin{IEEEproof}
From Theorem~\ref{thm:multi-summary}, the distortion contribution
analogous to \eqref{eq:chat-fixed} without the chatting bit $Y$ is
$\left\| \gamma_{1}^2 f_{X_1} \right\|_{1/3}$.
Thus the fact we wish to prove is a statement about $\Lonethird$ quasinorms
of weighted densities and their conditional forms.

We proceed as follows:
\beqan
D_1 & = & \sum_{y=0}^1 \left\| \P{Y=y} \gamma^2_{1|Y}(x \mid y)
                                           f_{X_1|Y}(x \mid y) \right\|_{1/3} \\
  & \geqlabel{a} & \frac{1}{4}\left\|
                     \sum_{y=0}^1 \P{Y=y} \gamma^2_{1|Y}(x \mid y)
                                           f_{X_1|Y}(x \mid y) \right\|_{1/3} \\
  & = & \frac{1}{4}\left\| f_{X_1}(x)
                     \sum_{y=0}^1 \frac{\P{Y=y}
                                           f_{X_1|Y}(x \mid y)}
                                       {f_{X_1}(x)}
                                          \gamma^2_{1|Y}(x \mid y)
                                                               \right\|_{1/3} \\
  & \eqlabel{b} & \frac{1}{4}\left\| f_{X_1}(x)
                     \sum_{y=0}^1 \P{Y=y \mid X_1 = x}
                                          \gamma^2_{1|Y}(x \mid y)
                                                               \right\|_{1/3} \\
  & \eqlabel{c} & \frac{1}{4}\left\| f_{X_1}(x)
                                          \gamma^2_{1}(x)
                                                               \right\|_{1/3} \mbox{,}
\eeqan
where (a) uses a quasi-triangle inequality that may be established via well-known
inequalities (see Appendix \ref{app:triangle} for a statement and proof);
(b) is an application of Bayes's Rule; and
(c) is based on an evaluation of the (unconditional) functional sensitivity
via the total expectation theorem with conditioning on $Y$.
This proves the theorem.
\end{IEEEproof}

Note that
while the bit $Y$ leads a reduction of $D_1$ by at most a factor of 4 and therefore
a reduction of $D^{\rm HR}$ by at most a factor of $4^{1/n}$,
an identical reduction in distortion is achieved simply by increasing the rate $R$ to the centralized
decoder by one bit\@.
Generalizing to any number of chatting bits, we obtain
the following corollary.
\begin{corollary} 
\label{cor:FixedChattingBad}
For fixed-rate functional quantization, communication of some number of bits
from encoder $j$ to encoder $k$ performs \emph{at best} as well as increasing
the communication to the centralized decoder by the same
number of bits.
\end{corollary}

In general, the idea that bits from encoder 2 to encoder 1 are as good as
bits from encoder 1 to the decoder is optimistic.
In particular, if $\E{\gamma_1^2(X_1)} > 0$,
then $D_1$ is bounded away from zero for any amount of communication
from encoder 2 to encoder 1\@.

\subsection{Variable-Rate Quantization}
In a variable-rate scenario,
the rate could be made to depend on the chatting bit $Y$,
introducing a bit allocation problem between the cases of
$Y=0$ and $Y=1$.
Even without such dependence, we can demonstrate that the bit $Y$
can reduce the first variable's contribution to the functional
distortion by an arbitrary factor.

Analogous to \eqref{eq:chat-fixed},
\beq
  D_1  =  \sum_{y=0}^1 \P{Y=y} 
                     2^{2h(X_1 \mid Y=y) + 2\E{\log \gamma_{1 | Y=y}(X_1)}}
             \label{eq:chat-var}
\eeq
by comparison with \eqref{eq:multi-var-alloc-dist}.
In contrast to the $\Lonethird$ quasinorms in \eqref{eq:chat-fixed},
this linear combination can be arbitrarily smaller than
$$
                      2^{2h(X_1) + 2\E{\log \gamma_{1}(X_1)}}\mbox{.}
$$
We demonstrate this through a simple example.

\begin{figure}
 \centering
  \psfrag{1}[r][][1][0]{\large $\frac{1}{2}$}
  \psfrag{2}[r][][1][0]{\large $1$}
  \psfrag{a}[r][][1][0]{\large $1$}
  \psfrag{b}[r][][1][0]{\large $L$}
  \psfrag{x1}[r][][1][0]{\small $x_1$}
  \psfrag{x2}[r][][1][0]{\small $x_2$}
  \includegraphics[width=0.2\textwidth]{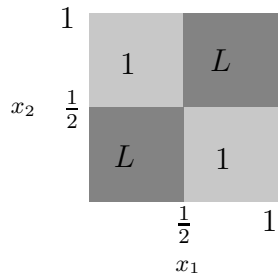}
  \caption{Illustration for Example~\ref{ex:chat-var}. 
       Shown is the unit square $[0,1]^2$ with quadrants marked with the
       value of $g_1(x_1,x_2)$, the derivative of $g$ with respect to $x_1$.}
  \label{fig:quadrant}
\end{figure}

\begin{example}
\label{ex:chat-var}
Let sources $X_1$ and $X_2$ be uniformly distributed on $[0,1]^2$.
We specify the function of interest $g$ through its partial derivatives.
Let $g_2(x_1,x_2) = 1$ for all $(x_1,x_2)$ and let
$g_1(x_1,x_2)$ be piecewise constant as shown in Fig.~\ref{fig:quadrant},
where $L$ is a positive constant.  

While $g(x_1,x_2)$ is not continuous everywhere and condition MF1 is
therefore not strictly satisfied, the points of discontinuity
fall along the line $x_2 = 1/2$.  As observed following the proof of
Corollary \ref{cor:single-discontinuous}, this variety of discontinuity
can be easily and intuitively merged with high-resolution analysis:
one simply places an extra quantizer cell boundary at $x_2 = 1/2$
for every quantizer in the sequence being considered.  This increases
the resolution $K$ by 1, but has negligible impact on the rate in the limit
$K\rightarrow \infty$.

We can easily derive the first functional sensitivity profile of $g$ to be
\[ 
\gamma_1(x) = \sqrt{\half(L^2+1)} \mbox{.}
\]
This also allows us to find the distortion contribution factor $D_1$
without chatting to be
\[ 
D_1 = {\ts\frac{1}{4}}(L^2+1)^2\mbox{.}
\]

In this example, one bit about $X_2$ is enough to allow the encoder for $X_1$
to perfectly tailor its point density to match the sensitivity of
$g$ at $(X_1,X_2)$.  Of course, the chatting bit should simply be
\[
Y = \left\{ \begin{array}{ll} 0, & \mbox{if $X_2 > 1/2$}; \\
                              1, & \mbox{otherwise}\mbox{.}
            \end{array}
    \right.
\]
The first conditional functional sensitivity profiles for $g$ are then
\[
\gamma_{1|Y}(x \mid y)
   = \left\{ \begin{array}{ll}
         1, & \mbox{for $Y=0$ and $X_1 \leq 1/2$} \\
            & \mbox{or $Y=1$ and $X_1 > 1/2$}; \\
         L, & \mbox{otherwise.}
             \end{array}
     \right.
\]
Now for either value of $y$, we have
$\int_0^1 \gamma_{1|Y}(x \mid y) \, dx = \half(L+1)$
and
$\E{\log \gamma_{1|Y=y}(X_1)} = \half \log L$.
Thus, evaluating \eqref{eq:chat-var} gives
$$
  D_1  = {\ts\frac{1}{4}}(L+1)^2L\mbox{.}
$$
This is smaller than the $D_1$ with no chatting by about a factor of $L$.
The performance gap can be made arbitrarily large by increasing $L$---all
from just one bit of information communicated between encoders per sample.
\hfill $\Box$
\end{example}

\subsection{Comparison with Ordinary Source Coding}
The results of this section are strikingly different from those of ordinary
source coding.   Consider first the discrete scenario in which we
with to recreate $X_1^n$ perfectly at the decoder. 
Can communication between encoders enable a reduction in the rate of
communication to the decoder?
According to Slepian and Wolf, the answer is a resounding ``no.''
Even in the case of unlimited collaboration via fused encoders,
the minimum sum rate to the decoder remains unchanged.

How about in lossy source coding?  If quantization is variable-rate
and Slepian--Wolf coding is employed on the quantization indices,
no gains are possible from encoder interactions. 
This is a consequence of the work of Rebollo-Monedero \emph{et al.}~\cite{Rebollo-MonederoRAG2006}
on high-resolution Wyner--Ziv coding, where it is shown that
there is no gain from supplying the source encoder with
the decoder side information.

\section{Summary}
\label{sec:Conclusion}
We have developed asymptotically-optimal companding designs of functional quantizers
using high-resolution quantization theory.
This has shown that accounting for a function while quantizing a source can
lead to arbitrarily large improvements in distortion. 
In certain scenarios (Section~\ref{sec:Scaling}),
this improvement can grow exponentially with the number of sources.
In others (Section~\ref{sec:DontCare}),
it can grow exponentially with rate. 

Additionally, our study of functional quantization has highlighted some
striking distinctions between fixed- and variable-rate cases:
\begin{enumerate}
\item For certain simple functions of order statistics,
  distortion relative to ordinary quantization
  falls polynomially with the number of sources in the fixed-rate
  case, whereas in the variable-rate case it falls exponentially.
\item The distortion associated with fixed-rate quantizers will always
  exhibit $-6$ dB/bit rate dependence at high rates, whereas
  the decay of distortion can be faster in some variable-rate cases.
\item Information sent from encoder-to-encoder can lead to arbitrarily-large
  improvements in distortion for variable-rate, whereas for fixed-rate this
  information can be no more useful than an equal amount of information
  sent to the decoder.
\end{enumerate}

The second and third of these have extensions or analogues beyond functional
quantization.  Rate amplification is a feature of quantizing sources with
mixed probability distributions, and the results on chatting encoders continue to hold
when the function $g$ is the identity operation. 

\appendices
\section{Proof of Theorem~\ref{thm:multi-distortion}}
\label{app:distortion}
%\begin{lemma}
%\label{lem:absContinuous}
%Let $g$ be a real-valued function defined over a bounded interval $S$, and suppose
%the derivative $g'$ is defined almost-everywhere in $S$ with respect to the Jordan measure.
%Then $g$ is absolutely continuous.
%\end{lemma}
%
%\begin{proof}
%Let $a,b \in S$ such that $a < b$.  We will show that $g(b)-g(a) = \int_a^b g'(x)dx$.
%First, let $A_g = \{x \in S: g'(x) \text{ exists}\}$.  Because $A_g$ has Jordan measure
%same as $S$, $A_g$ may be expressed as a countable union of disjoint intervals.
%This follows because the Jordan measure may be expressed as the limiting
%measure of a sequence of sets $B_i
%

\begin{lemma}
\label{lem:1Dlemma}
Let $X$ be a real-valued random variable distributed over a bounded interval $S$,
and let $g(x)$ be an absolutely continuous, real-valued function on $S$ with
bounded derivative $g'(x)$.
If $g'(x)$ is defined for \emph{almost every} $x \in S$, then
$\var \left(g(X) \right) \leq \var \left(bX \right)$, where
$b = \sup_{x\in S} |g'(x)|$.
If $g'(x)$ is furthermore defined for \emph{every} $x \in S$, then
$\var \left(g(X) \right) \geq \var \left(aX \right)$,
where $a = \inf_{x\in S} |g'(x)|$.
\end{lemma}
\begin{IEEEproof}
Draw $X_1$ and $X_2$ i.i.d.\ according to the distribution of $X$,
and define the following functions:
\beqan
D_a(X_1,X_2) & = & a(X_1 - X_2), \\
D_g(X_1,X_2) & = & g(X_1)- g(X_2), \\
D_b(X_1,X_2) & = & b(X_1 - X_2).
\eeqan

To prove the first part of the lemma, assume the derivative of $g$ is defined for \emph{almost every} $x \in S$.
By the absolute continuity of $g$, 
$$
D_g(X_1,X_2)  =  g(X_2) - g(X_1) 
 =  \int_{X_1}^{X_2} g'(x) \, dx \mbox{.}
$$
Therefore, the magnitude $|D_g(X_1,X_2)|$ can be bounded above as follows:
\[
|D_g(X_1,X_2)| \leq \int_{X_1}^{X_2} |g'(x)| \, dx \leq b|X_1 - X_2| = |D_b(X_1,X_2)| \mbox{.}
\]
This then implies that $\iE{D_g^2} \leq \iE{D_b^2}$ and, since $\iE{D_g} = \iE{D_b} = 0$,
that 
\beq
\var(D_g) \leq \var(D_b) \mbox{.} \label{eq:varInequality1}
\eeq
Since each $D_b$ and $D_g$ is a sum of i.i.d.\ variables,
$\var(D_b) = 2 \var(bX)$ and $\var(D_g) = 2\var(g(X))$.
Inserting these into \eqref{eq:varInequality1} and dividing by 2 proves the first part of the lemma.

Now to prove the second part of the lemma, assume further that $g'(x)$ is defined for \emph{every} $x \in S$.
By the mean value theorem, there exists $X_0$ between $X_1$ and $X_2$
such that
\[
| D_g(X_1,X_2)| = \left|\int_{X_1}^{X_2} g'(x) \, dx \right| = |g'(X_0)(X_1 - X_2)| \geq a|X_1 - X_2| = |D_a(X_1,X_2)| \mbox{.}
\]
As before, this implies that 
\beq
\var(D_g) \geq \var(D_a) \mbox{.} \label{eq:varInequality2}
\eeq
Since 
$\var(D_a) = 2 \var(aX)$ and $\var(D_g) = 2\var(g(X))$,
substituting into \eqref{eq:varInequality2} and dividing by 2 proves the second part of the lemma.

\end{IEEEproof}

We now define a function $\widetilde{g}_{\widetilde{x}_1^{j-1}}(x)$
that will appear in the proof of the theorem after we
establish properties of the function in a lemma.

\begin{definition}
Suppose $\widetilde{X}_1^n$ is uniformly distributed over a rectangular
region $S$.  The $j$th \emph{reduced-dimension function}
(with parameter vector $\widetilde{x}_{1}^{j-1}$) is defined as
\[
\widetilde{g}_{\widetilde{x}_1^{j-1}}(x) = \E{g(\widetilde{X}_1^n) \mid \widetilde{X}_1^{j-1} = \widetilde{x}_1^{j-1},\, \widetilde{X}_j = x} \mbox{.}
\]
\end{definition}

\begin{lemma}
\label{lem:gtwiddle}
Let $\widetilde{X}_1^n$ be uniformly distributed over a rectangular region $S = S_1 \times S_2 \times \cdots \times S_n$,
and let $g$ be Lipschitz continuous over $S$.  If the first and second
derivatives of $g$ are defined and bounded \emph{almost everywhere} in $S$, then:
\begin{enumerate}
\item $\widetilde{g}_{\widetilde{x}_1^{j-1}}(x)$ is Lipschitz continuous
in $x$.
%\item For almost all $\widetilde{x}_1^{j-1} \in S_1 \times S_2 \times \cdots \times S_{j-1}$, the derivative $\widetilde{g}_{\widetilde{x}_1^{j-1}}'(x)$ is defined
%for almost all $x \in S_j$.
\item Where defined, $|\widetilde{g}_{\widetilde{x}_1^{j-1}}'(x)| \leq b_j = \sup_{x_1^n \in S} |g_j(x_1^n)|$.
\end{enumerate}
If the first and second derivatives of $g$ are furthermore defined and bounded \emph{everywhere} in $S$, then:
\begin{enumerate}
%\item $\widetilde{g}_{\widetilde{x}_1^{j-1}}(x)$ is continuous
%in $x$.
\item[3)] For all $\widetilde{x}_1^{j-1} \in S_1 \times S_2 \times \cdots \times S_{j-1}$, the derivative $\widetilde{g}_{\widetilde{x}_1^{j-1}}'(x)$ is defined
for all $x \in S_j$.
\item[4)] The magnitude of derivative may be lower bounded:
$|\widetilde{g}_{\widetilde{x}_1^{j-1}}(x)| \geq a_j = \inf_{x_1^n \in S} |g_j(x_1^n)|$.
\end{enumerate}
\end{lemma}
\begin{IEEEproof}
First, assume that the first and second derivatives of $g$ are 
defined and bounded \emph{almost everywhere} in $S$.
\begin{enumerate}
\item Since $\widetilde{g}_{\widetilde{x}_1^{j-1}}(x)$ is an average of
functions with a common Lipschitz constant, it too is Lipschitz with this constant.   
%\item Since the $j$th first partial derivative $g_j(\widetilde{x}_1^n)$ is defined for almost all
%$\widetilde{x}_1^n \in S$, the conditional expectation $\iE{g_j(\widetilde{X}_1^n) \mid \widetilde{X}_1^j= \widetilde{x}_1^j}$
%is defined for almost all $\widetilde{x}_1^j \in S_1 \times S_2 \times \cdots \times S_j$.
%By linearity of the derivative operation and the definition of
%the $j$th reduced-dimension function, this conditional expectation is precisely
%the derivative $\widetilde{g}_{\widetilde{x}_1^{j-1}}'(\widetilde{x}_j)$,
%which is defined at almost all points $\widetilde{x}_j \in S_j$
%for almost all parameter vectors
%$\widetilde{x}_1^{j-1} \in S_1 \times S_2 \times \cdots \times S_j$.
\item Where $\widetilde{g}_{\widetilde{x}_1^{j-1}}'(x)$ is defined, we have
\beqan
 \left| \widetilde{g}_{\widetilde{x}_1^{j-1}}'(x_j) \right| & = & \left| \E{g_j(X_1^n) \mid \widetilde{X}_1^j = \widetilde{x}_1^j} \right| \\
 & \leq & \E{\left| g_j(X_1^n) \right| \mid \widetilde{X}_1^j = \widetilde{x}_1^j} \\
 & \leq & b_j \mbox{.}
\eeqan
\end{enumerate}
Now assume that furthermore the first and second derivatives of $g$ are defined
and bounded \emph{everywhere} in $S$.
\begin{enumerate}
%\item As before, continuity of $g$ guarantees continuity of $\widetilde{g}_{\widetilde{x}_1^{j-1}}$.
\item[3)] Since $g_j(\widetilde{x}_1^n)$ is defined for all
$\widetilde{x}_1^n \in S$, the average derivative $\widetilde{g}_{\widetilde{x}_1^{j-1}}'(\widetilde{x}_j)$
is defined for all $\widetilde{x}_j \in S_j$.
\item[4)] We now obtain a lower bound on the derivative.  As before, we note that
$\widetilde{g}_{\widetilde{x}_1^{j-1}}'(x_j) =\iE{g_j(X_1^n) \mid \widetilde{X}_1^j = \widetilde{x}_1^j}$.
Because the derivatives of $g_j$ are defined everywhere in $S$, and because the expectation
under a uniform distribution is just an average, the mean value theorem guarantees the existence of an
$e_{\widetilde{x}_1^j} \in S$ such that $g_j(e_{\widetilde{x}_1^j}) = \iE{g_j(X_1^n) \mid \widetilde{X}_1^j = \widetilde{x}_1^j}$.  Finally, since $e_{\widetilde{x}_1^j} \in S$, we have $|g_j(e_{\widetilde{x}_1^j})| \geq a_j$.
To summarize:
\beqan
\left| \widetilde{g}_{\widetilde{x}_1^{j-1}}'(x_j) \right| & = & \left| \E{g_j(X_1^n) \mid \widetilde{X}_1^j = \widetilde{x}_1^j} \right| \\
& = & \left| g_j\left(e_{\widetilde{x}_1^j}\right) \right| \\
& \geq & a_j \mbox{.}
\eeqan
\end{enumerate}

\end{IEEEproof}

We now bound the variance of the function within a rectangular
cell, assuming the source is uniformly distributed.  This
will later be adapted to the case where the source is nonuniformly
distributed.

\begin{lemma}
\label{lem:uniformDistLemma}
Let $g(x_1^n)$ be a Lipschitz continuous function defined over a rectangular cell
$S = S_1 \times \cdots \times S_n$ with edge lengths $\Delta_1, \ldots, \Delta_n$, let 
$a_j$ and $b_j$ be lower and upper bounds to $|g_j(x_1^n)|$, when it exists,
and let $\overline{g}$ denote the average value of $g$ within $S$:
\[ \overline{g} = \frac{1}{\prod_{j=1}^n\Delta_j} \int_S g(x_1^n) \, dx_1^n \mbox{.} \]
If the first and second derivatives of $g$ are defined \emph{almost everywhere}
in $S$, then 
 \[ 
 \left( {\ts \prod_{j=1}^n \frac{1}{\Delta_j}} \right) \int_S |g(x_1^n) - \overline{g}|^2 \, dx_1^n \leq
 \sum_{j=1}^n \frac{b_j^2 \Delta_j^2}{12} \mbox{.}
 \]
If the first and second derivatives of $g$ are furthermore defined \emph{everywhere}
in $S$, then
 \[
 \sum_{j=1}^n \frac{a_j^2 \Delta_j^2}{12} \leq 
 \left( {\ts \prod_{j=1}^n \frac{1}{\Delta_j}} \right) \int_S |g(x_1^n) - \overline{g}|^2 \, dx_1^n \mbox{.}
 \]
\end{lemma}
\begin{IEEEproof}
Since $\widetilde{X}_1^n$ is uniformly distributed over $S$,
\[  \frac{1}{\prod_{j=1}^n\Delta_j} \int_S |g(x_1^n) - \overline{g}|^2 \, dx_1^n 
=  \var(g(\widetilde{X}_1^n)) \mbox{.} \]
This may be expanded by repeated application of the law of total variance:
\beqa
\var(g(\widetilde{X}_1^n)) & \eqlabel{a}  &
\E{ \var \left( g(\widetilde{X}_1^n) \mid \widetilde{X}_1 \right) } 
+ \var \left( \E{g(\widetilde{X}_1^n) \mid \widetilde{X}_1} \right) \nonumber \\
& \eqlabel{b} & 
\E{ \E{ \var \left( g(\widetilde{X}_1^n) \mid \widetilde{X}_1, \widetilde{X}_2 \right) \mid \widetilde{X}_1} 
+ \var \left( \E{g(\widetilde{X}_1^n) \mid \widetilde{X}_1, \widetilde{X}_2} \mid \widetilde{X}_1 \right)}
+ \var \left( \E{g(\widetilde{X}_1^n) \mid \widetilde{X}_1} \right) \nonumber\\
& \eqlabel{c} &
\E{\var \left( g(\widetilde{X}_1^n) \mid \widetilde{X}_1, \widetilde{X}_2 \right)} 
+ \E{\var \left( \E{g(\widetilde{X}_1^n) \mid \widetilde{X}_1, \widetilde{X}_2} \mid \widetilde{X}_1 \right)}
+ \var \left( \E{g(\widetilde{X}_1^n) \mid \widetilde{X}_1} \right) \nonumber\\
& \eqlabel{d} &
\E{\var \left( g(\widetilde{X}_1^n) \mid \widetilde{X}_1^{n-1} \right)}
+ \sum_{j=1}^n \E{\var \left( \E{g(\widetilde{X}_1^n) \mid \widetilde{X}_1^{j}} \mid \widetilde{X}_1^{j-1} \right) } \nonumber\\
& = &
\sum_{j=1}^n \E{ 
                           \var \left( 
                                          \E{ g(\widetilde{X}_1^n) \mid \widetilde{X}_1^j } 
                                          \mid \widetilde{X}_1^{j-1} 
                                \right) 
                         } \nonumber \\
 & = & \sum_{j=1}^n \E{ \var \left(
                                          \widetilde{g}_{\widetilde{X}_1^{j-1}}
                                          (\widetilde{X}_j) 
                                   \mid \widetilde{X}_1^{j-1} 
                                \right) 
                         }
                         \label{eq:totalVarianceAppDistortion} 
\mbox{,}
\eeqa
where (a) follows from the law of total variance with conditioning on $\widetilde{X}_1$;
(b) uses the law of total variance applied to the variance within
the expectation in the first term, with conditioning performed on $\widetilde{X}_2$;  
(c) simplies the first term using iterated expectation; and
(d) applies the law of total variance repeatedly
to the variance within the first expectation, as in step (b), with conditioning
on $\widetilde{X}_j$ during the $j$th iteration.

\emph{Upper bound.} Let $A(\widetilde{x}_1^{j-1}) \in S_j$
be the set of points $x_j \in S_j$ where the derivative 
$\widetilde{g}_{\widetilde{x}_1^{j-1}}'(x_j)$ is undefined.
By Lemma \ref{lem:gtwiddle}, $\widetilde{g}_{\widetilde{x}_1^{j-1}}(x_j)$ is 
Lipschitz continuous and therefore $A(\widetilde{x}_1^{j-1})$
is of measure zero.
%We may therefore condition the expectation in \eqref{eq:totalVarianceAppDistortion}
%on this event:
%\[
%\sum_{j=1}^n \E{ \var \left[ 
%                                          \widetilde{g}_{\widetilde{X}_1^{j-1}}
%                                          (\widetilde{X}_j) 
%                                   \mid \widetilde{X}_1^{j-1} 
%                                \right] 
%                         }
%= 
%\sum_{j=1}^n \E{ \var \left[ 
%                                          \widetilde{g}_{\widetilde{X}_1^{j-1}}
%                                          (\widetilde{X}_j) 
%                                   \mid \widetilde{X}_1^{j-1} 
%                                \right] 
%                   \mid \mu\left(A(\widetilde{X}_1^{j-1})\right)=0      }
%                   \mbox{.}
%\]
As such, for every value of $\widetilde{X}_1^{j-1}$ considered within the expectation,
$\var \left( 
                                          \widetilde{g}_{\widetilde{X}_1^{j-1}}
                                          (\widetilde{X}_j) 
                                   \mid \widetilde{X}_1^{j-1} 
                                \right) $
is now the variance of a function satisfying the upper bound conditions 
for Lemma \ref{lem:1Dlemma}.  Applying this upper bound within
the expectation, we have
\beqan
\sum_{j=1}^n \E{ \var \left( 
                                          \widetilde{g}_{\widetilde{X}_1^{j-1}}
                                          (\widetilde{X}_j) 
                                   \mid \widetilde{X}_1^{j-1} 
                                \right) 
                         }
& \leq &
\sum_{j=1}^n \E{ \var \left( b_j \widetilde{X}_j
                                   \mid \widetilde{X}_1^{j-1} 
                                \right) 
                 \mid \mu\left(A(\widetilde{X}_1^{j-1}) \right)=0} \\
& = & 
\sum_{j=1}^n \frac{b_j^2 \Delta_j^2}{12} \mbox{,}
\eeqan                 
which proves the first half of the lemma.

\emph{Lower bound.} We return to \eqref{eq:totalVarianceAppDistortion},
now assuming that the first and second derivatives of $g$ are defined
everywhere in the cell $S$.  By Lemma~\ref{lem:gtwiddle}, for any choice
of $\widetilde{x}_1^j \in S_1 \times \cdots \times S_j$
the function $\widetilde{g}_{\widetilde{x}_1^{j-1}}(x_j)$ satisfies the conditions for
the lower bound in Lemma~\ref{lem:1Dlemma}.  Inserting this lower bound
into the expectation, we obtain
\beqan
\sum_{j=1}^n \E{ \var \left( 
                                          \widetilde{g}_{\widetilde{X}_1^{j-1}}
                                          (\widetilde{X}_j) 
                                   \mid \widetilde{X}_1^{j-1} 
                                \right) 
                         }
& \geq &      \sum_{j=1}^n \E{ \var \left( a_j(\widetilde{X}_j) 
                                   \mid \widetilde{X}_1^{j-1} 
                                \right)  } \\
& = & \sum_{j=1}^n \frac{a_j^2 \Delta_j^2}{12} \mbox{,}
\eeqan                                
which proves the second half of the lemma.

\end{IEEEproof}

Armed with this lemma, we may now determine upper and lower bounds
to the distortion of $g(x_1^n)$ within a single quantizer cell.

\begin{lemma}
\label{lem:boundDistOneCell}
Suppose that over a rectangular cell $S \subset [0,1]^n$ the function $g(x_1^n)$ 
is Lipschitz continuous and the probability density $f(x_1^n)$ is continuous, and suppose $g(x_1^n)$
has bounded first and second derivatives almost-everywhere in $S$.  Let $A_S$ denote the subset of $S$ where the first and
second derivatives of $g(x_1^n)$ are defined.
Then, defining $a_j = \inf_{x_1^n \in S} \mathbf{1}_{A_S}(x_1^n) |g_j(x_1^n)|$
and $b_j = \sup_{x_1^n \in S} |g_j(x_1^n)|$,
\[
f(\chi \mid X_1^n \in S) \sum_{j=1}^n \frac{a_j^2 \Delta_j^2}{12} \prod_{i = 1}^n \Delta_i
\leq 
\var \left( \left. g(X_1^n) \right| X_1^n \in S \right)
\leq
f(\xi \mid X_1^n \in S) \sum_{j=1}^n \frac{b_j^2 \Delta_j^2}{12} \prod_{i=1}^n \Delta_i
\]
for some $\chi,\,\xi \in S$.
\end{lemma}

\begin{IEEEproof}
We first prove the lower bound.  If $A_S$ is nonempty, 
$a_j = 0$ for every $j$ so the lower bound is trivially true.  Now
suppose $A_S$ is empty and therefore that the first
and second derivatives of $g(x_1^n)$ are defined everywhere
in $S$. In this case,
\beqan
\var \left( g(X_1^n) \mid X_1^n \in S \right) & = & 
\int_S f(x_1^n \mid X_1^n \in S)  \left( g(X_1^n) - \E{g(X_1^n) \mid X_1^n \in S}  \right)^2 \, dx_1^n \\
& \eqlabel{a} & f(\chi \mid X_1^n \in S) \int_S  \left( g(X_1^n) - \E{g(X_1^n) \mid X_1^n \in S}  \right)^2 \, dx_1^n \\
& = & f(\chi \mid X_1^n \in S) \left({\ts \prod_{i=1}^n \Delta_i} \right) \int_S  \frac{1}{\left(\prod_{i=1}^n \Delta_i \right) } \left( g(X_1^n) - \E{g(X_1^n) \mid X_1^n \in S}  \right)^2 \, dx_1^n \\
& \eqlabel{b} & 
f(\chi \mid X_1^n \in S) \left({\ts \prod_{i=1}^n \Delta_i} \right)
\E{ \left( g(\widetilde{X}_1^n) - \E{g(X_1^n) \mid X_1^n \in S} \right)^2} \\
& \geqlabel{c} & 
f(\chi \mid X_1^n \in S) \left({\ts \prod_{i=1}^n \Delta_i} \right) \var \left( g(\widetilde{X}_1^n) \right) \\
& \geqlabel{d} &
f(\chi \mid X_1^n \in S) \left({\ts \prod_{i=1}^n \Delta_i} \right) \sum_{j=1}^n \frac{a_j^2 \Delta_j^2}{12} \mbox{,}
\eeqan
where (a) follows from the first mean value theorem for integration;  
(b) introduces the random vector $\widetilde{X}_1^n$ that is uniform over $S$;
(c) is true because the variance is the smallest possible mean squared error from
a constant estimate; and (d) follows from the lower bound in
Lemma~\ref{lem:uniformDistLemma}.
%This proves the lower bound.

For the upper bound, we proceed in a similar manner:
\beqan
\var \left( g(X_1^n) \mid X_1^n \in S \right) & = & 
\int_S f(x_1^n \mid X_1^n \in S)  \left( g(X_1^n) - \E{g(X_1^n) \mid X_1^n \in S}  \right)^2 \, dx_1^n \\
& \leqlabel{a} & \int_S f(x_1^n \mid X_1^n \in S)  \left( g(X_1^n) - \overline{g}_S  \right)^2 \, dx_1^n \\
& \eqlabel{b} & f(\xi \mid X_1^n \in S) \int_S \left( g(X_1^n) - \overline{g}_S \right)^2 \, dx_1^n \\
& = & f(\xi \mid X_1^n \in S)  \left({\ts \prod_{i=1}^n \Delta_i} \right) \int_S \frac{1}{\prod_{i=1}^n \Delta_i}
\left( g(X_1^n) - \overline{g}_S \right)^2 \, dx_1^n \\
& \leqlabel{c} & f(\xi \mid X_1^n \in S)  \left({\ts \prod_{i=1}^n \Delta_i} \right) \sum_{j=1}^n 
\frac{b_j^2 \Delta_j^2}{12} \mbox{,}
\eeqan
where in (a) we reintroduce the notation $\overline{g} = \int_S g(X_1^n) \frac{1}{ \prod_{i=1}^n \Delta_i}$
for the average value of $g$ with respect to a uniform distribution,
and the inequality
is valid because the expected value of a random variable minimizes the mean-squared error
of the estimate;
%and $\overline{g}$ is only the expected value if $f(X_1^n \mid X_1^n \in S)$ is uniform; and
(b) is due to the first mean value theorem for integration; and (c) follows from the upper bound in
Lemma~\ref{lem:uniformDistLemma}.

\end{IEEEproof}

%\begin{definition}
%Suppose $g(x_1^n)$ is a continuous real-valued function,
%and suppose its partial derivatives $(g_1(x_1^n), \ldots, g_n(x_1^n))$ are defined almost everywhere.
%Then the \emph{patched} partial derivative is defined as
%\[
%\underline{g_j}(x_1^n) = \left\{ \begin{array}{l}
%  									g_j(x_1^n) \text{ if } g_j(x_1^n) \text{ is defined.} \\
%									0		 \text{ otherwise.}
%					        \end{array}
%				 	\right.
%\]					
%\end{definition}

%\begin{lemma}
%\label{lem:boundedSlopeMulti}
%Let $g(x_1^n)$ be a continuous function with gradient defined almost everywhere over a rectangular cell $S \in \R^n$ with
%edge lengths $\{\Delta_1, \ldots, \Delta_n\}$, and let $\overline{g}$ denote the average value of $g$ within $S$:
%\[ \overline{g} = \frac{1}{\prod_{j=1}^n\Delta_j} \int_S g(x_1^n) dx_1^n \mbox{.} \]
% Suppose the patched partial derivatives of $g(x_1^n)$ are of bounded magnitude: $0 \leq a_j \leq |\underline{g_j}(x_1^n)| \leq b_j$.  Then
% \[
% \sum_{j=1}^n \frac{a_j^2 \Delta_j^2}{12} \leq 
% \left( \prod_{j=1}^n \frac{1}{\Delta_j} \right) \int_S |g(x_1^n) - \overline{g}|^2 dx_1^n \leq
% \sum_{j=1}^n \frac{b_j^2 \Delta_j^2}{12} \mbox{.}
% \]
%Alternatively, letting $\widetilde{X}_1^n$ be uniformly distributed over $S$,
%\[
% \sum_{j=1}^n \frac{a_j^2 \Delta_j^2}{12} \leq 
% \var(g(\widetilde{X}_1^n)) \leq
% \sum_{j=1}^n \frac{b_j^2 \Delta_j^2}{12} \mbox{.}
% \]
%\end{lemma}

At this point, we provide a proof of the theorem.
\begin{IEEEproof}
The distortion $d(\KB;\lB)$ is given by
\[
d_g = \sum_{i_1^n} \P{X_1^n \in S_{i_1^n}} \var \left( g(X_1^n) \mid X_1^n \in S_{i_1^n} \right)\mbox{.}
\]
Let $A \in [0,1]^n$ denote the set of points $x_1^n$ where both 
the first and second derivatives of $g(x_1^n)$ are defined.  By 
assumption MF1, $[0,1]^n \setminus A$
has both Jordan and Lebesgue measure zero.
Defining $a_{i_1^n,j} = \inf_{x_1^n \in S_{i_1^n}} \mathbf{1}_A(x_1^n)|g_j(x_1^n)|$ and 
$b_{i_1^n,j} = \sup_{x_1^n \in S_{i_1^n}} |g_j(x_1^n)|$,
% within cell $S_{i_1^n}$,
we may obtain lower and upper bounds to $d_g$ by 
applying Lemma~\ref{lem:boundDistOneCell}
to each term within the summation:
\beqan
\sum_{i_1^n} \left( \prod_{j=1}^n \Delta_{i_1^n,j} \right) 
\P{X_1^n \in S_{i_1^n}} f(\chi_{i_1^n} \mid X_1^n \in S_{i_1^n}) \sum_{j = 1}^n 
\frac{a_{i_1^n,j}^2 \Delta_{i_1^n,j}^2}{12}
\leq d_g  \\
\leq
\sum_{i_1^n} \left( \prod_{j=1}^n \Delta_{i_1^n,j} \right) 
\P{X_1^n \in S_{i_1^n}} f(\xi_{i_1^n} \mid X_1^n \in S_{i_1^n}) \sum_{j = 1}^n 
\frac{b_{i_1^n,j}^2 \Delta_{i_1^n,j}^2}{12}  \mbox{.}
\eeqan
Let $K_j$ be the number of cells in the quantizer for $X_j$. 
For any cell $S$ of this quantizer, $\int_S \lambda_j(x_j) dx_j = 1/K_j$.
By continuity of $\lambda_j$ and the first mean value theorem, this implies that
the length of interval $S$ is given by 
$(K_j \lambda_j(\eta))^{-1}$ for some $\eta \in S$.  Therefore, 
$\Delta_{i_1^n,j}$ in the above expression may be replaced by $(K_j \lambda_j(\eta_{i_1^n}))$
for some $\eta_{i_1^n} \in S_{i_1^n}$:
\beqan
\sum_{i_1^n} \left( \prod_{j=1}^n \Delta_{i_1^n,j} \right) 
\P{X_1^n \in S_{i_1^n}} f(\chi_{i_1^n} \mid X_1^n \in S_{i_1^n}) \sum_{j = 1}^n 
\frac{a_{i_1^n,j}^2}{12K_j^2 \lambda_j(\eta_{i_1^n})^2}
\leq d_g  \\
\leq
\sum_{i_1^n} \left( \prod_{j=1}^n \Delta_{i_1^n,j} \right) 
\P{X_1^n \in S_{i_1^n}} f(\xi_{i_1^n} \mid X_1^n \in S_{i_1^n}) \sum_{j = 1}^n 
\frac{b_{i_1^n,j}^2 }{12K_j^2 \lambda_j(\eta_{i_1^n})^2}  \mbox{.}
\eeqan
Furthermore, we may recognize that $\P{X_1^n \in S_{i_1^n}} f(\xi_{i_1^n} \mid X_1^n \in S_{i_1^n}) = f(\xi_{i_1^n})$, simplifying the bounds further:
\[
\sum_{i_1^n} 
f(\chi_{i_1^n}) \sum_{j = 1}^n 
\frac{a_{i_1^n,j}^2}{12K_j^2 \lambda_j(\eta_{i_1^n})^2}\prod_{j=1}^n \Delta_{i_1^n,j}
\leq d_g  \leq
\sum_{i_1^n} 
f(\xi_{i_1^n}) \sum_{j = 1}^n 
\frac{b_{i_1^n,j}^2 }{12K_j^2 \lambda_j(\eta_{i_1^n})^2}\prod_{j=1}^n \Delta_{i_1^n,j}  \mbox{.}
\]

Consider the $j$th term in the lower-bound summation,
\[ \sum_{i_1^n} f(\chi_{i_1^n}) \frac{a_{i_1^n,j}^2}{12\lambda_j(\eta_{i_1^n})} \prod_{i=1}^n \Delta_{i_1^n,j}\mbox{.} \]
One may observe that this expression approaches a Riemann integral:
\begin{enumerate}
\item By Lemma \ref{lem:boundDistOneCell}, $\chi_{i_1^n} \in S_{i_1^n}$.
\item By definition, $a_{i_1^n,j}$ is the minimal value of $\mathbf{1}_A(x_1^n)|g_j(x_1^n)|$
within the cell $S_{i_1^n}$.
\item By its definition, $\eta_{i_1^n}$ is also an element in $S_{i_1^n}$.  
\item The product $\prod_{j=1}^n \Delta_{i_1^n,j}$ is the size of the cell, and
because the largest quantizer cell size goes to zero as every element of the vector $\KB$ grows, 
the mesh of this summation also goes to zero.
\end{enumerate}
Since $A$ is Jordan-measureable, $\mathbf{1}_A(x_1^n)$ is Riemann integrable.
By assumption MF2, the expression $f(x_1^n)  g_j(x_1^n)^2/\left(12\lambda_j(x_j)^2\right)$
is Riemann integrable.  Since the product of two integrable functions
is integrable, the Riemann integral of $\mathbf{1}_A(x_1^n) f(x_1^n) g_j(x_1^n)^2/\left(12\lambda_j(x_j)^2\right)$
is defined and
\beqan
\sum_{i_1^n} f(\chi_{i_1^n}) \frac{a_{i_1^n,j}^2}{12\lambda_j(\eta_{i_1^n})} \prod_{i=1}^n \Delta_{i_1^n,j} & \sim &
\int_{[0,1]^n} f(x_1^n) \mathbf{1}_A(x_1^n) \frac{g_j(x_1^n)^2}{12\lambda_j(x_j)^2} \, dx_1^n \\
& = & \int_{A} f(x_1^n) \frac{g_j(x_1^n)^2}{12\lambda_j(x_j)^2} \, dx_1^n \\
& \eqlabel{a} & \int_{[0,1]^n} f(x_1^n) \frac{g_j(x_1^n)^2}{12\lambda_j(x_j)^2} \, dx_1^n \\
& \eqlabel{b} & \int_{[0,1]} f(x_j) \frac{\gamma_j(x_j)^2}{12\lambda_j(x_j)^2} \, dx_j
\mbox{,}
\eeqan
where (a) follows from the Jordan measure of $[0,1]^n \setminus A$ being zero;
and (b) is the result of integrating over $x_1^{j-1}$ and $x_{j+1}^n$.
%Because $K_j = \lfloor K^{\alpha_j} \rfloor$, $1/K_j^2 \sim K^{-2\alpha_j}$ and 
This relation then yields
\[
\sum_{i_1^n} f(\chi_{i_1^s}) \frac{a_{i_1^n,j}^2}{12K_j^2 \lambda_j(\eta_{i_1^n})} \prod_{i=1}^n \Delta_{i_1^n,j} \sim 
\int_{[0,1]} f(x_j) \frac{\gamma_j(x_j)^2}{12K_j^2 \lambda_j(x_j)^2} \, dx_j 
= \frac{1}{12 K_j^2}\E{\left( \frac{\gamma_j(X_j)}{\lambda_j(X_j)} \right)^2}\mbox{.}
\]
Since this holds for any $j \in \{1,\ldots,n\}$, it holds for the sum over $j$ as well:
\[
\sum_{i_1^n} f(\chi_{i_1^s}) \sum_{j=1}^n \frac{a_{i_1^n,j}^2}{12K_j^2 \lambda_j(\eta_{i_1^n})} \prod_{i=1}^n \Delta_{i_1^n,j} \sim
\sum_{j=1}^n \frac{1}{12 K_j^2}\E{\left( \frac{\gamma_j(X_j)}{\lambda_j(X_j)} \right)^2}\mbox{.}
\]

Similarly,
\[
\sum_{i_1^n} f(\xi_{i_1^s}) \sum_{j=1}^n \frac{b_{i_1^n,j}^2}{12K_j^2 \lambda_j(\eta_{i_1^n})} \prod_{i=1}^n \Delta_{i_1^n,j} \sim
\sum_{j=1}^n \frac{1}{12 K_j^2}\E{\left( \frac{\gamma_j(X_j)}{\lambda_j(X_j)} \right)^2}\mbox{.}
\]
Since $d_g$ is bounded between these two quantities, this proves the theorem.
\end{IEEEproof}

\section{Proof of Lemma~\ref{lem:DistortionRateUnoptimized}}
\label{app:DistortionRateUnoptimized}
We start by defining the distortion-resolution optimization function 
$\KB_{\rm fr}(R;\lB)$ as the resolution vector that minimizes 
$d_{\rm fr}(\KB_{\rm fr}(R;\lB);\lB)$, the distortion subject to a fixed rate constraint. 
 We define $\KB_{\rm vr}(R;\lB)$ and $\KB_{\rm sw}(R;\lB)$ analogously, and we write
 $\KB_{\rm fr,vr,sw}$ when we can combine all three cases to be handled identically.
We similarly define the high-resolution distortion-resolution optimizing function
$\KB_{\rm fr,vr,sw}^{\rm HR}(R;\lB)$ as the resolution vector that minimizes
$d_{\rm fr,vr,sw}^{\rm HR}(\KB_{\rm fr,vr,sw}^{\rm HR}(R;\lB); \lB)$ under a rate constraint.  
Note that by definition $d_{\rm fr,vr,sw}(\KB_{\rm fr,vr,sw}(R;\lB); \lB) = D_{\rm fr,vr,sw}(R;\lB)$.

\begin{lemma}
\label{lem:KDiverges}
Under assumptions MF1--4, every component of the vectors $\KB_{\rm fr,vr,sw}(R;\lB)$
and $\KB_{\rm fr,vr,sw}^{\rm HR}(R;\lB)$ diverges with increasing $R$.
\end{lemma}
\begin{IEEEproof}
For a function, source, and quantizer point density that together satisfy conditions MF1--4,
we demonstrate that every component of both $\KB_{\rm fr,vr,sw}(R;\lB)$ and
$\KB_{\rm fr,vr,sw}^{\rm HR}(R;\lB)$ diverges.  Suppose first that the $j$th element
of $\KB_{\rm fr,vr,sw}(R;\lB)$ is bounded by a finite value $K$ for any $R$.  
Then the quantizer $Q_K^{\lambda_j}$ is a sufficient description of $X_j$ for
achieving arbitrarily small distortion for the function $g(X_1^n)$.
More precisely, there exists
a reconstruction function $\widehat{g}$ such that 
$g(X^n) = \widehat{g}(X_1^{j-1},Q_K^{\lambda_j}(X_j),X_{j+1}^n)$ with probability one.
This then implies that $\gamma_j(X_j)$ is zero with probability one, but this violates
condition MF4 and thus every component of $\KB_{\rm fr,vr,sw}(R;\lB)$
diverges with $R$.

If the $j$th component of $\KB_{\rm fr,vr,sw}^{\rm HR}$ has a finite upper bound $K$,
then the high-resolution distortion is lower bounded by 
\[
d_{\rm fr,vr,sw}^{\rm HR}(\KB;\lB) \geq \frac{1}{12K^2} \E{\left( \frac{\gamma_j(X_j)}{\lambda_j(X_j)} \right)^2}
\mbox{.}
\]
By condition MF4, this lower bound is strictly positive, and therefore this choice of
$\KB_{\rm fr,vr,sw}^{\rm HR}$ is suboptimal.
\end{IEEEproof}

Using this lemma, we are able to connect the distortion-rate function to the
high-resolution distortion-resolution function.

\begin{lemma}
\label{lem:DistortionRateEqualsOptimizedDistortionResolution}
The distortion-rate function is asymptotically equal to the 
optimized high-resolution distortion-resolution function:
$D_{\rm fr,vr,sw}(R;\lB) \sim d_{\rm fr,vr,sw}^{\rm HR}(\KB^{\rm HR}_{\rm fr,vr,sw}(R;\lB);\lB)$.
\end{lemma}
\begin{IEEEproof}
Since by Lemma \ref{lem:KDiverges}
both $\KB_{\rm fr,vr,sw}$ and $\KB_{\rm fr,vr,sw}^{\rm HR}$ diverge in every component,
Theorem \ref{thm:multi-distortion} tells us that
\[
d_{\rm fr,vr,sw}(\KB_{\rm fr,vr,sw}(R;\lB);\lB) \sim d_{\rm fr,vr,sw}^{\rm HR}(\KB_{\rm fr,vr,sw}(R;\lB);\lB)
\]
and
\[
d_{\rm fr,vr,sw}(\KB^{\rm HR}_{\rm fr,vr,sw}(R;\lB);\lB) \sim d_{\rm fr,vr,sw}^{\rm HR}(\KB^{\rm HR}_{\rm fr,vr,sw}(R;\lB);\lB)
\mbox{.}
\]
Furthermore, by definition we have that
\[
d_{\rm fr,vr,sw}(\KB_{\rm fr,vr,sw}(R;\lB);\lB) 
\leq 
d_{\rm fr,vr,sw}(\KB^{\rm HR}_{\rm fr,vr,sw}(R;\lB);\lB)
\]
and
\[
d_{\rm fr,vr,sw}^{\rm HR}(\KB^{\rm HR}_{\rm fr,vr,sw}(R;\lB);\lB) 
\leq
d_{\rm fr,vr,sw}^{\rm HR}(\KB_{\rm fr,vr,sw}(R;\lB);\lB)
\mbox{.}
\]
Therefore, 
\[
D_{\rm fr,vr,sw}(R;\lB) = d_{\rm fr,vr,sw}(\KB_{\rm fr,vr,sw}(R;\lB);\lB) \sim d_{\rm fr,vr,sw}^{\rm HR}(\KB_{\rm fr,vr,sw}^{\rm HR}(R;\lB);\lB) \mbox{.}
\]
\end{IEEEproof}

Before proceeding with the proof, we define three countably infinite subsets of $\R^n$
that describe the rate vectors achievable by a certain choice of point densities $\lB$:
\beqa
\mathcal{R}_{\rm fr} & = & \{ (\log K_1, \log K_2, \ldots, \log K_n): \KB \in \mathbb{N}^{n} \}  \mbox{,} \label{eq:FixedRatePoints}\\
\mathcal{R}_{\rm vr}& = & \{ (H(Q_{K_1}^{\lambda_1}(X_1)), H(Q_{K_2}^{\lambda_2}(X_2)), \ldots, H(Q_{K_n}^{\lambda_n}(X_n))): \KB \in \mathbb{N}^n\}
\mbox{,} \label{eq:VariableRatePoints} \\
\mathcal{R}_{\rm sw}& = & \{ (H(Q_{K_1}^{\lambda_1}(X_1)), H(Q_{K_2}^{\lambda_2}(X_2) \mid Q_{K_1}^{\lambda_1}(X_1)), \ldots, H(Q_{K_n}^{\lambda_n}(X_n) \mid 
Q_{K^{n-1}}^{\lambda^{n-1}}(X^{n-1}))): \KB \in \mathbb{N}^n \} \mbox{.} \label{eq:SlepianWolfRatePoints}
\eeqa

Using these definitions and Lemma \ref{lem:DistortionRateEqualsOptimizedDistortionResolution},
we may rephrase the distortion-rate functions somewhat:
\beqa
D_{\rm fr}(R;\lB) \sim d_{\rm fr}^{\rm HR}(\KB^{\rm HR}_{\rm fr}(R;\lB); \lB)
& =&  \min_{\RB \in \mathcal{R}_{\rm fr}: \sum R_j \leq R} 
\sum_{j=1}^n \frac{1}{12} 2^{-2R_j} \E{\left(\frac{\gamma_j(X_j)}{\lambda_j(X_j)}\right)^2} 
\label{eq:FRDistortionRateGranular}
\\
D_{\rm vr}(R;\lB) \sim d_{\rm vr}^{\rm HR}(\KB^{\rm HR}_{\rm vr}(R;\lB); \lB)
&=& \min_{\RB \in \mathcal{R}_{\rm vr}: \sum R_j \leq R} 
\sum_{j=1}^n \frac{1}{12} 2^{-2R_j + 2h(X_j) +2\E{\log \lambda_j(X_j)}} 
\E{\left(\frac{\gamma_j(X_j)}{\lambda_j(X_j)}\right)^2} 
\label{eq:VRDistortionRateGranular}
\\
D_{\rm sw}(R;\lB) \sim d_{\rm sw}^{\rm HR}(\KB^{\rm HR}_{\rm sw}(R;\lB); \lB)
&=& \min_{\RB \in \mathcal{R}_{\rm sw}: \sum R_j \leq R}
\sum_{j=1}^n \frac{1}{12} 2^{-2R_j + 2h(X_j \mid X^{j-1}) +2\E{\log \lambda_j(X_j)}} \E{\left(\frac{\gamma_j(X_j)}{\lambda_j(X_j)}\right)^2} 
\label{eq:SWDistortionRateGranular}
\eeqa

Additionally, we introduce the concept of \emph{increasing granularity}:
\begin{definition}
A countably infinite set $\mathcal{R} \subset \R^n$ is said to be \emph{increasingly granular} if
for any $r \geq 0$ there exists a vanishing nonnegative function $\delta(r): [0,\infty) \rightarrow [0, \infty)$ such that for any $\RB \in \R^n$ whose components are
each greater than $r$, there exists a point $\overline{\RB} \in \mathcal{R}$ within
$\delta(r)$ of each component of $\RB$: $\max_j |R_j - \overline{R}_j| \leq \delta(r)$.
The function $\delta(r)$ is called the \emph{granularity function} of the set $\mathcal{R}$.
\end{definition}

\begin{lemma}
\label{lem:RegionsAreGranular}
The sets $\mathcal{R}_{\rm fr}$, $\mathcal{R}_{\rm vr}$, and $\mathcal{R}_{\rm sw}$
are increasingly granular.
\end{lemma}
\begin{IEEEproof}
Let $r > 0$, and let every component of $\RB \in \R^n$ be greater than $r$.
We prove the granularity of each of the three sets in turn.

$\mathcal{R}_{\rm fr}$: Define the point $\overline{\RB} \in \R^n$ so that
$\overline{R}_j = \log \lfloor 2^{R_j} \rfloor$.  This point is clearly a member of
$\mathcal{R}_{\rm fr}$.  Furthermore, we can easily bound the distance between 
$R_j$ and $\overline{R}_j$:
\[ |R_j - \overline{R}_j| \leq \log \frac{2^{R_j}}{2^{R_j}-1} \leq \log \frac{2^r}{2^r-1}\mbox{.} \]
Defining $\delta(r) = \log \left( \frac{2^r}{2^r-1} \right) \rightarrow 0$, we have shown that $\mathcal{R}_{\rm fr}$ is
increasingly granular.

$\mathcal{R}_{\rm vr}$: Define $\overline{\RB}$ so that
$\overline{R}_j = H(Q_{K_j}^{\lambda_j}(X_j))$ where 
$K_j$ is chosen according to 
\[
K_j = \argmin_{K} \left| R_j - h(X_j) - \E{\log \lambda_j} - \log K \right| \mbox{.}
\]
We may then bound the distance between $R_j$ and $\overline{R_j}$:
\beqan
|R_j - \overline{R}_j |& \leq & \left| \overline{R}_j - h(X_j) - \E{\log \lambda_j} - \log K_j \right|
+ \left| h(X_j) + \E{\log \lambda_j} + \log K_j  - R_j\right| \\
& \leq & \left| \overline{R}_j - h(X_j) - \E{\log \lambda_j} - \log K_j \right| 
+ \log \frac{K_j}{K_j-1} \\
& \rightarrow & 0 \mbox{,}
\eeqan
where the first term goes to zero by Lemma \ref{lem:ResRateErrorSW} and the second by Lemma \ref{lem:KDiverges}.
%\beqan
%|R_j - \overline{R}_j| & \leq & |H(Q_{K_j+1}^{\lambda_j}(X_j)) - H(Q_{K_j}^{\lambda_j}(X_j))| \\
%& \leq & \left| H(Q_{K_j+1}^\lambda(X_j)) - h(X_j) - \E{\log \lambda(X_j)} - \log (K_j+1) \right| \\
%& &
%+ \left| H(Q_{K_j}^\lambda(X_j)) - h(X_j) - \E{\log \lambda(X_j)} - \log {K_j}\right| \\
%& & + \left| \log(K_j+1) - \log K_j \right| \\
%& \rightarrow & 0 \mbox{,}
%\eeqan
%where the first two terms go to zero according to Lemma \ref{lem:ResolutionRate}.

$\mathcal{R}_{\rm sw}$: Define $\overline{\RB}$ so that
$\overline{R}_j = H(Q_{K_j}^{\lambda_j}(X_j) | Q_{K^{j-1}}^{\lambda^{j-1}}(X^{j-1}))$ where 
$K_j$ is chosen according to 
\[
K_j = \argmin_{K} \left| R_j - h(X_j \mid X^{j-1}) - \E{\log \lambda_j} - \log K \right| \mbox{.}
\]

\noindent The distance between $R_j$ and $\overline{R}_j$ may then be bounded in the 
following manner:
\beqan
|R_j - \overline{R}_j |& \leq & \left| \overline{R}_j - h(X_j | X^{j-1}) - \E{\log \lambda_j} - \log K_j \right|
+ \left| h(X_j | X^{j-1}) + \E{\log \lambda_j} + \log K_j  - R_j\right| \\
& \leq & \left| \overline{R}_j - h(X_j | X^{j-1}) - \E{\log \lambda_j} - \log K_j \right| 
+ \log \frac{K_j}{K_j-1} \\
& \rightarrow & 0 \mbox{.}
\eeqan
To show that the first term goes to zero, we invoke Lemma \ref{lem:ResRateErrorSW}
to state that 
\[
H(Q^{\lambda_j}_{K_j}(X_j), Q^{\lambda^{j-1}}_{K^{j-1}}(X^{j-1})) - \sum_{i=1}^j \log K_i \rightarrow 
h(X^j) + \E{\sum_{i=1}^j \log \lambda_i(X_i)} \mbox{.}
\]
Subtracting from this the similar expression (also obtained from Lemma \ref{lem:ResRateErrorSW})
\[
H(Q^{\lambda^{j-1}}_{K^{j-1}}(X^{j-1})) - \sum_{i=1}^{j-1} \log K_i \rightarrow
h(X^{j-1}) + \E{\sum_{i=1}^{j-1} \log \lambda_i(X_i)} \mbox{,}
\]
yields that
\[
\overline{R}_j - h(X_j \mid X^{j-1}) - \E{\log \lambda_j} - \log K_j = 
H(Q_{K_j}^{\lambda_j}(X_j) | Q_{K^{j-1}}^{\lambda^{j-1}}(X^{j-1})) - h(X_j \mid X^{j-1}) - \E{\log \lambda_j} - \log K_j
\rightarrow 0 \mbox{.}
\]

\end{IEEEproof}

We now establish an important property of increasingly granular sets.
\begin{lemma}
\label{lem:Granularity}
Suppose 
\[ f(R) = \min_{\RB \in \mathcal{R}: \sum_{j=1}^n R_j \leq R} \sum_{j=1}^n \alpha_j 2^{-R_j} \mbox{,}
\]
and
\[
\widetilde{f}(R)  = \min_{\RB \in \R^n: \sum_{j=1}^n R_j \leq R} \sum_{j=1}^n \alpha_j2^{-R_j} \mbox{,}
\]
where $\alpha_j > 0$ for all $j$ and $\mathcal{R}$ is an increasingly granular subset of
$\R^n$.  Then $f(R) \sim \widetilde{f}(R)$.
\end{lemma}
\begin{IEEEproof}
Since $\mathcal{R} \subset \R^n$, we have that 
\[ 
\widetilde{f}(R) \leq f(R) \mbox{.}
\]

Let $\RB^* = \argmin_{\RB \in \R^n: \sum_{j=1}^n R_j \leq R} \sum_{j=1}^n \alpha_j2^{-R_j}$,
let $R^*_{\inf}$ indicate the smallest element of $\RB^*$, and let $\delta(r)$ be the granularity
function of $\mathcal{R}$.  Since there must be an element of $\mathcal{R}$ within distance
$\delta_{R^*_{\inf}}$ of each of the coordinates of $\RB^*$, we may create a bound in the opposite direction:
\[
f(R) \leq \sum_{j=1}^n \alpha_j 2^{-(R^*_j - \delta_{R^*_{\inf}})} = 2^{\delta_{R^*_{\inf}}} \widetilde{f}(R) \mbox{.}
\]
Because $\alpha_j > 0$ for $j \in \{1,\ldots, n\}$, $R^*_{\inf}$ diverges with $R$ and $\delta_{R^*_{\inf}}$ vanishes.
Combining the two bounds, we have that $\widetilde{f}(R) \sim f(R)$, which proves the lemma.
\end{IEEEproof}

Applying Lemmas \ref{lem:RegionsAreGranular} and \ref{lem:Granularity} to 
\eqref{eq:FRDistortionRateGranular}, \eqref{eq:VRDistortionRateGranular}, and
\eqref{eq:SWDistortionRateGranular}, we may widen the optimization to occur
over any positive real-valued rate vector $\RB$:
\beqa
D_{\rm fr}(R;\lB) \sim d_{\rm fr}^{\rm HR}(\KB^{\rm HR}_{\rm fr}(R;\lB); \lB)
& \sim&  \min_{\RB \in \R^n: \sum R_j \leq R} 
\sum_{j=1}^n \frac{1}{12} 2^{-2R_j} \E{\left(\frac{\gamma_j(X_j)}{\lambda_j(X_j)}\right)^2} 
\label{eq:FRDistortionRateContinuous}
\\
D_{\rm vr}(R;\lB) \sim d_{\rm vr}^{\rm HR}(\KB^{\rm HR}_{\rm vr}(R;\lB); \lB)
&\sim& \min_{\RB \in \R^n: \sum R_j \leq R} 
\sum_{j=1}^n \frac{1}{12} 2^{-2R_j + 2h(X_j) +2\E{\log \lambda_j(X_j)}} 
\E{\left(\frac{\gamma_j(X_j)}{\lambda_j(X_j)}\right)^2} 
\label{eq:VRDistortionRateContinuous}
\\
D_{\rm sw}(R;\lB) \sim d_{\rm sw}^{\rm HR}(\KB^{\rm HR}_{\rm sw}(R;\lB); \lB)
&\sim& \min_{\RB \in \R^n: \sum R_j \leq R}
\sum_{j=1}^n \frac{1}{12} 2^{-2R_j + 2h(X_j \mid X^{j-1}) +2\E{\log \lambda_j(X_j)}} \E{\left(\frac{\gamma_j(X_j)}{\lambda_j(X_j)}\right)^2} 
\label{eq:SWDistortionRateContinuous}
\eeqa

\noindent The proof is completed by a straightforward application of Lemma \ref{lem:bit-alloc} to 
optimize the rate allocation in each of these expressions.

\section{Proof of Theorem~\ref{thm:EquivalenceFree}}
\label{app:equivalencefree}
The theorem asserts that when the function is equivalence-free,
$w_j$ failing to be one-to-one on the support of $X_j$ creates
a component of the distortion that cannot be eliminated by
quantizing more finely.
The proof here lower-bounds the distortion
by focusing on the contribution from just the $j$th variable.
The bound is especially crude because it is based on observing
$\{X_i\}_{i \neq j}$ and $w_j(X_j)$ without quantization and
it uses only the contribution from $X_j \in S \cup t(S)$.

We wish to first bound the functional distortion in terms of a contribution
from the $j$th variable:
\beqa
  d_g(K;\wB) & \geqlabel{a} & \E{ \var(g(X_1^n) \mid \Yhat_1^n ) }
    \nonumber \\
    & \geqlabel{b} & \E{ \var(g(X_1^n) \mid \Yhat_j,\, \{X_i\}_{i\neq j} ) }
    \nonumber \\
    & \geqlabel{c} & \E{ \var(g(X_1^n) \mid w_j(X_j),\, \{X_i\}_{i\neq j} ) }
    \nonumber \\
    & \eqlabel{d}  & \E{ \var(g(X_1^n) \mid w_j(X_j),\, \{X_i\}_{i\neq j} )
                         \mid A } \P{A} \nonumber \\
    &       & {}   + \E{ \var(g(X_1^n) \mid w_j(X_j),\, \{X_i\}_{i\neq j} )
                         \mid A^c } \P{A^c} \nonumber \\
    & \geqlabel{e}  & \E{ \var(g(X_1^n) \mid w_j(X_j),\, \{X_i\}_{i\neq j} )
                         \mid A } \P{A} \nonumber \\
    & =             & \E{ \var(g(X_1^n) \mid X_j \in w_j^{-1}(X_j),\, \{X_i\}_{i\neq j} )
                         \mid A } \P{A} \nonumber \\
    & \eqlabel{f} 	 &  \int_{x\in S \cup t(S)} \E{\var(g(X_1^n) \mid X_j \in w_j^{-1}(x),\, \{X_i\}_{i\neq j} )} dx
                      \mbox{,}
  \label{eq:equiv-free-proof-1}
\eeqa
where $A$ is the event $X_j \in S \cup t(S)$.
Step (a) will hold with equality when the optimal estimate
(the conditional expectation of $g(X_1^n)$ given the quantized values)
is used;
(b) holds because, for each $i \neq j$, $\Yhat_i$ is a function of $X_i$;
(c) holds because $\Yhat_j$ is a function of $w_j(X_j)$;
(d) is an application of the law of total expectation;
(e) holds because the discarded term is nonnegative; and
(f) converts the expectation over $A$ into integral form.
It remains to use the hypotheses of the theorem to bound
the conditional variance in the final expression.

Since the function is equivalence free, for every set $B \subset [0,1]$ of cardinality
greater than one,
$$
  \E{ \var\left( g(X_1^n) \mid X_j \in B ,\, \{X_i\}_{i \neq j} \right)}
     > 0 \mbox{.}
$$
Since $w_j(s) = w_j(t(s))$ for any $s\in S$, the set $w_j^{-1}(x_j)$ is of
cardinality greater than one for any $x_j$ in $S\cup t(S)$.  Therefore for
any $x\in S\cup t(S)$,
$$
	\E{ \var\left( g(X_1^n) \mid X_j \in w_j^{-1}(x),\, \{X_i\}_{i \neq j} \right)}
     > 0 \mbox{,}
$$
and \eqref{eq:equiv-free-proof-1} is therefore greater than zero and
independent of rate.

\section{Proof of Theorem~\ref{thm:variableratedifferent}}
\label{app:dontCareVar}
It is already shown in Theorem~\ref{thm:fixedratesame} that 
the distortion-resolution expression \eqref{eq:dontCareDist} holds
when a codeword is allocated to each of the don't-care intervals.
After an appropriate rate analysis, we will optimize the point densities
outside of the don't-care intervals.

The key technical problem is that the rate analysis
\eqref{eq:1drate} does not hold when there are intervals where
$f_X$ is positive but $\lambda$ is not.
This is easily remedied by only applying \eqref{eq:1drate} conditioned
on $A_j$:
\beq
  \lim_{K_j \rightarrow \infty} \left[ H(\Xhat_j \mid A_j) -  \log(K_j-M_j) \right]
  = h(X_j \mid A_j)
            + \E{\log\lambda_j(X_j) \mid A_j} \mbox{.}
 \label{eq:dontcareRate}
\eeq
Note that this approximation can be shown to be asymptotically valid
in the same manner as in Lemmas \ref{lem:ResolutionRate} and \ref{lem:ResRateError}.
Now conditioned on $A_j$,
the dependence of distortion and rate on $\lambda_j$ is precisely in
the standard form of Section~\ref{sec:Multi}.
Thus, following Theorem~\ref{thm:multi-summary},
the optimal point density outside of $Z_j$ is given by
\eqref{eq:multi-var-lambda}.

Since the previous results now give the distortion in terms of the
conditional entropies $H(\Xhat_j \mid A_j)$, what remains is to
relate these to the rates:
\beqan
R_j & = & H(\Xhat_j) \\
    & \eqlabel{a} & H(\Xhat_j, I_j) \\
    & = & H(I_j) + H(\Xhat_j \mid I_j) \\
    & \eqlabel{b} & H(I_j) + \P{A_j} H(\Xhat_j \mid A_j)\mbox{,}
\eeqan
where (a) uses that $I_j$ is a deterministic function of $\Xhat_j$;
and (b) uses that specifying any $I_j \neq 0$ determines $\Xhat_j$ uniquely.
In anticipation of evaluating \eqref{eq:dontCareDist}, we define
the high-rate resolution-rate function as before:
\beqan
  \log(\Khr_j(R_j;\lambda_j)-M_j)
    & \sim & (\P{A_j})^{-1}\left(R_j - H(I_j)\right) \\
    & & {} - h(X_j \mid A_j) - \E{\log(\lambda_j(X_j) \mid A_j}\mbox{.}
\eeqan
Asymptotic accuracy of this approximation follows from \eqref{eq:dontcareRate}.
As before, one may insert this into the high-resolution distortion-resolution
expression \eqref{eq:dontCareDist} and bound the effect of the approximation as a multiplying factor
that goes to one.
Now evaluating \eqref{eq:dontCareDist} with optimal point densities
\eqref{eq:multi-var-lambda} gives \eqref{eq:dontCareDistVar}.

\section{A Quasi-Triangle Inequality}
\label{app:triangle}
\begin{lemma}
The $\Lonethird$ ``norm'' is a quasinorm with constant 4.
Equivalently, letting $x$ and $y$ be functions $\R \rightarrow \R^+$
with finite $\Lonethird$ quasinorms,	
\[ 
\|x + y\|_{1/3} \leq 4\left( \|x\|_{1/3} + \|y\|_{1/3} \right) \mbox{.}
\]
\end{lemma}
\begin{IEEEproof}
First, we prove the relation $4(a^3+b^3) \geq (a+b)^3$
for positive real numbers $a$ and $b$:
\beqan
\lefteqn{4(a^3 + b^3) - (a+b)^3} \\
 & = & 4a^3 + 4b^3 - a^3 - b^3 - 3a^2b - 3ab^2 \\
 & = & 3(a+b)(a-b)^2
 \ \geq \ 0 \mbox{.}
\eeqan
Now by this relation,
with $a = \int x(t)^{1/3} \, dt$
and $b = \int y(t)^{1/3} \, dt$:
\begin{eqnarray*}
\| x \|_{1/3} + \| y \|_{1/3} & = & 
\left(\int x(t)^{1/3} \, dt\right)^3 + \left(\int y(t)^{1/3} \, dt\right)^3 \\
  & \geq & \frac{1}{4}\left( \int \left(x(t)^{1/3} + y(t)^{1/3}\right) \, dt \right)^3 \\
  & \geq & \frac{1}{4}\left( \int \left((x(t)+y(t))^{1/3} \right) \, dt \right)^3 \\
  & = & \frac{1}{4}\|x+y\|_{1/3},
\end{eqnarray*}
where the second inequality uses, pointwise over $t$,
the concavity of the cube-root function on $[0,\infty)$.
\end{IEEEproof}

\section*{Acknowledgments}
The authors would like to thank John Sun for his insights and suggestions about
both the problem and the manuscript.  They would also like to thank both the reviewers
and the Associate Editor, Erik Ordentlich,
for helping to significantly improve the clarity and rigor of the results presented here.
\bibliographystyle{IEEEtran} 
%\bibliography{abrv,fq}

\begin{IEEEbiography}{Vinith Misra}
received the S.B. and M.Eng. degrees in electrical engineering from the Massachusetts Institute of Technology
in 2008, where his thesis was awarded the David Adler Memorial M.Eng. thesis prize.  He is 
currently pursuing a Ph.D. in Stanford University's department of electrical engineering.

He is a Stanford Graduate Fellow and a recipient of the National Defense Science and Engineering
Graduate Fellowship.  His research interests include information theory, signal processing, and
mixed-signal circuit design, with applications to both communications and medical devices.
\end{IEEEbiography}

\begin{biography}{Vivek K Goyal} (S'92-M'98-SM'03)
received the B.S. degree in mathematics and the B.S.E. degree in
electrical engineering from the University of Iowa, Iowa City,
where he received the John Briggs Memorial Award for the top
undergraduate across all colleges.  He received the M.S. and
Ph.D. degrees in electrical engineering from the University of
California, Berkeley, where he received the Eliahu Jury Award for
outstanding achievement in systems, communications, control, or
signal processing.

He was a Member of Technical Staff in the Mathematics of Communications
Research Department of Bell Laboratories, Lucent Technologies,
1998--2001; and a Senior Research Engineer for Digital Fountain, Inc.,
2001--2003\@.  He is currently Esther and Harold E. Edgerton Associate
Professor of electrical engineering at the Massachusetts Institute of
Technology.  His research interests include source coding theory,
sampling, quantization, and computational imaging.

Dr.\ Goyal is a member of Phi Beta Kappa, Tau Beta Pi, Sigma Xi,
Eta Kappa Nu and SIAM\@.  He was awarded the 2002 IEEE Signal Processing
Society Magazine Award and an NSF CAREER Award.
He served on the IEEE Signal Processing SocietyÕs Image and Multiple
Dimensional Signal Processing Technical Committee.
He is a Technical Program Committee Co-chair of IEEE ICIP 2016 and
a permanent Conference Co-chair of the SPIE Wavelets conference series.
\end{biography}

\begin{IEEEbiography}{Lav R. Varshney}
received the B. S. degree with honors in electrical and computer
engineering (magna cum laude) from Cornell University, Ithaca, New York in 2004.
He received the S. M., E. E., and Ph. D. degrees in electrical engineering and
computer science from the Massachusetts Institute of Technology (MIT),
Cambridge in 2006, 2008, and 2010, respectively.

He is a research staff member at the IBM Thomas J. Watson Research Center,
Hawthorne, NY. He was a National Science Foundation graduate research fellow
and held various research and teaching positions at MIT.  His research
interests include information theory, coding, and neuroscience.

Dr. Varshney is a member of Tau Beta Pi, Eta Kappa Nu, and Sigma Xi. He received
the E. A. Guillemin Thesis Award for Outstanding Electrical Engineering S.M.
Thesis, the Capocelli Prize at the 2006 Data Compression Conference, the Best
Student Paper Award at the 2003 IEEE Radar Conference, and was a winner of the
IEEE 2004 Student History Paper Contest.
\end{IEEEbiography}
\end{document}